%% file: main.tex
\documentclass[12pt, letterpaper]{article}
\usepackage[margin=1.25in]{geometry}
\usepackage{multirow, multicol}
\usepackage{booktabs}
\usepackage{dsfont,amssymb,amsmath,amsthm,centernot,graphicx,accents}
\usepackage{enumitem}
\usepackage{natbib}
\usepackage[colorlinks = true, citecolor = blue]{hyperref}
\usepackage{cleveref}
\usepackage{setspace}
\usepackage{comment}
\usepackage{algorithm}
\usepackage{algpseudocode}
\usepackage{ragged2e}
\usepackage{caption}
\usepackage{subcaption}

\usepackage{pgfplots} 
\pgfplotsset{compat = newest} 
\usetikzlibrary{positioning, arrows.meta} 
\usepgfplotslibrary{fillbetween}

\newtheorem{theorem}{Theorem}
\newtheorem*{theorem*}{Theorem}
\Crefname{theorem}{Theorem}{Theorems}
\crefname{theorem}{theorem}{theorems}

\newtheorem{proposition}{Proposition}
\newtheorem*{proposition*}{Proposition}
\Crefname{proposition}{Proposition}{Propositions}
\crefname{proposition}{proposition}{propositions}

\newtheorem{corollary}{Corollary}
\newtheorem*{corollary*}{Corollary}
\Crefname{corollary}{Corollary}{Corollaries}
\crefname{corollary}{corollary}{corollaries}

\newtheorem*{claim*}{Claim}
\Crefname{claim}{Claim}{Claims}
\crefname{claim}{claim}{claims}

\newtheorem{lemma}{Lemma}
\newtheorem*{lemma*}{Lemma}
\Crefname{lemma}{Lemma}{Lemmas}
\crefname{lemma}{lemma}{lemmas}

\newtheorem*{conjecture*}{Conjecture}
\Crefname{conjecture}{Conjecture}{Conjectures}
\crefname{conjecture}{conjecture}{conjectures}

\theoremstyle{definition}

\newtheorem*{definition*}{Definition}
\Crefname{definition}{Definition}{Definitions}
\crefname{definition}{definition}{definitions}

\newtheorem{assumption}{Assumption}
\newtheorem*{assumption*}{Assumption}
\Crefname{assumption}{Assumption}{Assumptions}
\crefname{assumption}{assumption}{assumptions}

\newtheorem{remark}{Remark}
\newtheorem*{remark*}{Remark}
\Crefname{remark}{Remark}{Remarks}
\crefname{remark}{remark}{Remarks}

\newtheorem{example}{Example}
\newtheorem*{example*}{Example}
\Crefname{example}{Example}{Examples}
\crefname{example}{example}{Examples}

\newcommand{\Var}{\textup{Var}}
\newcommand{\diag}{\textup{diag}}

\newcommand{\corr}{\textup{corr}}

\allowdisplaybreaks

\title{Inference with Many Weak Instruments and Heterogeneity}

\date{\today }
\author{Luther Yap\thanks{Email: \texttt{lyap@princeton.edu}. Department of Economics, Princeton University. This paper benefited from discussions with Michal Koles{\'a}r, David Lee, Ulrich M\"{u}ller, and seminar participants at Duke University, Cornell University, National University of Singapore, Princeton University, University of Notre Dame, University of Pennsylvania, and University of Southern California. }}

\begin{document}


\maketitle


\begin{abstract}
This paper considers inference in a linear instrumental variable regression model with many potentially weak instruments, in the presence of heterogeneous treatment effects. 
I first show that existing test procedures, including those that are robust to either weak instruments or heterogeneous treatment effects, can be arbitrarily oversized.
I propose a novel and valid test based on a score statistic and a ``leave-three-out" variance estimator. 
In the presence of heterogeneity and within the class of tests that are functions of the leave-one-out analog of a maximal invariant, this test is asymptotically the uniformly most powerful unbiased test.
In two applications to judge and quarter-of-birth instruments, the proposed inference procedure also yields a bounded confidence set while some existing methods yield unbounded or empty confidence sets. 
\end{abstract}


\vskip 0.2in


\section{Introduction}

Many empirical studies in economics involve instrumental variable (IV) models with many instruments.
A prominent example is the examiner design: several studies argue that judges or case workers are as good as randomly assigned and can affect the treatment status, so they are used as instruments to study the effects of foster care \citep{doyle2007child}, incarceration \citep{kling2006incarceration}, detention \citep{dobbie2018effects}, disability benefits \citep{autor2019disability}, and misdemeanor prosecution \citep{agan2023misdemeanor}, among others. 
When the IV is a vector of indicators for judges, the number of instruments can be large relative to the sample size.
Another example of many IV is a single instrument interacted with discrete covariates. 
When \citet{angrist1991does} used the quarter of birth as an instrument to study the returns to education, interacting the quarter of birth with the state of birth generates 150 instruments.

Despite the pervasiveness and importance of this setting, there does not yet exist an inference procedure that is robust to both heterogeneous treatment effects and weak instruments, which is a gap this paper aims to fill.
Weak IV refers to a setting where the first-stage coefficients converge to zero at a rate such that no consistent estimator for the object of interest exists (following from the definition in \citet{mikusheva2022inference} rather than \citet{chao2012asymptotic}); and heterogeneous treatment effects refers to a setting where different subsets of the many IV may estimate different local average treatment effects (LATE).
There are several recent proposals \citep{crudu2021inference, mikusheva2022inference, matsushita2022jackknife} that are robust to weak IV, but they assume constant treatment effects. 
A separate literature \citep{evdokimov2018inference} proposed variance estimators for the Jackknife IV Estimator (JIVE) that are robust to heterogeneous treatment effects, but their $t$-statistic test is still not robust to many weak IV.
While it is clear that weak IV can lead to substantial distortions in inference (e.g., \citet{dufour1997some, StaigerStock97}), it is less obvious if procedures developed under constant treatment effects that are robust to weak IV are still valid with heterogeneous treatment effects.

In this paper, I first show that neglecting either heterogeneity or weak instruments can result in substantial distortions in inference. 
\Cref{sec:challenges} presents a simple simulation that has both weak instruments and heterogeneous treatment effects.
For a nominal 5\% test, using the procedure from \citet{mikusheva2022inference} (MS22), which is robust to weak instruments but not heterogeneity, can result in 100\% rejection under the null, because their test statistic is not centered correctly when there is heterogeneity. 
This result is attributed to how their test is a joint test of both the parameter value and the null of no heterogeneity.
Similarly, the procedure from \citet{evdokimov2018inference} (EK18), which is robust to heterogeneity but not weak instruments, can be severely oversized, as expected due to \citet{dufour1997some}.
Additionally, this section documents how an empirically common practice of constructing a ``leniency measure" that combines the many instruments and then using weak IV robust procedures from the just-identified IV literature is invalid. 

Given the evidence on how existing methods have substantial distortions in inference, \Cref{sec:valid_inference} proposes a procedure for valid inference. 
Following the many instruments literature, the JIVE estimand is the object of interest --- this estimand can be interpreted as a weighted average of treatment effects when there is heterogeneity (e.g., EK18).
Using weak identification asymptotics, I show that the Lagrange Multiplier (LM) (i.e., score) statistic, earlier proposed by \citet{matsushita2022jackknife} under constant treatment effects, is mean zero and asymptotically normal even with treatment effect heterogeneity.
In fact, I prove a stronger normality result that a set of jackknife statistics that includes the LM is jointly normal, which is the first technical challenge of this paper.
This normality result uses an asymptotic environment that nests the asymptotic environments of EK18 and MS22 in that normality holds if either the number of instruments is large or the instruments are strong.
This normality implies that, as long as the variance of LM is consistently estimable, a $t$-statistic can be calculated and critical values from the standard normal distribution are valid for inference. 
Obtaining a consistent variance estimator is the second technical challenge of the paper, since reduced-form coefficients are not consistently estimable when there are few observations per instrument.
Motivated by \citet{anatolyev2023testing} who proposed a method to jointly test the significance of many covariates in OLS, I construct a leave-three-out (L3O) variance estimator for the LM variance and show that it is consistent, even when reduced-form coefficients are not consistently estimable.
Due to the generality of the setting considered, beyond its robustness to weak IV and heterogeneity, the procedure proposed in this paper is also robust to heteroskedasticity, and potentially many covariates, so it retains the advantages of existing procedures in the literature. 

\Cref{sec:power} argues that the proposed LM procedure is powerful.
In the over-identified IV environment with normal homoskedastic errors, \citet{moreira2009maximum} showed that, if we are willing to restrict our attention to tests that are invariant to rotations of the instruments, it suffices to consider tests that are functions of three statistics.
These three statistics are known as a ``maximal invariant".
To be robust to non-normality and heteroskedasticity in the many IV environment, I focus on the leave-one-out (L1O) analog of this maximal invariant.
The proposed LM statistic is one of the three statistics in the L1O analog, and I show that the two-sided LM test is asymptotically uniformly most powerful unbiased (UMPU) within the class of tests that are functions of this L1O analog, for the interior of the alternative space (i.e., where heterogeneity is imposed). 
Simulation results in \Cref{sec:app_sim} show how the procedure is robust even with a small number of instruments, and it is reasonably powerful even with constant treatment effects. 
\Cref{sec:empirical_applications} contains two empirical applications that show how being robust to many weak IV and heterogeneity can change conclusions.\footnote{
Implementation code can be found at: \url{https://github.com/lutheryap/mwivhet}.
}
In the \citet{angrist1991does} quarter of birth application, the \citet{matsushita2022jackknife} procedure that are robust to many weak IV but not heterogeneity have unbounded confidence sets while L3O has a bounded confidence set.
In the \citet{agan2023misdemeanor} judge application, MS22 has an empty confidence set, and the length of the L3O confidence interval is more than twice that of EK18, which is not robust to many weak IV.

This paper contributes to the following strands of literature. First, this paper contributes to a growing literature on many weak instruments.
There is a strand of literature dealing with many instruments (e.g., \citet{chao2012asymptotic}) and another separate strand dealing with weak instruments (e.g., \citet{StaigerStock97, lmmpy23}). 
While recent procedures accommodate both simultaneously (e.g., \citet{crudu2021inference, mikusheva2022inference, matsushita2022jackknife, lim2024conditional}), their focus has been on the linear IV model with constant treatment effects. 
This paper augments their setup by allowing for heterogeneity in treatment effects, and contributes new results on the limitations of their procedures under heterogeneity.
Further, I show how heterogeneity can be understood in a framework analogous to weak instruments.

Second, this paper contributes to the literature on heterogeneous treatment effects (e.g., \citet{kolesar2013estimation, evdokimov2018inference, blandhol2022tsls}). 
The previous papers exploit consistent estimation of the object of interest to conduct inference. 
In contrast, this paper uses the (more general) weak IV environment where the object of interest may not be consistently estimated. 
Two recent papers allow weak IV and heterogeneity.
\citet{boot2024inference} study a single discrete instrument interacted and saturated with many covariates.
Their setup is a special case of the environment considered in this paper, so it is unclear if their procedure generalizes to many instruments without covariates (e.g., judges). 
\citet{kleibergen2025double} target a continuous updating (CU) GMM estimator with a fixed number of instruments rather than many instruments.\footnote{
They show that their CU-GMM estimator corresponds to the limited information maximum likelihood (LIML) estimator.
However, it is also known that the LIML estimand may not be interpretable as a weighted average of LATE's. \citep{kolesar2013estimation}
I am unaware of any paper that allows both weak IV and heterogeneity with a fixed number of instruments and targets a parameter that is a weighted average of LATE's.
}


Third, this paper contributes to a literature on inference when coefficients cannot be consistently estimated. 
The difficulty in having such a general robust inference procedure lies in consistent variance estimation when the number of coefficients is large.
Recent literature that has made substantial progress in a different context. 
In doing inference in OLS with many covariates, \citet{cattaneo2018inference} and \citet{anatolyev2023testing} proposed consistent variance estimators that are robust to heteroskedasticity, which involve inverting a large ($n$ by $n$, where $n$ is the sample size) matrix and a L3O approach respectively. \citet{boot2024inference} adapt the \citet{cattaneo2018inference} variance estimator for inference. 
In contrast, this paper adapts the approach from \citet{anatolyev2023testing} that does not require an inversion of an $n$ by $n$ matrix, and whose L3O implementation is fast when using matrix operations.

Fourth, this paper contributes to a literature on optimal tests. 
While the UMPU test for just-identified IV has been established since \citet{moreira2009tests}, obtaining a UMPU test in the over-identified IV environment has thus far been more challenging.
In the over-identified IV environment with constant treatment effects, several statistics are informative of the object of interest. 
Consequently, there is a large literature that numerically compares various valid tests and characterizes various forms of optimality (e.g., \citet{andrews2016conditional, andrews2019optimal, van2023power, lim2024conditional}).
By imposing heterogeneity in the environment, the problem is (somewhat surprisingly) simplified. Since only one statistic in the asymptotic distribution is directly informative of the object of interest, I obtain a UMPU result.


\section{Challenges in Conventional Practice} \label{sec:challenges}

This section explains the challenges faced in conventional practice by considering a simple potential outcomes model without covariates that exhibits weak instruments and heterogeneity in treatment effects. 
This model is a special case of the general model in Section \ref{sec:valid_inference}.
A simulation from the model shows how weak instruments and heterogeneity can lead to substantial distortions in inference for procedures recently proposed in the econometric literature. 
A common empirical practice of constructing a leave-one-out instrument and then applying inference methods for the instrument as if it is not constructed also has high rejection rates. 
In contrast, the method proposed in this paper has a rejection rate that is close to the nominal rate. 

\subsection{Setting for Simple Example} \label{sec:simple_setting}
The simple example uses the canonical latent variable framework of \citet{heckman2005structural}. 
We are interested in the effect of $X_i \in \{ 0,1 \}$ (e.g., incarceration) on some outcome $Y_i$, for $i = 1, \cdots n$ that indexes individuals.
To instrument for $X_i$, we use a vector of judges indicators: $Z_i$ is a $(K+1)$-dimensional vector of indicators for judges, indexed $k=1,\cdots, K+1$, each with $c=5$ individual cases, so the vector takes value 1 for the $k$th component when individual $i$ is matched to judge $k$, and 0 elsewhere.
Then, $n = (K+1)c$.
The problem of many instruments arises when $c$ is fixed while $K$ increases.
Let $Y_i(0)$ and $Y_i(1)$ denote the untreated and treated potential outcomes respectively, and we observe $Y_i = Y_i (X_i)$.
The treatment status given some instrument value $z$ is $X_i(z)$, and we observe $X_i (Z_i)$.
The model is:
\begin{equation}
X_i(z) = 1\{ z^\prime \lambda > v_i \}, \text{ and } Y_i(x) = xf(v_i) + \varepsilon_i,
\end{equation}
where $1\{ \cdot \}$ is an indicator function that takes the value 1 if the argument is true and 0 otherwise.
Here, $Z_i^\prime \lambda = \lambda_{k(i)}$, where $k(i)$ is the judge that individual $i$ is matched to.
With individual unobservable $v_i \sim U[0,1]$, the probability of treatment (i.e., $X_i=1$) given judge $k$ is $\lambda_k$.
I set $\lambda_k = 1/2$ for the base judge, and evenly split all other $K$ judges to take 4 different values of $\lambda_k$.
Potential outcomes are $Y_i(0) = \varepsilon_i$ and $Y_i(1) = f(v_i) + \varepsilon_i$ so $Y_i(1) - Y_i(0) = f(v_i)$ is the treatment effect. 
The individual-specific residuals $v_i$ and $\varepsilon_i$ are allowed to be arbitrarily correlated.
Let $\beta_k$ denote the local average treatment effect (LATE) when comparing judge $k$ to the base judge: for instance, when $\lambda_k > 1/2$, $\beta_k = \frac{1}{\lambda_k-1/2} \int_{1/2}^{\lambda_k} f(v) dv$.
The values of $(\lambda_k, \beta_k)$ for the 4 groups of judges are $(1/2 -s, \beta - h/s), (1/2(1-s), \beta + 2h/s), (1/2(1+s), \beta - 2h/s)$, and $(1/2+s, \beta + h/s)$.
The function $f(v)$ that delivers these parameters and further details of this example are in \Cref{sec:sect2DGP}.

The $\lambda_k$ and $\beta_k$ values are parameterized by objects $s$ and $h$, which control the IV strength and heterogeneity in the model respectively. 
The impact of these parameters are illustrated in \Cref{fig:RFSimpleExample} that plots the point masses for the four groups of judges in reduced-form. 
Parameter $s$ controls how far $E[X\mid Z]$ are spread across judges, which then affects the instrument strength.
Parameter $h$ controls the distance between the mass points and a line with slope $\beta$ --- this slope is the object of interest. 
If the impact of $X$ on $Y$ is homogeneous, then $h=0$, and all mass points \emph{must} lie on a line --- this implication is falsifiable by the data. 

\begin{figure}
    \centering
    \caption{IV Strength and Heterogeneity in Reduced Form}
    \label{fig:RFSimpleExample}
\begin{tikzpicture}
\begin{axis}[ xmin = 0, xmax = 1, ymin = -1, ymax = 1, axis lines* = left, clip = false, ]  
\addplot[color = black, solid, thick] coordinates{(0,0.2) (1,-0.2)};
\addplot[color = black, dashed, thin] coordinates{(0.5,-1) (0.5,1)};

\addplot[color = black, solid, thin] coordinates{(0.75,-0.1) (0.75,-0.15)};
\addplot[color = black, solid, thin] coordinates{(0.75,-0.15) (7/8,-0.15)};
\node [below] at (0.75,-0.15) {$\beta$};

\addplot[color = black, mark= *, only marks, mark size = 3pt] coordinates {(0.25,0.5+0.1) (3/8,-0.5+0.05) (0.5,0) (5/8,-0.5-0.05) (0.75,0.5-0.1)};

\addplot[color = blue, dashed, thick, <->] coordinates {(0.75,0.5-0.1) (0.5,0.5-0.1)}; 
\node [above] at (5/8,0.5-0.1) {\textcolor{blue}{$s$}};
\addplot[color = red, dashed, thick, <->] coordinates {(0.25,0+0.1) (0.25,0.5+0.1)}; 
\node [right] at (0.25,0.25+0.1) {\textcolor{red}{$h$}};

\node [left] at (-0.1,1) {$E[Y\mid Z]$}; 
\node [below] at (1,-1.1) {$E[X \mid Z]$};
\end{axis} 
\end{tikzpicture} 
\end{figure}

\begin{table}
    \centering
    \caption{Rejection rates under the null for nominal size 0.05 test} \label{tab:sim_mte_base}
    \include{fig_v8/simf_example_res}
    \justifying \small
    Notes: The table displays rejection rates of various procedures (in columns) for various designs (in rows). 
    Details of the data generating process are in \Cref{sec:sect2DGP}.
    I use $K=400,c=5,\beta=0$ with 1000 simulations. 
    TSLS implements the standard two-stage-least-squares $t$-test for an over-identified IV system. 
    EK implements the procedure in \citet{evdokimov2018inference}. 
    MS uses $T_{AR}$ with the cross-fit procedure in \citet{mikusheva2022inference}.
    MO uses the $T_{LM}$ statistic with the variance estimator proposed in \citet{matsushita2022jackknife}. 
    $\tilde{X}$-t uses a constructed instrument and runs TSLS for a just-identified IV system. 
    $\tilde{X}$-AR uses the \citet{anderson1949estimation} (AR) procedure for a just-identified system using a constructed instrument.
    L3O uses the variance estimator proposed in this paper. 
    LMorc is the infeasible theoretical benchmark that uses an LM statistic with an oracle variance. 
    ARorc uses the AR statistic with an oracle variance. 
\end{table}

The simulation designs vary the values of $s$ and $h$ through the parameters:
\begin{equation} \label{eqn:CS_CH}
    E[T_{FS}] = \frac{5}{8} \sqrt{K} (c-1) s^2, \text{ and } E[T_{AR}] = \frac{5}{8} \sqrt{K} (c-1) h^2.
\end{equation}
The statistic $T_{FS}$ is the leave-one-out (L1O) analog of the first-stage ``F" statistic in this model, and $T_{AR}$ is similarly the L1O analog of the Anderson-Rubin statistic under the null. 
These objects are explained in detail in \Cref{sec:valid_inference}, but it suffices to mention here that, for the given $c$ and $K$, there is a one-to-one mapping between $(E[T_{FS}],E[T_{AR}])$ and $(s,h)$.
Using \citet{StaigerStock97} asymptotics, $E[T_{FS}]$ is the parameter that determines whether there is strong or weak identification.
Where $C$ is some positive arbitrary constant, $E[T_{FS}] \rightarrow \infty$ is an environment with strong identification where the object of interest can be estimated consistently, and $E[T_{FS}] \rightarrow C < \infty$ is an environment with weak identification where no consistent estimator exists. 

For every design, I generate data under the null and calculate the frequency that each inference procedure rejects the null of $\beta_{0}=0$.
These procedures include the standard TSLS $t$-test, procedures that are robust to either weak instruments (MO, MS) or heterogeneity (EK), and procedures that use a constructed instrument ($\tilde{X}$).
The results are presented in \Cref{tab:sim_mte_base}, which I will refer to in the remainder of this section as I explain them.

\subsection{Issue with Many Weak Instruments}
Many IV and weak IV are different but related issues.
The many IV problem arises when the number of cases per judge $c$ does not diverge to infinity, so that $K$ is large relative to $n$. 
When $c$ is small, the judge-specific $\lambda_k$ and $\beta_k$ cannot be consistently estimated and hence inference procedures like the TSLS $t$-test can be oversized. 
In \Cref{fig:RFSimpleExample}, a small $c$ is attributed to the sample uncertainty surrounding each black circle.
The weak IV problem arises from $E[T_{FS}]$ not diverging: since $E[T_{FS}]$ is a function of $K,c$, and $s$, the weak IV issue is related to the many IV issue. 

If we simply run the TSLS $t$-test for an over-identified model, then the estimator can be asymptotically biased and inference is invalid, a fact already known in the literature. 
This fact is also evident in \Cref{tab:sim_mte_base}, where TSLS has 100\% rejection in many designs. 
In TSLS, the first stage regresses $X$ on $Z$ to get a predicted $\hat{X} = Z \hat{\pi}$, where $\hat{\pi}$ is the estimated coefficient; the second stage regresses $Y$ on $\hat{X}$. 
With constant treatment effects, the asymptotic bias of the TSLS estimator depends on $\sum_i \varepsilon_i \hat{X}_i / \sum_i \hat{X}_i^2 $. 
When every judge only has $c=5$ cases, the influence of $v_i$ on $\hat{\pi}_{k(i)}$ and hence $\hat{X}_i$ is non-negligible.
Since $\varepsilon_i$ and $v_i$ can be arbitrarily correlated, the numerator is biased.
If the instruments are weak such that the denominator $\sum_i \hat{X}_i^2$ does not diverge sufficiently quickly, then the asymptotic bias can be large.
Due to the asymptotic bias, the $t$-statistic is not centered around $\beta_0$ when data is generated under the null, so we observe over-rejection in \Cref{tab:sim_mte_base}. 

Since the bias in the TSLS estimator arises from using $X_i$ to estimate $\hat{\pi}$, a natural solution to address that bias is to use the JIVE to estimate $\beta$. 
Instead of using $\hat{X}_i = Z_i^\prime \hat{\pi}$ in the second stage, we instead use $\tilde{X}_i = Z_i^\prime \hat{\pi}_{-i}$, where $\hat{\pi}_{-i}$ is the coefficient from the first-stage regression that leaves out observation $i$. 
I call $\hat{\pi}_{-i}$ the leave-one-out (L1O) coefficient.
With $P=Z\left(Z^{\prime}Z\right)^{-1}Z^{\prime}$ denoting the projection matrix, $\tilde{X}_i = Z_i^\prime \hat{\pi}_{-i}$ can be written as $\tilde{X}_i = \sum_{j\ne i}P_{ij}X_{j}$. 
Then, the JIVE is:
\begin{equation}
\hat{\beta}=\frac{\sum_{i}Y_{i}\left(\sum_{j\ne i}P_{ij}X_{j}\right)}{\sum_{i}X_{i}\left(\sum_{j\ne i}P_{ij}X_{j}\right)}.
\end{equation}
In the many IV context with constant treatment effects, the asymptotic distribution of the $t$-statistic of the JIVE is the same as the distribution of the $t$-statistic of the TSLS estimator in the just-identified environment \citep{mikusheva2022inference} --- it is a ratio of two normally distributed random variables. 
It is well-known that, in the just-identified IV context with weak IV, the rejection rate of the standard $t$-statistic can be up to 100\% for a nominal 5\% test (e.g., \citet{dufour1997some}). 
Hence, like the just-identified IV context, by using a structural model that has sufficiently weak instruments and high covariance, the simulation can deliver high rejection rates.

EK18 have a procedure that is robust to heterogeneity, but not weak instruments, so even if we use their variance estimator for the $t$-statistic, this problem is not alleviated. 
This fact is evident in the EK column of \Cref{tab:sim_mte_base}, where, with a sufficiently large correlation in the individual unobservables, rejection rates can be large.\footnote{
The rejection rate of EK can be 100\% under the null in some simulations: one example is given in \Cref{tab:sim_cont} in Online \Cref{sec:details_power}.
}
Hence, ignoring the issue of weak instruments can lead to substantial distortions in inference.
In fact, even with strong instruments, there is no guarantee that EK18 achieves the nominal rate, because their variance estimation method requires consistent estimation of the first-stage coefficients $\hat{\pi}$.
A condition for consistent variance estimation is that the number of cases per judge is large, which is not $c=5$.  

\begin{remark} \label{remark:weakIV_def}
In the literature, there have been several definitions of weak IV, which I clarify in this remark.
Using \Cref{eqn:CS_CH}, there are three asymptotic regimes, ordered from the strongest to the weakest: (i) $\frac{1}{\sqrt{K}}E[T_{FS}] \rightarrow \infty$, (ii) $E[T_{FS}] \rightarrow \infty$, and (iii) $E[T_{FS}] \rightarrow C < \infty$.
Regime (i) is a necessary condition for the TSLS estimator to be consistent, so $\frac{1}{\sqrt{K}}E[T_{FS}] \rightarrow C < \infty$ is what \citet{StockYogo05} would refer to as weak instruments.
Regime (ii) is a necessary condition for the JIVE to be consistent (e.g., \citet{chao2012asymptotic}, EK18).
Regime (iii) is where no estimator is consistent (e.g., MS22). 
If $K$ is fixed, then (i) and (ii) are the same asymptotically, and (iii) is the relevant weak identification asymptotic regime.
If $K \rightarrow \infty$, then there is more ambiguity in what weakness means: \citet{chao2012asymptotic} and EK18 who assume (ii) are robust to weak instruments when defined in the \citet{StockYogo05} sense, because $s$ can converge to 0, albeit at a slower rate than $\sqrt{K}$.
In this paper, I follow the \citet{StaigerStock97} standard of weak identification where no consistent estimator exists, which corresponds to (iii) that EK18 is not robust to.
\end{remark}

\subsection{Issue with Heterogeneity}
Next, consider proposals for inference that are developed for contexts with many weak IV.
MS22 (and \citet{crudu2021inference}) propose using an Anderson-Rubin (AR) statistic $T_{AR}=\frac{1}{\sqrt{K}}\sum_{i}\sum_{j\ne i}P_{ij}e_{i}\left( \beta_0 \right)e_{j}\left( \beta_0 \right)$, for $e_{i}\left( \beta_0 \right):=Y_{i}-X_{i}\beta_{0}$ where $\beta_{0}$ is the hypothesized null value.
With constant treatment effects, $e_i := Y_i - X_i \beta$ is the residual.
Hence, if the instrument is orthogonal to the residual, then $E[Z_i e_i]=0$.\footnote{
An equivalent way to see how heterogeneity affects inference is through the framework of \citet{hall2003large} and \citet{lee2018consistent}: $E[Z_i e_i]=0$ is a special case of a misspecified over-identified GMM problem.
The instruments are individually valid, but every component of the $K$ moments in $E[Z_i e_i]=0$ identifies a different treatment effect, so there is no parameter that satisfies all moments simultaneously under heterogeneity.
Then, while the estimand is still interpretable as a combination of these treatment effects due to how GMM weights these moments, there are additional components in the variance that affect inference.}
Then, $T_{AR}$ is the L1O analog for the quadratic form that tests the moment $E[Z_i e_i]=0$.
Since observations are independent, the critical value for the test is obtained from a mean-zero normal distribution.
In this model, $E\left[T_{AR}\right]=\sqrt{K} (c-1) h^2$ under the null.\footnote{This result can be obtained as a special case of \Cref{thm:normality} in \Cref{sec:valid_inference} and using the fact that $\sum_{i}\sum_{j\ne i}P_{ij}^{2}=\sum_{i}\sum_{j\ne i}\left(1/c^{2}\right)=\sum_{i}\frac{c-1}{c^{2}}=\sum_{k}\frac{c-1}{c}$.}
Hence, when there are constant treatment effects such that $h=0$ for all $k$, the statistic is unbiased. 
However, in the setup with heterogeneity, the $T_{AR}$ can be biased: in fact, when $h$ does not converge to zero, $E[T_{AR}]$ diverges, resulting in a 100\% rejection rate under the null, even if the oracle variance were used.
Further, there does not exist any estimand $\beta$ such that $E\left[T_{AR}\right]=0$, as stated in \Cref{lem:MS_estimand} of \Cref{sec:main_supp_sect2}. 

A further problem with the feasible MS procedure is that when there is strong heterogeneity ($E[T_{AR}]=2\sqrt{K}$) in this simulation, their cross-fit variance estimate is negative for all simulation draws, as the negative heterogeneity terms are larger in magnitude than the positive variances of the residuals.
The formal analysis requires more notation from \Cref{sec:valid_inference}, so details are deferred to \Cref{sec:main_supp_sect2}.  

Another proposal in the literature that is robust to many weak instruments is \citet{matsushita2022jackknife} (MO22) who use the statistic $T_{LM}=\frac{1}{\sqrt{K}}\sum_{i}\sum_{j\ne i}P_{ij}e_{i} \left( \beta_0 \right) X_{j} = \frac{1}{\sqrt{K}}\sum_{i}e_{i}\left( \beta_0 \right)\tilde{X}_{i}$. 
This statistic can be interpreted as the LM (or score) statistic that uses the moment $E[e \tilde{X}] =0$.
They propose the following variance estimator $\hat{\Psi}_{MO}$:
\begin{equation} \label{eqn:VMO_def}
\hat{\Psi}_{MO} := \frac{1}{K} \sum_{i}\left(\sum_{j\ne i}P_{ij}X_{j}\right)^{2}e_{i}\left( \beta_0 \right)^{2}+\frac{1}{K}\sum_{i}\sum_{j\ne i}P_{ij}^{2}X_{i}e_{i}\left( \beta_0 \right)X_{j}e_{j}\left( \beta_0 \right).
\end{equation}
While $T_{LM}$ has zero mean under the null even with heterogeneity, a result shown later in \Cref{sec:valid_inference}, the MO22 variance estimator was constructed under constant treatment effects, so the variance estimand differs from the true variance. 
It can be shown that $E\left[\hat{\Psi}_{MO}\right]\ne \Var\left(T_{LM}\right)$, and $\hat{\Psi}_{MO}$ is inconsistent in general, so when it is used to construct the $t$-statistic of $T_{LM}$, the normalized statistic is not distributed $N(0,1)$ asymptotically.
Consequently, by constructing a DGP where $\hat{\Psi}_{MO}$ underestimates the variance, it is possible to get over-rejection of the MO22 procedure, as in the cases of \Cref{tab:sim_mte_base} where $E[T_{AR}]$ diverges.
As expected, when there is no heterogeneity such that $h=0$, the rejection rate of MO22 and MS22 are close to the nominal rate. 

\subsection{Issue with a Constructed Instrument} \label{sec:constructedIV}
In light of problems with weak identification and heterogeneity, there is a large applied literature that transforms a many IV environment into a just-identified single-IV environment. 
With a single IV, the \citet{anderson1949estimation} (AR) procedure (among others) is robust to both weak identification and heterogeneity. 
However, this subsection will argue that such an approach is invalid.

Due to how the JIVE is written, there are several empirical papers that treat $\tilde{X}_{i}=\sum_{j\ne i}P_{ij}X_{j}$ as the ``instrument'' so that $\hat{\beta}=\sum_{i}Y_{i}\tilde{X}_{i}/\sum_{i}X_{i}\tilde{X_{i}}$, and proceed with inference as if $\tilde{X}_{i}$ is not constructed, but is an observed scalar instrument, usually referred to as a leniency measure.
While the resulting estimator is numerically identical to JIVE, there are distortions in inference because the variance estimators do not account for the variability in constructing $\tilde{X}_{i}$. 

If the TSLS $t$-statistic inference is used as if $\tilde{X}_{i}$ is the instrument, then its rejection rates in designs with heterogeneity are usually higher than rejection rates of EK18 that accounts for the variance accurately, by comparing the $\tilde{X}$-t and EK columns in \Cref{tab:sim_mte_base}. 
Consequently, in the cases where EK under-rejects, $\tilde{X}$-t can have close to nominal rejection rates by coincidence.

Even if the weak IV robust AR procedure for just-identified IV were used, there can still be distortion in inference (see $\tilde{X}$-AR in \Cref{tab:sim_mte_base}).
The AR $t$-statistic is $t_{\tilde{X}AR} := \sum_{i}e_{i}\left( \beta_0 \right)\tilde{X}_{i}/ \sqrt{\hat{V}}$, where $\hat{V} = \sum_{i}\tilde{X}_{i}^{2}\hat{\varepsilon}_{i}^{2}/ \left(\sum_{i}\tilde{X}_{i}^{2}\right)^{2}$ and $\hat{\varepsilon}_{i}=e_{i}\left( \beta_0 \right)-\tilde{X}_{i}\left( \sum_{i}e_{i}\left( \beta_0 \right)\tilde{X}_{i}\right) / \left( \sum_{i}\tilde{X}_{i}^{2} \right)$.
Even though $t_{\tilde{X}AR}$ is mean zero and asymptotically normal, the variance estimand is inaccurate, much like MO22. 
In particular, when $\beta_0=\beta=0$, the leading term of the variance estimand is $E\left[\sum_{i}\tilde{X}_{i}^{2}e_{i}^{2}\right]$, and it does not converge to the true variance derived in \Cref{sec:valid_inference} in general.
Hence, using the just-identified AR procedure with a constructed instrument results in over-rejection.
There are several papers that cluster standard errors by judges, but this approach faces a similar issue.\footnote{
Details of this discussion are relegated to Online \Cref{sec:compare_var_estimands}.
} 

As a preview, the L3O procedure proposed in this paper has rejection rates close to the nominal rate while the other procedures can over-reject.

\section{Valid Inference} \label{sec:valid_inference}
In light of how existing procedures are invalid in an environment with many weak instruments and heterogeneity as documented in the previous section, this section describes a novel inference procedure and shows that it is valid. 
I set up a general model, then show that an LM statistic is asymptotically normal and a feasible variance estimator is consistent, which suffices for inference.

\subsection{Setting: Model and Asymptotic Distribution} \label{sec:setting}

The general setup mimics \citet{evdokimov2018inference}. 
With an independently drawn sample of individuals $i=1,\dots,n$, we observe each individual's scalar outcome $Y_{i}$, scalar endogenous variable $X_{i}$, instrument $Z_{i}$, and covariates $W_{i}$, with $\dim\left(Z_{i}\right)=K$.\footnote{
The endogenous variable $X_i$ can be extended to a vector with some technical modifications and without conceptual complications.
} 
For every instrument value $z$, there is an associated potential treatment $X_{i}\left(z\right)$, and we observe $X_{i}=X_{i}\left(Z_{i}\right)$. 
Similarly, potential outcomes are denoted $Y_{i}\left(x\right)$, with $Y_{i}=Y_{i}\left(X_{i}\right)$. 
Let $R_{i}:=E\left[X_{i}\mid Z_{i},W_{i}\right]$ and $R_{Yi}:=E\left[Y_{i}\mid Z_{i},W_{i}\right]$ be linear in $Z_i$ and $W_i$. 
The model, written in the reduced-form and first-stage equations, is:
\begin{align*}
Y_{i} & =R_{Yi}+\zeta_{i}, \text{ where }\quad &R_{Yi}&=Z_{i}^{\prime}\pi_{Y}+W_{i}^{\prime}\gamma_{Y}, \quad &E\left[\zeta_{i}\mid Z_{i},W_{i}\right]&=0, \text{ and }\\
X_{i} & =R_{i}+\eta_{i},\text{ where } \quad &R_{i}&=Z_{i}^{\prime}\pi+W_{i}^{\prime}\gamma,\quad &E\left[\eta_{i}\mid Z_{i},W_{i}\right]&=0.
\end{align*}
The setup implicitly conditions on $Z_{i},W_{i}$, so $R_{i},R_{Yi}$ are nonrandom.\footnote{If we are interested in a superpopulation where $Z$ is random, then the estimands would be defined as the probability limit of the conditional objects. Then, it suffices to have regularity conditions to ensure that the conditional object converges to the unconditional object.} 
Linearity in $Z$ and $W$ is not necessarily restrictive when there is full saturation or when $K$ is large.\footnote{Any nonlinear function of the instruments can be arbitrarily well-approximated by a spline with a large number of pieces or a high-order polynomial. 
Moreover, the arguments in this paper could presumably be extended to a linear approximation of nonlinear functions as long as there are regularity conditions to ensure that higher-order terms are asymptotically negligible. } 

Define $e_{i}:=Y_{i}-X_{i}\beta$, where $\beta$ is some estimand of interest, and $e_{i}$ is a linear transformation. 
Let $e_i \left( \beta_0 \right) := Y_i - X_i \beta_0$ denote the feasible null-imposed linear transformation.
Let $R_{\Delta i}:=R_{Yi}-R_{i}\beta$ and $\nu_{i}:=\zeta_{i}-\eta_{i}\beta$. 
These definitions imply $e_{i}=R_{\Delta i}+\nu_{i}$ and $R_{\Delta i} = Z_i^\prime (\pi_Y - \pi \beta) + W_i^\prime (\gamma_Y - \gamma \beta)$.
Since $E\left[\nu_{i}|Z_{i},W_i\right]=0$ from the model, $E\left[e_{i}|Z_{i},W_i\right]=R_{\Delta i}$, which need not be zero. 
For data matrix $A$, let $H_{A}=A\left(A^{\prime}A\right)^{-1}A^{\prime}$ denote the hat (i.e., projection) matrix and $M_{A}=I-H_{A}$ its corresponding annihilator matrix. 
With $Z,W$ denoting the corresponding data matrices of the instrument and covariates, let $Q=(Z,W)$, $P=H_{Q}$, and $M=I-P$. 
$C$ denotes arbitrary constants. 

\begin{remark}
While $E\left[e_{i}|Z_{i}, W_i\right]=R_{\Delta i}$ need not be zero under heterogeneous treatment effects, $E\left[e_{i}|Z_{i}, W_i\right]=R_{\Delta i}=0$ under constant treatment effects.
Since $R_{\Delta i} = Z_i^\prime (\pi_Y - \pi \beta) + W_i^\prime (\gamma_Y - \gamma \beta)$ for all $i$, constant treatment effects with $E[Y_i - X_i \beta \mid Z_{i}, W_i] = 0$ also implies $\pi_Y = \pi \beta$ and $\gamma_Y = \gamma \beta$ outside of edge cases (e.g., when $Z_i,W_i$ are always 0). 
These $R_\Delta$ objects hence capture the impact of having heterogeneous treatment effects in the many instruments model. 
\end{remark}

The (conditional) object of interest and its corresponding estimator are:
\[
\beta_{JIVE}:=\frac{\sum_{i}\sum_{j\ne i}G_{ij}R_{Yi}R_{j}}{\sum_{i}\sum_{j\ne i}G_{ij}R_{i}R_{j}}, \text{ and }\quad\hat{\beta}_{JIVE}=\frac{\sum_{i}\sum_{j\ne i}G_{ij}Y_{i}X_{j}}{\sum_{i}\sum_{j\ne i}G_{ij}X_{i}X_{j}}, 
\]
where $G$ is an $n\times n$ matrix that can take several forms.
If there are no covariates, using the projection matrix $G=H_Z=P$ is the standard JIVE, and when there are covariates, I use the unbiased JIVE ``UJIVE" \citep{kolesar2013estimation} with $G=\left(I-\diag\left(H_{Q}\right)\right)^{-1}H_{Q}-\left(I-\diag\left(H_{W}\right)\right)^{-1}H_{W}$.
In an environment with a binary instrument and many covariates interacted with the instrument, the saturated estimand ``SIVE" \citep{chao2023jackknife, boot2024inference} uses $G=P_{BN}-M_QD_{BN}M_Q$, where $P_{BN}=M_{W}Z\left(Z^{\prime}M_{W}Z\right)^{-1}Z^{\prime}M_{W}$ and $D_{BN}$ is defined as a diagonal matrix with elements such that
$P_{BN,ii}=\left[M_QD_{BN}M_Q\right]_{ii}$. 
With constant treatment effects, the estimand is the same for all the estimators: $R_{Yi}=R_{i}\beta$ so $\beta_{JIVE}=\beta$. 
Depending on the application, the estimand is usually interpretable as some weighted average of treatment effects when using JIVE without covariates or UJIVE with covariates with a saturated regression.\footnote{
In the judge example without covariates above, we have $G=P$ and $\pi_{Yk} = \beta_k \pi_k$ where $\beta_k$ is the local average treatment effect (LATE) between judge $k$ and the base judge, so $\beta_{JIVE} = \frac{\sum_k \pi_{Yk} \pi_k}{\sum_k \pi_k^2 } = \frac{\sum_k \pi_k^2 \beta_k}{\sum_k \pi_k^2}$ is a weighted average of LATE's. 
} (EK18)
The focus of this paper is on inference, so I will not discuss the estimand in detail. 
The results for valid inference are established for any $G$ that satisfies properties that will be formally stated in the theorem. 

This paper restricts its attention to the following statistics:
\begin{equation}
\left(T_{AR},T_{LM},T_{FS}\right)^{\prime}:=\frac{1}{\sqrt{K}}\sum_{i}\sum_{j\ne i}G_{ij}\left(e_{i}\left( \beta_0 \right)e_{j}\left( \beta_0 \right),e_{i}\left( \beta_0 \right)X_{j},X_{i}X_{j}\right)^{\prime}.
\end{equation}
It suffices to focus on $\left(T_{AR},T_{LM},T_{FS}\right)$ for inference as they correspond to a linear transformation of the leave-one-out analog of a maximal invariant --- details are in \Cref{sec:max_inv}. 
$T_{AR}$ is the (unnormalized) AR statistic used by MS22 for inference, and $T_{LM}$ is the LM (score) statistic used by MO22. 
$T_{FS}$ corresponds to a first-stage F statistic that can be used as a diagnostic for weak instruments. 

Asymptotic theory in this paper uses $r_n/ \sqrt{K} \rightarrow \infty$, where
\begin{equation} \label{eqn:rn}
    r_n := \sum_i \left( \sum_{j \ne i} G_{ij} R_j \right)^2 +\sum_i \left( \sum_{j \ne i} G_{ij} R_{\Delta j} \right)^2 + \sum_i \sum_{j \ne i} G_{ij}^2,
\end{equation}
which nests the environments of EK18, MS22, and MO22: as long as one of the three objects in \Cref{eqn:rn} diverges at a rate above $\sqrt{K}$, we obtain $r_n/ \sqrt{K} \rightarrow \infty$. 
EK18 assume strong identification that translates to $\sum_i \left( \sum_{j \ne i} G_{ij} R_j \right)^2/ \sqrt{K} \rightarrow \infty$ in this setting, but $r_n/ \sqrt{K} \rightarrow \infty$ can also be achieved if either of the latter terms in $r_n$ diverges.
MS22 and MO22 assume $K \rightarrow \infty$. 
Without covariates, $G=P$, so $\sum_i \sum_{j \ne i} G_{ij}^2 = O(K)$, and hence $r_n/ \sqrt{K} \rightarrow \infty$.
Hence, to apply the asymptotic theory in this paper, it suffices to have either strong identification, or $K \rightarrow \infty$.
The only case ruled out is where $K$ is fixed, \emph{and} there is weak identifcation in that $\sum_i \left( \sum_{j \ne i} G_{ij} R_j \right)^2/ \sqrt{K}$ does not diverge.

\begin{assumption} \label{asmp:normality}
\begin{enumerate} [topsep=0pt,label=(\alph*)]
\item There exists $C < \infty$ such that $E[\eta_i^4] + E[\nu_i^4] \leq C$ for all $i$.
\item $E\left[\nu_{i}^{2}\right]$ and $E\left[\eta_{i}^{2}\right]$ are bounded away from 0 and $|\corr\left(\nu_{i},\eta_{i}\right)|$ is bounded away from 1. 
\item There exists $\underline{c}>0$ such that for any $c_1, c_2, c_3$ that are not all 0, \\ $\frac{1}{r_n} \sum_i \left( c_{3}\sum_{j\ne i}\left(G_{ij}+G_{ji}\right)R_{j}+c_{2}\sum_{j\ne i}G_{ji}R_{\Delta j} \right)^2$ \\$+ \frac{1}{r_n} \sum_i \left( c_{1}\sum_{j\ne i}\left(G_{ij}+G_{ji}\right)R_{\Delta j}+c_{2}\sum_{j\ne i}G_{ij}R_{j} \right)^2 $ \\ $+ \frac{1}{r_n} \Var\left( \sum_{i}\sum_{j\ne i} G_{ij}\left( c_1 \nu_i \nu_j + c_2 \nu_i \eta_j + c_3 \eta_i \eta_j \right) \right)$ $ \geq \underline{c}$.
\item $\frac{1}{r_n^2} \sum_{i} \left( \left(\sum_{j\ne i}G_{ij}R_{j}\right)^{4} + \left(\sum_{j\ne i}G_{ij}R_{\Delta j}\right)^{4} + 
\left(\sum_{j\ne i}G_{ji}R_{j}\right)^{4} + \left(\sum_{j\ne i}G_{ji}R_{\Delta j}\right)^{4} \right) \rightarrow0$.
\item $|| \frac{1}{r_n} G_{L}G_{L}^{\prime}||_{F}+|| \frac{1}{r_n} G_{U}G_{U}^{\prime}||_{F}\rightarrow 0$,
where $G_{L}$ is a lower-triangular matrix with elements $G_{L,ij}=G_{ij}1\left\{ i>j\right\} $
and $G_{U}$ is an upper-triangular matrix with elements $G_{U,ij}=G_{ij}1\left\{ i<j\right\} $.
\end{enumerate}
\end{assumption}

\Cref{asmp:normality} states high-level conditions that mimic EK18 so that a central limit theorem (CLT) can be applied. 
These conditions hence accommodate the $G$ that EK18 consider with covariates. 
Having bounded moments in (a) is standard.
Conditions (b) and (c) are sufficient to ensure that the variance is non-zero asymptotically.
In particular, (b) rules out perfect correlation: in the simulation, $\corr(\eta_i,\nu_i) = -1$ is the pathological case that makes the variance zero, but $\corr(\eta_i,\nu_i) =1$ still allows non-zero variance.
Conditions (d) and (e) ensure that the weights placed on the individual stochastic terms are not too large.
The condition that $r_n/\sqrt{K} \rightarrow \infty$ is implied by (e) when $G=P$: due to Lemma B3 of \citet{chao2012asymptotic}, under weak IV asymptotics where $P_{ii}\leq C<1$, we obtain $||G_{L}G_{L}^{\prime}||_{F} \leq C \sqrt{K}$.

Mechanically, if there is weak IV and fixed K, then $|| \frac{1}{r_n} G_{L}G_{L}^{\prime}||_{F} =  \frac{1}{K} O(\sqrt{K}) \ne o(1)$, so (e) fails when $r_n/\sqrt{K}$ does not diverge.
Notably, the conditions do not require $P_{ii} \rightarrow 0$ so the $\pi,\pi_{Y}$ coefficients need not be consistently estimated.

\begin{theorem} \label{thm:normality}
If \Cref{asmp:normality} holds and $\beta = \beta_0$, then $\hat{\beta}_{JIVE}-\beta_{JIVE} = T_{LM}/ T_{FS}$, and for $V = Var(T_{AR},T_{LM},T_{FS})$, 
\begin{equation} \label{eqn:thm1_distr}
V^{-1/2} \left(\begin{array}{c}
T_{AR} - \frac{1}{\sqrt{K}}\sum_{i}\sum_{j\ne i}G_{ij}R_{\Delta i}R_{\Delta j}\\
T_{LM}\\
T_{FS} -\frac{1}{\sqrt{K}}\sum_{i}\sum_{j\ne i}G_{ij}R_{i}R_{j}
\end{array}\right)\xrightarrow{d}N\left(\left(\begin{array}{c}
0\\
0\\
0
\end{array}\right),I_3\right).
\end{equation}
\end{theorem}

In \Cref{thm:normality}, $E[T_{FS}] = \frac{1}{\sqrt{K}}\sum_{i}\sum_{j\ne i}G_{ij}R_{i}R_{j}$ is the concentration parameter corresponding to the instrument strength.
In the model of \Cref{sec:challenges}, the mapping to the reduced-form $\pi$ can be found in \Cref{sec:sect2DGP}, so the concentration parameter is given by $E[T_{FS}] = \frac{1}{\sqrt{K}} \sum_k (c-1) \pi_k^2 = \frac{5}{8} \sqrt{K} (c-1) s^2$.\footnote{
This concentration parameter is comparable to the concentration parameter in just-identified IV.
With slight abuse of notation, suppose the just-identified IV model has a first stage equation with $X= Z \pi + v$ where $\pi = s$. 
Then, omitting variance normalizations, using the notation from \citet{lmmpy23}, the concentration parameter is $f_0 = \sqrt{n} s$, which determines if the TSLS estimator is consistent. 
In the L1O asymptotics, $\sqrt{K} (c-1)s^2 \approx ns^2/\sqrt{K}$ by using $n=(K+1)c$ and approximations $\sqrt{K/(K+1)} \approx 1$ and $(c-1)/\sqrt{c} \approx \sqrt{c}$. 
By comparing the L1O concentration parameter $ns^2/\sqrt{K}$ with the just-identified IV concentration parameter $f_0 = \sqrt{n} s$, I obtain the notions of weak identification in \Cref{remark:weakIV_def}.
}
If the instruments are strong, then $E[T_{FS}]\rightarrow\infty$, so $\hat{\beta}_{JIVE}-\beta_{JIVE}\xrightarrow{d}0$.
With weak IV, $E[T_{FS}]$ converges to some constant $C<\infty$, so comparing the JIVE $t$-statistic with the standard normal distribution leads to invalid inference even in large samples.

The asymptotic distribution follows from establishing a quadratic CLT that may be of independent interest: it is proven by rewriting the leave-one-out sums as a martingale difference array, and then applying the martingale CLT. 
While there are existing quadratic CLT available, they do not fit the context exactly. 
\citet{chao2012asymptotic} Lemma A2 requires $G$ to be symmetric, which works for $G=P$, but $G$ for UJIVE is not symmetric in general. 
EK18 Lemma D2 is established for scalar random variables, so I extend it to random vectors.

\Cref{thm:normality} states that $T_{LM}$ is mean zero and asymptotically normal. 
Hence, if we have access to the oracle variance of $T_{LM}$, we can simply use the statistic $T_{LM}/\sqrt{\Var(T_{LM})}$ for testing because it has a standard normal distribution under the null. 
Obtaining a consistent estimator is an issue addressed in the next subsection.


\subsection{Variance Estimation}

To test the null that $H_{0}:\beta=\beta_{0}$, we can calculate $T_{LM}$ using the null-imposed $\beta_{0}$ and an estimator for the variance of $\sqrt{K} T_{LM}$, $\hat{V}_{LM}$, defined later in this section. 
Then, reject if $KT_{LM}^{2}/\hat{V}_{LM}\geq\Phi\left(1-\alpha/2\right)^{2}$ for a size $\alpha$ test where $\Phi(.)$ is the standard normal CDF.
This procedure is valid when $T_{LM}$ is asymptotically normal with mean zero as we have established in the previous section, and when $\hat{V}_{LM}$ is consistent.

Before stating the variance estimator, I first decompose the variance expression in the equation below, which follows from substituting $e_i=R_{\Delta i} + \nu_i$ and $X_i = R_i + \eta_i$ into the variance.
For $V_{LM}:=\Var\left(\sum_{i}\sum_{j\ne i}G_{ij}e_{i}X_{j}\right)$,
\begin{equation} \label{lem:var_expr}
\begin{split}
V_{LM} & =\sum_{i}\sum_{j\ne i}\sum_{k\ne i}E\left[\nu_{i}^{2}\right]G_{ij}G_{ik}R_{j}R_{k}+\sum_{i}\sum_{j\ne i}G_{ij}^{2}E\left[\nu_{i}^{2}\right]E\left[\eta_{j}^{2}\right]+\sum_{i}\sum_{j\ne i}G_{ij}G_{ji}E\left[\eta_{i}\nu_{i}\right]E\left[\eta_{j}\nu_{j}\right]\\
 & \quad+2\sum_{i}\sum_{j\ne i}\sum_{k\ne i}E\left[\nu_{i}\eta_{i}\right]G_{ij}G_{ki}R_{j}R_{\Delta k}+\sum_{i}\sum_{j\ne i}\sum_{k\ne i}E\left[\eta_{i}^{2}\right]G_{ji}G_{ki}R_{\Delta j}R_{\Delta k}.
\end{split}
\end{equation}

With constant treatment effects, only the first line appears in the variance as $R_{\Delta} =0$. 
With $G=P$, the expression for $\Var\left(\sum_{i}\sum_{j\ne i}P_{ij}e_{i}X_{j}\right)$ matches the expression in EK18 Theorem 5.3, but their variance estimator cannot be used directly as they required consistent estimation of reduced-form coefficients. 
By adapting the leave-three-out (L3O) approach of \citet{anatolyev2023testing} (AS23), an unbiased and consistent variance estimator can be obtained. 
Intuitively, just as the own-observation bias in TSLS that involves a single sum can be addressed with L1O, an unbiased estimator for the variance expression that involves a triple sum can be obtained with L3O.
Let $\tau:=(\pi^\prime, \gamma^\prime)^\prime$ and $\tau_\Delta := ((\pi_Y-\pi \beta)^\prime, (\gamma_Y - \gamma \beta)^\prime)^\prime$ denote the coefficients on $Q$ when running the regression of $X$ and $e$ respectively. 
The variance estimator is:
\begin{equation} \label{eqn:Vhat_LM}
\hat{V}_{LM}:=A_1+A_2 +A_3 + A_4 +A_5 ,
\end{equation}
with
\begin{align*}
A_{1} & := \sum_i \sum_{j\ne i}\sum_{k\ne i}G_{ij}X_{j}G_{ik}X_{k}e_{i}\left(\beta_0\right)\left(e_{i}\left(\beta_0\right)-Q_{i}^{\prime}\hat{\tau}_{\Delta,-ijk}\right),\\
A_{2} & :=2\sum_i \sum_{j\ne i}\sum_{k\ne i}G_{ij}X_{j}G_{ki}e_{k}\left(\beta_0\right)e_{i}\left(\beta_0\right)\left(X_{i}-Q_{i}^{\prime}\hat{\tau}_{-ijk}\right),\\
A_{3} & :=\sum_i \sum_{j\ne i}\sum_{k\ne i}G_{ji}e_{j}\left(\beta_0\right)G_{ki}e_{k}\left(\beta_0\right)X_{i}\left(X_{i}-Q_{i}^{\prime}\hat{\tau}_{-ijk}\right),\\
A_{4} & :=-\sum_{i}\sum_{j\ne i}\sum_{k\ne j}G_{ji}^{2} X_{i}\check{M}_{ik,-ij}X_{k}e_{j}\left(\beta_0\right)\left(e_{j}\left(\beta_0\right)-Q_{j}^{\prime}\hat{\tau}_{\Delta,-ijk}\right),\\
A_{5} & :=-\sum_{i}\sum_{j\ne i}\sum_{k\ne j}G_{ij}G_{ji}e_{i}\left(\beta_0\right)\check{M}_{ik,-ij}X_{k}e_{j}\left(\beta_0\right)\left(X_{j}-Q_{j}^{\prime}\hat{\tau}_{-ijk}\right),
\end{align*}
where 
\begin{align*}
\hat{\tau}_{-ijk} & :=\left(\sum_{l\ne i,j,k}Q_{l}Q_{l}^{\prime}\right)^{-1}\sum_{l\ne i,j,k}Q_{l}X_{l},\\
\hat{\tau}_{\Delta,-ijk} & :=\left(\sum_{l\ne i,j,k}Q_{l}Q_{l}^{\prime}\right)^{-1}\sum_{l\ne i,j,k}Q_{l}e_{l}\left(\beta_0\right),\\
D_{ij} & :=M_{ii}M_{jj}-M_{ij}^{2}, \text{ and }\\
\check{M}_{ik,-ij} & :=\frac{M_{jj}M_{ik}-M_{ij}M_{jk}}{D_{ij}}=-Q_{i}^{\prime}\left(\sum_{l\ne i,j}Q_{l}Q_{l}^{\prime}\right)^{-1}Q_{k}.
\end{align*}

Following AS23, \Cref{asmp:L3O} ensures that the L3O estimator is well-defined.\footnote{
If these conditions are not satisfied, then we can follow the modification in AS23 so that the variance estimator is conservative.}
\begin{assumption} \label{asmp:L3O}
\begin{enumerate} [topsep=0pt,label=(\alph*)]
\item $\sum_{l\ne i,j,k}Q_{l}Q_{l}^{\prime}$ is invertible for every $i,j,k\in\left\{ 1,\cdots,n\right\} $. 
\item $\max_{i\ne j\ne k\ne i}D_{ijk}^{-1}=O_{P}(1)$, where $D_{ijk}:=M_{ii}D_{jk}-\left(M_{jj}M_{ik}^{2}+M_{kk}M_{ij}^{2}-2M_{jk}M_{ij}M_{ik}\right)$.
\end{enumerate}
\end{assumption}

\Cref{asmp:L3O}(a) corresponds to AS23 Assumption 1 and \Cref{asmp:L3O}(b) corresponds to AS23 Assumption 4. 
For consistent variance estimation, we additionally require regularity conditions that are stated in \Cref{asmp:Vconst_reg} of \Cref{sec:assumptions_inference}.
These conditions are satisfied when $G$ is a projection matrix. 
With these conditions, \Cref{thm:var_consistent} below claims that the variance estimator is consistent. 

\begin{theorem} \label{thm:var_consistent}
If $\beta = \beta_0$, Assumptions \ref{asmp:normality}-\ref{asmp:L3O} hold, and \Cref{asmp:Vconst_reg} in \Cref{sec:assumptions_inference} holds, then  $E\left[\hat{V}_{LM}\right]=V_{LM}$ and $\hat{V}_{LM}/V_{LM}\xrightarrow{p}1$.
\end{theorem}

With many instruments and potentially many covariates, the reduced-form coefficients $\pi,\pi_{Y},\gamma,\gamma_Y$ are not consistently estimable. 
The usual approach to constructing variance estimators calculates residuals by using the estimated coefficients, but this approach no longer works when these estimated coefficients are inconsistent.
To be precise, applying Chebyshev's inequality for any $\epsilon>0$ yields:
\begin{equation} \label{eqn:L3O_chebyshev}
\Pr \left(\left|\frac{\hat{V}_{LM}-V_{LM}}{V_{LM}}\right|>\epsilon\right) \leq\frac{1}{\epsilon^{2}}\frac{ Var\left( \hat{V}_{LM} \right)}{V_{LM}^{2}} + \frac{1}{\epsilon^{2}}\frac{ \left( E \left[  \hat{V}_{LM} \right] - V_{LM} \right)^2 }{V_{LM}^{2}}.
\end{equation}
Without an unbiased estimator and when reduced-form coefficients cannot be consistently estimated, the second term in (\ref{eqn:L3O_chebyshev}) is not necessarily asymptotically negligible.
To overcome this problem, I use an unbiased variance estimator so that the second term is exactly zero.
Then, it suffices to show that the variance of individual components of the variance are asymptotically small compared to $V_{LM}^2$, so that the first term in (\ref{eqn:L3O_chebyshev}) is $o(1)$ by applying the Cauchy-Schwarz inequality.

To obtain an unbiased estimator, I use estimators for the reduced-form coefficients $\pi,\pi_{Y},\gamma,\gamma_Y$ that are unbiased and independent of objects that they are multiplied with.
The leave-three-out (L3O) approach has this unbiasedness property for linear regressions: when leaving three observations out in the inner-most sum of the $A$ expressions, the estimated coefficient $\hat{\tau}_{-ijk}$ is independent of $i,j,k$ and is unbiased for $\tau$. 
Then, when taking the expectation through a product of random variables of $i,j,k$ and $\hat{\tau}_{-ijk}$, $\tau$ can be used in place of the $\hat{\tau}_{-ijk}$ component, and the expectations of individual components can be isolated.
For instance,
\begin{equation}
\begin{split}
E&\left[\sum_{i}\sum_{j\ne i}\sum_{k\ne i,j}G_{ij}X_{j}G_{ik}X_{k}e_{i}\left(e_{i}-Q_{i}^{\prime}\hat{\tau}_{\Delta,-ijk}\right) \right] \\
 &=\sum_{i}\sum_{j\ne i}\sum_{k\ne i,j}G_{ij}E\left[X_{j}\right]G_{ik}E\left[X_{k}\right]E\left[e_{i}\left(e_{i}-Q_{i}^{\prime}\hat{\tau}_{\Delta,-ijk}\right)\right] =\sum_{i}\sum_{j\ne i}\sum_{k\ne i,j}G_{ij}R_{j}G_{ik}R_{k}E\left[\nu_{i}^{2}\right],
\end{split}
\end{equation}
which recovers the triple sums in the $V_{LM}$ expression of (\ref{lem:var_expr}).
Without leaving out observations $j$ and $k$, we would not be able to isolate $E[X_j]$ and $E[X_k]$ in the first equality.
Without leaving out observation $i$, we would not be able to isolate $\tau_{\Delta}$ on expectation to obtain $E[\nu_i^2]$ in the second equality.
An analogous argument applies to other components of $V$ in (\ref{eqn:thm1_distr}).
Assuming that the residuals have zero mean conditional on $Q$ is crucial: if we merely have $E[Q \zeta]=0$, this argument can no longer be applied.

\begin{remark}
While the proposed $\hat{V}_{LM}$ is motivated by AS23, the contexts and estimators are different. First, the statistic that we are estimating the variance for is different: AS23 demeaned their $\mathcal{F}$ statistic using $\hat{E}_\mathcal{F}$, where $\hat{E}_\mathcal{F}$ is estimated using L1O, so they are interested in the variance of $\mathcal{F} - \hat{E}_\mathcal{F}$ that is mean zero; I use a mean-zero L1O statistic directly in $T_{LM}$. Second, the expectation of their variance estimator takes the form of their (9), which is analogous to the sum of $A_1$ and $A_4$ using the notation above, so repeated applications of their estimator is insufficient to recover all five terms. Hence, to adjust for the $A_4$ and $A_5$ terms here, I additionally require another estimator, and its form is similarly motivated by a L3O reasoning.
\end{remark}

Inverting the test to obtain a confidence set is straightforward, as the test statistic $T_{LM}^2$ and variance estimator $\hat{V}_{LM}=B_{0}+B_{1}\beta_{0}+B_{2}\beta_{0}^{2}$ are quadratic in $\beta_0$ (for some $B_0,B_1,B_2$ that are functions of the data), so the confidence set is obtained by solving a quadratic inequality.\footnote{Details are relegated to the online appendix.}


\section{Power Properties} \label{sec:power}
This section characterizes power properties of the valid LM procedure. 
I first argue that we can restrict our attention to three statistics that are jointly normal by extending the argument from \citet{moreira2009maximum}.
Since the covariance matrix can be consistently estimated, the remainder of the section focuses on the 3-variable normal distribution with a known covariance matrix.
With this asymptotic distribution, I show that the two-sided LM test is the uniformly most powerful unbiased test within the interior of the parameter space. 

\subsection{Sufficient Statistics and Maximal Invariant} \label{sec:max_inv}

As is standard in the literature, I consider the canonical model without covariates where the reduced-form errors are normal and homoskedastic (e.g., \citet{andrews2006optimal, moreira2009maximum}). 
Suppose $\left(\eta,\zeta\right)$ in the model of \Cref{sec:setting} are jointly normal with known variance:
\begin{equation} \label{eqn:rf_hom_mod}
\left(\begin{array}{c}
\zeta_{i}\\
\eta_{i}
\end{array}\right)\sim N\left(0,\Omega\right)=N\left(0,\left[\begin{array}{cc}
\omega_{\zeta \zeta} & \omega_{\zeta\eta}\\
\omega_{\zeta\eta} & \omega_{\eta \eta}
\end{array}\right]\right).
\end{equation}
Define:
\[
\left(\begin{array}{c}
s_{1}\\
s_{2}
\end{array}\right):=\left(\begin{array}{c}
\left(Z^{\prime}Z\right)^{-1/2}Z^{\prime}Y\\
\left(Z^{\prime}Z\right)^{-1/2}Z^{\prime}X
\end{array}\right).
\]
I restrict attention to tests that are invariant to rotations of $Z$, i.e., transformations of the form $Z \rightarrow ZF^\prime$ where $F$ is a $K \times K$ orthogonal matrix.
In particular, an invariant test $\phi(s_1,s_2)$ is one for which $\phi(Fs_1, Fs_2) =\phi(s_1,s_2)$ for all $K \times K$ orthogonal matrices $F$.
If we focus on invariant tests, then the maximal invariant contains all relevant information from the data for inference.

Due to \citet{moreira2009maximum} Proposition 4.1, $\left(s_{1}^{\prime},s_{2}^{\prime}\right)^{\prime}$ are sufficient statistics for $\left(\pi_{Y}^{\prime},\pi^{\prime}\right)^{\prime}$.
Further, $\left(s_{1}^{\prime}s_{1},s_{1}^{\prime}s_{2},s_{2}^{\prime}s_{2}\right)$ is a maximal invariant, and 
\[
\left(\begin{array}{c}
s_{1}\\
s_{2}
\end{array}\right)\sim N\left(\left(\begin{array}{c}
\left(Z^{\prime}Z\right)^{1/2}\pi_{Y}\\
\left(Z^{\prime}Z\right)^{1/2}\pi
\end{array}\right),\Omega\otimes I_{K}\right).
\]
The maximal invariant $(s_{1}^{\prime}s_{1},s_{1}^{\prime}s_{2},s_{2}^{\prime}s_{2})$ is jointly normal with a mean that depends on $\Omega$ when $K\rightarrow \infty$.\footnote{
This result is stated in Online \Cref{sec:further_analytics_power}.
}
Extending the argument to allow for heterogeneous treatment effects, the object of interest is $\beta = \frac{\pi^\prime Z^\prime Z \pi_Y}{\pi^\prime Z^\prime Z \pi}$ (following EK18), which is invariant to rotations of the instrument.\footnote{
With rotation matrix $F^\prime$ such that $F^\prime F = I$, observe that $X= ZF^\prime F\pi + \eta$, so if we were to run the regression on $ZF^\prime$ instead of $Z$, we would obtain coefficients $F \pi$ instead of $\pi$. 
Then, the estimand is $ \frac{\pi^\prime F^\prime F Z^\prime Z  F^\prime F \pi_Y}{\pi^\prime F^\prime F Z^\prime Z F^\prime F \pi} =  \frac{\pi^\prime Z^\prime Z \pi_Y}{\pi^\prime Z^\prime Z \pi}$ as before.
}

To be robust to many instruments, heteroskedasticity, and non-normality, I use the leave-one-out (L1O) analog of the maximal invariant (following MS22; \citet{lim2024conditional}).\footnote{
With heteroskedasticity, the variances are not consistently estimable, so we cannot correct for the variances directly. 
These variances no longer feature in the L1O analog of the maximal invariant. 
} 
Without covariates such that $G=P$, the L1O analog $\frac{1}{\sqrt{K}} \sum_i \sum_{j\ne i} P_{ij} (Y_i Y_j, Y_i X_j, X_i X_j)$ is a linear transformation of $(T_{AR}, T_{LM}, T_{FS})$.\footnote{
To see that $\frac{1}{\sqrt{K}} \sum_i \sum_{j\ne i} P_{ij} (Y_i Y_j, Y_i X_j, X_i X_j)$ is a linear transformation, use the fact that $e= Y+ X\beta$. Then, $\frac{1}{\sqrt{K}} \sum_i \sum_{j\ne i} P_{ij} ((e_i+X_i\beta) (e_j+X_j\beta), (e_i+X_i\beta) X_j, X_i X_j) = (T_{AR} + 2 T_{LM} \beta + T_{FS}\beta^2, T_{LM} - T_{FS}\beta, T_{FS})$.}
In the remainder of this section, I focus on testing the null that $\beta_0 =0$ so $e\left( \beta_0 \right)= Y$ and the L1O of the maximal invariant is exactly $(T_{AR}, T_{LM}, T_{FS})$.
The results are generalized in the appendix.

The asymptotic problem involving $(T_{AR}, T_{LM}, T_{FS})$ is: 
\begin{equation} \label{eqn:asymp_prob_e}
\left(\begin{array}{c}
T_{AR}\\
T_{LM}\\
T_{FS}
\end{array}\right)\sim N\left(\mu,\Sigma\right),\mu=\left(\begin{array}{c}
\frac{1}{\sqrt{K}}\sum_{i}\sum_{j\ne i}P_{ij}R_{\Delta i}R_{\Delta j}\\
\frac{1}{\sqrt{K}}\sum_{i}\sum_{j\ne i}P_{ij}R_{\Delta i}R_{j}\\
\frac{1}{\sqrt{K}}\sum_{i}\sum_{j\ne i}P_{ij}R_{i}R_{j}
\end{array}\right),\Sigma=\left(\begin{array}{ccc}
\sigma_{11} & \sigma_{12} & \sigma_{13}\\
\cdot & \sigma_{22} & \sigma_{23}\\
\cdot & \cdot & \sigma_{33}
\end{array}\right).
\end{equation}
While $\mu_2 =0$ under the null, $\mu_2$ may not be zero under the alternative.
There are several restrictions in the $\mu$ vector, which is assumed to be finite.
Since $P$ is a projection matrix, $\sum_i \sum_{j \ne i} P_{ij} R_{i} R_j = \sum_i R_i (\sum_{j } P_{ij}  R_j - P_{ii} R_i) = \sum_i  M_{ii} R_i^2$. 
Since the annihilator matrix $M$ has positive entries on its diagonal, we obtain $\mu_3 \geq 0$ and a similar argument yields $\mu_1 \geq 0$.
With $\mu_2 = \sum_i \sum_{j \ne i} P_{ij} R_{\Delta i} R_j = \sum_i M_{ii} R_{\Delta i} R_i$, the Cauchy-Schwarz inequality implies $\mu_2^2 \leq \mu_1 \mu_3$.
Constant treatment effects implies $\mu_2^2 = \mu_1 \mu_3$, which is a special case of the environment here.
Even with covariates, if the regression is fully saturated with $G$ given by UJIVE, the same inequality restrictions hold.\footnote{
See \Cref{prop:UJIVE_restr} in Online \Cref{sec:further_analytics_power}.
} 
These properties do not contradict the joint normality: even though $\mu_3\geq 0$, $T_{FS}$ can still be negative when using the L1O statistic. 
The inequalities $\mu_1, \mu_3 \geq0$ and $\mu_2^2 \leq \mu_1 \mu_3$ are also the \emph{only} restrictions on $\mu$, as it can be shown that there exists a structural model where there are no further restrictions.\footnote{
Online \Cref{sec:details_power} establishes that there exists a structural model where $\Sigma$ is uninformative about $\mu$, and $\mu_1, \mu_3 \geq0$. 
Since the model in \Cref{sec:challenges} is binary, it is insufficient for such a general result, and a continuous $X$ is required. 
While the result establishes that there exists a structural model where there are no further restrictions, for any given structural model, there can still be further restrictions.
}

\subsection{Optimality Result}


With a size $\alpha$ test, the two-sided LM test against the alternative that $\mu_2 \ne 0$ rejects when $T_{LM}^2/\Var(T_{LM}) > \Phi(1-\alpha/2)^2$.
I consider the benchmark of a uniformly most powerful unbiased test (e.g., \citet{lehmann2005testing, moreira2009tests}).
\begin{proposition} \label{prop:umpu}
Consider a restriction of the alternative $\mu$ space to the interior i.e., $\mu_1, \mu_3 >0$ and $\mu_2^2 < \mu_1 \mu_3$. 
Then, within the class of tests that are functions of $(T_{AR}, T_{LM}, T_{FS})$, the two-sided LM test is the uniformly most powerful unbiased test for testing $H_0: \mu_2 =0$ against $H_1: \mu_2 \ne 0$ in the asymptotic problem of (\ref{eqn:asymp_prob_e}).
\end{proposition}

The argument for optimality applies a standard optimality result from \citet{lehmann2005testing} on the exponential family, which includes the normal distribution.
To apply the \citet{lehmann2005testing} result, we require a convex parameter space and the the existence of alternative values above and below the null value.
It can be verified that the restricted parameter space is still convex, and the restriction to the interior ensures the latter condition is satisfied.
The proposition claims optimality within the class of unbiased tests, and makes no statement about tests that are biased (i.e., where the power somewhere in the alternative space can be lower than the size).

\begin{remark}
Beyond the two-sided UMPU result, we may also consider other power properties. 
The one-sided LM test is shown to be the most powerful test against a particular subset of the alternative space.
Numerically, using a covariance matrix calibrated from an empirical application, the power of the two-sided LM test is also close to that of the nearly optimal test against a weighted average over a grid of alternative values, constructed using the algorithm from \citet{elliott2015nearly}.
Details are in Online \Cref{sec:more_power_curves}.
\end{remark}

Studying optimality in the over-identified IV environment has thus far been complicated.
With constant treatment effects, both $s_1^\prime s_1$ and $s_1^\prime s_2$ are informative of the object of interest $\beta$, because constant treatment effects implies $\mu_1 = \beta^2 \mu_3$ in addition to $\mu_2 = \beta \mu_3$.
However, once we impose $\mu_1>0$ under the null that $\beta=0$, we rule out constant treatment effects by focusing on the interior of the alternative space.
Then, the statistic associated with $\mu_1$ is no longer directly informative of $\beta$.
Imposing heterogeneity is hence the key to obtaining this UMPU result. 



\section{Simulations} \label{sec:app_sim}
This section focuses on the simple example from \Cref{sec:challenges}.
I report two sets of simulations that assess the size and one that assesses power. 
One set of size simulations uses a large $K$ while the other a small $K$.
Robustness checks that involve different data generating processes are relegated to the online appendix.\footnote{There are more simulation results using several different structural models in \Cref{sec:further_sim}, including settings with continuous treatment $X$, and with covariates.
The results are qualitatively similar in those simulations, suggesting that the numerical findings are not unique to the data-generating process chosen.}

\Cref{tab:sim_mte_base} in \Cref{sec:challenges} reports rejection rates under the null for a relatively large number of judges with $K=400$, each with a small number of cases at $c=5$.
L3O performs well across various designs, while existing procedures can substantially over-reject in at least one design. 
The LMorc column is included as an infeasible theoretical benchmark that uses an oracle variance.
The difference between LMorc and L3O is attributed to the variance estimation procedure. 

\Cref{tab:sim_mte_fixedK} reports rejection rates under the null for a small number of judges with $K=4$ and a large number of cases at $c=200$. 
Based on the theory in \Cref{sec:valid_inference}, L3O should be valid when the instrument is strong, i.e., in the cases with $E[T_{FS}] = .5c$, which is what we observe.
Notably, even when $E[T_{FS}] = 2$ or $E[T_{FS}]=0$, the over-rejection for L3O is not too severe. 
EK performs very well in the cases with $E[T_{FS}]=.5c$ as expected in their theory.
In contrast, MS and MO can over-reject severely with strong heterogeneity, even when instruments are strong. 

\Cref{tab:sim_mte_power} reports rejection rates under the alternative. 
When $E[T_{FS}]=0$, the instrument should be completely uninformative about the true parameter, so we should have 0.05 rejection rate for a valid test, which is what we observe for L3O.
When $E[T_{FS}]=2\sqrt{K}$, all procedures, including L3O, are very informative.
Considering the designs with $E[T_{AR}]=0$ is most interesting, because this is an environment where MS and MO are valid, and the theoretical optimality result excludes this case. 
Looking at the case with $E[T_{AR}]=0, E[T_{FS}]=2$, L3O is less powerful than MS and MO in small samples, but the loss is less than 7 percentage points.

\begin{table}
    \centering
    \caption{Rejection rates under the null for nominal size 0.05 test}
    \label{tab:sim_mte_fixedK}
    \include{fig_v8/simf_example_fixedK_res}
    \justifying \small
    Notes: $K=4,c=200$. Designs and procedures are otherwise identical to \Cref{tab:sim_mte_base}.
\end{table}

\begin{table}
    \centering
    \caption{Rejection rates under the alternative for nominal size 0.05 test}
    \label{tab:sim_mte_power}
    \include{fig_v8/simf_example_power_res}
    \justifying \small
    Notes: $K=100,\beta=0.1,c=5$. Designs and procedures are otherwise identical to \Cref{tab:sim_mte_base}.
\end{table}

\section{Empirical Applications} \label{sec:empirical_applications}
\subsection{Returns to Education} \label{sec:returns_education}
\citet{angrist1991does} were interested in the impact of years of education (X) on log weekly wages (Y). 
They instrument for education using the quarter of birth (QOB).
I implement UJIVE using full interaction of QOB with the state of birth and year of birth (resulting in 1530 instruments) without other controls, which is similar to Table VII(2) of \citet{angrist1991does} that uses the same set of controls but without full saturation. 
The implementation here differs from the implementation of MS and MO in that I do not linearly partial out other covariates, but merely saturate on state and year of birth. 
This implementation is motivated by recent econometric research (e.g., \citet{blandhol2022tsls, sloczynski2020should}) that argue that the standard interpretation of estimands as a weighted average of LATE's is only retained with some parametric assumptions or when the specification controls for covariates richly, which can be achieved with full saturation.\footnote{
A further advantage of this implementation is that the code is fast: when $G$ is block-diagonal, it suffices to loop over blocks.
}
To ensure that the different procedures are directly comparable, I adapt the MS and MO inference procedures to target UJIVE, so that the estimand is the same across all procedures and differing results can be attributed purely to inference.
TSLS is consequently not a meaningful comparison as the estimand is different from the others.

The results are reported in \Cref{tab:ak91}. 
In addition to the aforementioned procedures, I include results from implementing the procedure in \citet{crudu2021inference} (CMS) that uses the $T_{AR}$ statistic like MS22, but uses a plug-in variance estimator like MO22.\footnote{
MS22 use a cross-fit variance estimator, while they refer to the CMS variance estimator as the ``naive" variance estimator.
MS22 argue that their cross-fit variance is more powerful, which corroborates how MS has a bounded confidence set while CMS does not.
}
Being robust to weak and many IV results in L3O having a longer confidence interval than EK. 
With full saturation, CMS and MO yield unbounded confidence sets, while L3O yields a bounded confidence set, showing how robustness to heterogeneity changes the shape of the confidence set in this context. 
The shape of the confidence set depends on the coefficient on $\beta_{0}^2$. 
In particular, for $\Psi_{2}:= \frac{1}{K}\sum_{i}\left(\sum_{j\ne i}G_{ij}X_{j}\right)^{2}X_{i}^{2}+\frac{1}{K}\sum_{i \ne j}^n G_{ij}^{2}X_{i}^{2}X_{j}^{2}$, MO is unbounded when $T_{FS}^{2}-q\Psi_{2}<0$ and L3O is unbounded when $T_{FS}^{2}-qB_{2}<0$, where $q$ is 3.84 for a 5\% test and $B_2$ is the coefficient on $\beta_0^2$ in the expression of $\hat{V}_{LM}$.
Consequently, in this application, we can think of $T_{FS}^2/ \Psi_2 = 0.102$ and $T_{FS}^2/ B_2 = 11.8$ as first-stage statistics for MO and L3O respectively that determine whether the confidence sets are bounded.\footnote{
These statistics are ``F" statistics with different variance estimators, suggesting that the instruments are meaningfully weak. 
The MS and MO variance estimators converge to the same object under weak identification such that $\frac{1}{K}\sum_{i} \sum_{j \ne i} G_{ij} R_i R_j \rightarrow 0$, which is not imposed by the asymptotic regime in this paper.  
}
Analogously, when solving a quartic equation in CMS, an unbounded set occurs as $T_{FS}^2/ \left( \frac{2}{K} \sum_{i \ne j}^n G_{ij}^{2}X_{i}^{2}X_{j}^{2} \right) = 0.0545$, where the denominator is their coefficient on $\beta_0^4$ in their variance estimator. 
In contrast, the MS confidence set is bounded with $T_{FS}^2/ \left( \frac{2}{K} \sum_{i \ne j}^n \frac{G_{ij}^{2}}{M_{ii} M_{jj} + M_{ij}^2} X_{i}^{2}X_{j}^{2} \right) = 23.9$.

Due to the $\check{M}$ terms in the L3O expression, it is difficult to compare the estimates directly. 
However, it is possible to compare the estimands of these coefficients in the judge example without covariates:
\begin{equation}
E\left[K\Psi_{2}\right]-E\left[B_{2}\right] = \sum_{i}M_{ii}R_{i}^{2}\left(R_{i}^{2}-3\left(1-2P_{ii}\right)E\left[\eta_{i}^{2}\right]\right).
\end{equation}
If $R_{i}^{2}>3\left(1-2P_{ii}\right)E\left[\eta_{i}^{2}\right]$, then MO is more likely unbounded. 
Intuitively, the MO estimator contains additional products of $R$ that are not present in the true variance, and there are products of $R$ and the error present in the true variance that MO does not account for, motivating the aforementioned difference. 
We can interpret this condition as MO being more likely unbounded when the signal-to-noise ratio $R_i^2/E\left[\eta_{i}^{2}\right]$ is sufficiently large.
In this application, by observing that the first-stage statistic of L3O is an order of magnitude larger than that of MO (i.e., $B_2$ is an order of magnitude smaller), and by comparing the CMS and MO first-stage statistics, there is evidence that $R_i^2/E\left[\eta_{i}^{2}\right]$ is large.\footnote{
Comparing the statistics between MO and MS implies $\frac{1}{K}\sum_{i}\left(\sum_{j\ne i}G_{ij}X_{j}\right)^{2}X_{i}^{2} < \frac{1}{K} \sum_{i} \sum_{j \ne i} G_{ij}^{2}X_{i}^{2}X_{j}^{2}$, which can equivalently be written as $\sum_i \sum_{j \ne i} \sum_{k \ne i,j} G_{ij} G_{ik} X_j X_k X_i^2 <0$.
With full saturation, observations $i,j$ have $G_{ij} <0 $ when they are in the same covariate group but have different instrument values.
Under the MS and MO asymptotic regimes where $\frac{1}{K} \sum_i \sum_{j \ne i} G_{ij} R_i R_j \rightarrow 0$ so $R_i^2/ E[\eta_i^2]$ is negligible, we obtain $\sum_{i} \sum_{j\ne i} G_{ij}^2 E[X_i^2] E[X_j^2]= \sum_{i} \sum_{j\ne i} G_{ij}^2  \left( R_i^2 + E[\eta_i^2] \right) \left( R_j^2 + E[\eta_j^2] \right) = \sum_{i} \sum_{j\ne i} G_{ij}^2 E[\eta_i^2] E[\eta_j^2] + o(1)$, and $\frac{1}{K}\sum_{i} \sum_{j\ne i} \sum_{k \ne i,j} G_{ij} G_{ik} E\left[ X_j X_k X_i^2 \right]  = \frac{1}{K}\sum_{i} \sum_{j\ne i} \sum_{k \ne i,j}  G_{ij} G_{ik} R_j R_k \left( R_i^2 + E[\eta_i^2] \right) = o(1)$ is asymptotically negligible.
Since the difference between MS and MO is the same magnitude as the MS statistic, $\frac{1}{K}\sum_{i} \sum_{j\ne i} \sum_{k \ne i,j} G_{ij} G_{ik} E\left[ X_j X_k X_i^2 \right]$ is of similar order as $\sum_{i} \sum_{j\ne i} G_{ij}^2 E[X_i^2] E[X_j^2]$, so $R_i^2 / E[\eta_i^2]$ is non-negligible in this application.
}
This result does not depend on heterogeneity, because the coefficient of $\beta_0^2$ depends only on how $X$ is combined in the variance estimator.

\begin{table}
\caption{95\% Confidence Sets for Returns to Education} \label{tab:ak91}
\centering
\begin{tabular}{lrrrrrrr}
\toprule
   & EK & CMS & MS & MO  & $\tilde{X}$-t & $\tilde{X}$-AR & L3O\\
\midrule
LB  & 0.033 & -$\infty$  & 0.019 &  -$\infty$ & 0.027 & 0.027 & 0.022\\
UB  & 0.173 & $\infty$  &  0.305 & $\infty$ & 0.179 & 0.189 & 0.210\\
Estimate  & 0.103 & 0.103 & 0.103& 0.103 & 0.103 & 0.103 & 0.103\\
CIlength & 0.140 & $\infty$ & 0.286 & $\infty$ & 0.152 & 0.161 & 0.188\\
\bottomrule
\end{tabular} \\
\justifying \small
Notes: Estimate reports the UJIVE. CMS implements the procedure from \citet{crudu2021inference}: use $T_{AR}$ with a plug-in variance.
Procedures are otherwise identical to \Cref{tab:sim_mte_base}.

\end{table}

\begin{table}
\caption{95\% Confidence Sets for Misdemeanor Prosecution} \label{tab:misdemeanor}
\centering
\begin{tabular}{lrrrrrrr}
\toprule
 & EK & CMS & MS & MO  & $\tilde{X}$-t & $\tilde{X}$-AR & L3O\\
\midrule
LB& -0.151 & $\emptyset$ & $\emptyset$ & -0.220 & -0.187 & -0.188 & -0.201\\
UB & -0.076 & $\emptyset$ & $\emptyset$ & -0.019 & -0.039 & -0.038 & -0.028\\
Estimate  & -0.113 & -0.113 & -0.113 & -0.113 & -0.113 & -0.113 & -0.113\\
CIlength  & 0.075 & $\emptyset$ & $\emptyset$ & 0.201 & 0.148 & 0.150 & 0.173\\
\bottomrule
\end{tabular} \\
\justifying \small
Notes: Procedures are identical to \Cref{tab:ak91}.
\end{table}

\subsection{Misdemeanor Prosecution}
\citet{agan2023misdemeanor} were interested in the effect of misdemeanor prosecution (X) on criminal complaint in two years (Y). 
They instrument for misdemeanor prosecution using the assistant district attorneys (ADAs) who decide if a case should be prosecuted in the Suffolk County District Attorney's Office in Massachusetts.
As \citet{agan2023misdemeanor} argued that as-if randomization holds conditional on court-by-time controls and that individual covariates are not required for relevance or exogeneity to hold in this context, the confidence set is constructed using full saturation of court-by-year and court-by-day-of-week fixed effects with no other controls for individual covariates. 

As reported in \Cref{tab:misdemeanor}, with full saturation, the UJIVE is $-0.11$, so not prosecuting decreases the probability of criminal involvement by 11 percentage points.\footnote{This result is smaller than $-0.36$ reported in their Table III(3) that uses TSLS with a leniency measure.
The result is more similar to the UJIVE robustness check in their Table A.1(5) of $-0.15$ with full saturation of the instrument, but their specification includes case/ defendant covariates, which results in a different estimator.
} 
The L3O confidence interval (CI) is more than twice that of EK: unlike \Cref{sec:returns_education} where $n/K = 221$, we have $n/K =11.9$ here, so a variance estimator that is robust to many IV has a larger impact on CI.
MS has an empty confidence set while L3O has a bounded set, showing how being robust to heterogeneity can change conclusions.\footnote{
An empty confidence set using the AR procedure also suggests that the model with constant treatment effects is rejected, so there is meaningful heterogeneity. 
Since the variances of MO and L3O converge to the same object under homogeneity, the difference between MO and L3O confidence sets also suggests that there is heterogeneity.
}
Mechanically, the confidence set for MS solves a quartic equation, so an empty set can occur, but it is difficult to characterize when this phenomenon occurs in general.

The L3O CI is also shorter than MO, so being robust to heterogeneity decreases the length of the CI. 
Considering how the length of the MO confidence set is longer than L3O while being oversized in simulations, there is a question of when MO is conservative. 
While it is difficult to compare the confidence intervals or variance estimators directly, it is possible to compare the null-imposed variance estimands in the judge example without covariates. 
It can be shown that:
\begin{align*}
E&\left[\hat{\Psi}_{MO}\right]-Var\left(\sum_{i}\sum_{j\ne i}P_{ij}e_{i}X_{j}\right) \\
 = &\sum_{i}M_{ii} R_{\Delta i}^{2} \left(R_{i}^{2} - (1- 2P_{ii}) E\left[\eta_{i}^{2}\right] \right) -2\sum_{i} M_{ii} (1- 2P_{ii}) E\left[\eta_{i}\nu_{i}\right]R_{i}R_{\Delta i}.
\end{align*}
Then, MO is conservative when: (i) $R_{i}^{2}> (1- 2P_{ii})E\left[\eta_{i}^{2}\right]$, and (ii) $E\left[\eta_{i}\nu_{i}\right]$ is negatively correlated with $R_{i}R_{\Delta i}$, when $P_{ii} < 1/2$. 
In (i), $R_{\Delta i}^2$ only affects the magnitude of the difference, and not the sign, so this condition can be interpreted as a condition on the signal-to-noise ratio as before.
Condition (ii) results from the $\sum_{i} M_{ii} (1- 2P_{ii}) E\left[\eta_{i}\nu_{i}\right]R_{i}R_{\Delta i}$ term that MO does not account for, and covariances can be positive or negative in general.

\section{Conclusion} \label{sec:conclusion}
This paper has documented how weak instruments and heterogeneity can interact to invalidate existing procedures in the environment of many instruments.
Addressing both problems simultaneously, this paper contributes a feasible and robust method for valid inference. 
The procedure is shown to be valid as the limiting distribution of commonly used statistics, including the LM statistic, in an environment with many weak instruments and heterogeneity, is normal, and a leave-three-out variance estimator is consistent for obtaining the variance of the LM statistic. 
Beyond its validity, the LM test is also optimal, as it is the uniformly most powerful unbiased test in the asymptotic distribution for the interior of the alternative space.
In light of the broader econometric literature on the value of saturated regressions and how many instruments can arise from them, this paper presents a highly applicable, robust, and powerful inference procedure for IV. 



\appendix



\section{High-level Assumptions for Inference} \label{sec:assumptions_inference}
Following AS23, to ease notation in the L3O derivations, I define:
\begin{align*}
    \check{M}_{il,-ijk} := \frac{M_{il} - M_{ij} \check{M}_{jl,-jk} - M_{ik} \check{M}_{kl,-jk}}{D_{ijk}/ D_{jk}},
\end{align*}
so that $X_i - Q_i^\prime \hat{\tau}_{-ijk} = \sum_{l \ne k} \check{M}_{il,-ijk} X_l$, for instance. 

\Cref{asmp:Vconst_reg} below states high-level conditions for consistency of the variance estimator. To ease notation, let $R_{mi}$ stand for either $R_{\Delta i}$ or $R_{i}$. Denote $\tilde{R}_i := \sum_{j\ne i} G_{ij} R_j$ and $\tilde{R}_{\Delta i} := \sum_{j \ne i} G_{ij} R_{\Delta j}$. 
Let $h_{2}\left(i,j\right)$ be a product of any number of $G_{i_{1}i_{2}},i_{1}\ne i_{2}$, and $\check{M}_{j_{1}j_{2}},j_{1}\ne j_{2}$ with $i_{1},i_{2},j_{1},j_{2}\in\left\{ i,j\right\}$. 
Similarly, $h_{3}\left(i,j\right)$ denotes a product of any number of $G_{i_{1}i_{2}},i_{1}\ne i_{2}$, and $\check{M}_{j_{1}j_{2}},j_{1}\ne j_{2}$ with $i_{1},i_{2},j_{1},j_{2}\in\left\{ i,j, k\right\}$ such that every index in $\left\{ i,j,k \right\}$ occurs at least once as an index of either $G_{i_{1}i_{2}}$ or $\check{M}_{j_{1}j_{2}}$.
Let $h_{4}\left(i,j,k,l\right)$ denote a product of any number of $G_{i_{1}i_{2}},i_{1}\ne i_{2}$ and $\check{M}_{j_{1}j_{2}},j_{1}\ne j_{2}$ with $i_{1},i_{2},j_{1},j_{2}\in\left\{ i,j,k,l\right\}$ such that every index in $\left\{ i,j,k,l\right\}$ occurs at least once as an index of either $G_{i_{1}i_{2}}$ or $\check{M}_{j_{1}j_{2}}$, and there is no partition such that $h_4(i_1,i_2,j_1,j_2) = h_2(i_1,i_2) h_2(j_1,j_2)$, where $i_1,i_2,j_1,j_2$ are all different indices.
For instance, $h_{4}(i,j,k,l)$ could be $G_{ij}\check{M}_{ik,-il}\check{M}_{lj,-ijk}$ but not $G_{ij}\check{M}_{lk,-il}$.
Let $\sum_{i\ne j}^{n}=\sum_{i}\sum_{j\ne i}$ so that sums without the $n$ superscript are still sums of individual indices, but sums with an $n$ superscript involves the sum over multiple indices. 
Objects like $\sum_{i\ne j\ne k}^{n}$ and $\sum_{i\ne j\ne k\ne l}^{n}$ are defined similarly. 
When I refer to the p-sum, I refer to the sum over p non-overlapping indices. 
For instance, a 3-sum is $\sum_{i\ne j\ne k}^{n}$. 
Let $F$ stand for either $G$ or $G^{\prime}$. 
$1\{ \cdot \}$ is an indicator function that takes the value 1 if the argument is true and 0 otherwise.
$I\left\{ \cdot\right\} $ is a function that takes value 1 if the argument is true and -1 if false. 
\begin{assumption} \label{asmp:Vconst_reg}
For some $C<\infty$, 
\begin{enumerate} [topsep=0pt,label=(\alph*)]
\item $\sum_{j}F_{ij}^{2}\leq C$, $\sum_{j\ne k}^{n}\left(\sum_{i\ne j,k}G_{ij}F_{ik}\right)^{2}\leq C\sum_{j\ne k}^{n}G_{jk}^{2}$, $\sum_{j\ne k}^{n}\left(\sum_{i\ne j,k}G_{ji}G_{ki}\right)^{2}\leq C\sum_{j\ne k}^{n}G_{jk}^{2}$, and
$|R_{mi}|\leq C$.
\item $\sum_{i\ne j\ne k}^{n}\left(\sum_{l\ne i,j,k}h_{4}\left(i,j,k,l\right)R_{ml}\right)^{2}\leq C\sum_{i}\tilde{R}_{mi}^{2}$,
$\sum_{i\ne j}^{n}\left(\sum_{k\ne i,j}\sum_{l\ne i,j,k}h_{4}\left(i,j,k,l\right)R_{ml}\right)^{2}\leq C\sum_{i}\tilde{R}_{mi}^{2}$,
and $\sum_{i}\left(\sum_{j\ne i}\sum_{k\ne i,j}\sum_{l\ne i,j,k}h_{4}\left(i,j,k,l\right)R_{ml}\right)^{2}\leq C\sum_{i}\tilde{R}_{mi}^{2}.$
\item $\sum_{i\ne j}^{n}\left(\sum_{k\ne i,j}h_{3}\left(i,j,k\right)R_{mk}\right)^{2}\leq C\sum_{i}\tilde{R}_{mi}^{2}$
and $\sum_{i}\left(\sum_{j\ne i}\sum_{k\ne i,j}h_{3}\left(i,j,k\right)R_{mk}\right)^{2}\leq C\sum_{i}\tilde{R}_{mi}^{2}$.
\item $\sum_{i}\left(\sum_{j\ne i}h_{2}\left(i,j\right)R_{mj}\right)^{2}\leq C\sum_{i}\tilde{R}_{mi}^{2}$. 
\end{enumerate}
\end{assumption}

The first condition requires the row and column sums of the squares of the $G$ elements to be bounded. \Cref{asmp:normality}(e) is insufficient because it does not rule out having $G_{ii}=K$ for some $i$ and 0 elsewhere in the $G$ matrix. 
These remaining conditions can be interpreted as (approximate) sparsity conditions on $M$ and $G$ as the p-sum of entries of $\check{M}$ and $G$ cannot be too large. 
The conditions primarily place a restriction on the types of $G$ that can be used: for instance, a $G$ matrix that contains all 1's is excluded.
Note that other elements of the covariance matrix can be analogously shown to be consistent using the same strategy by using the lemmas from \Cref{sec:main_text_proofs} by using $\tilde{R}_{Yi} := \sum_{j \ne i} G_{ij} R_{Yj}$ instead of $\tilde{R}_{\Delta i}$ where required. 

The judges example in \Cref{sec:challenges} satisfies this assumption when there are no covariates, $G=P$, and $R$ values are bounded. 
For condition (a), $\sum_{j}P_{ij}^{2}=P_{ii}\leq C$ and, since $P$ is idempotent, $\sum_{j\ne k}^{n}\left(\sum_{i\ne j,k}P_{ij}P_{ik}\right)^{2}=\sum_{j\ne k}^{n}\left(\sum_{i}P_{ij}P_{ik}-P_{jj}P_{jk}-P_{kk}P_{jk}\right)^{2}=\sum_{j\ne k}^{n}\left(P_{jk}-P_{jj}P_{jk}-P_{kk}P_{jk}\right)^{2}=\sum_{j\ne k}^{n}\left(1-P_{jj}-P_{kk}\right)^{2}P_{jk}^{2}\leq\sum_{j\ne k}^{n}P_{jk}^{2}$.
For any $\check{M}_{ij}$ and $G_{ij}$, these elements are nonzero only when $i$ and $j$ share the same judge $p$. 
Further, $R_{mi}=\pi_{mp(i)}$, where $\pi_{mp}$ can denote $\pi_{p}$ or $\pi_{\Delta p}$ in the model. 
Due to how the $h$ functions are defined, when every judge has at most $c$ cases, 
\begin{align*}
\sum_{i}&\left(\sum_{j\ne i}h_{2}\left(i,j\right)R_{mj}\right)^{2}  =\sum_{i}\left(\sum_{j\in\mathcal{N}_{p(i)}\backslash\{i\}}h_{2}\left(i,j\right)R_{mp(i)}\right)^{2}=\sum_{p}\sum_{i\in\mathcal{N}_{p}}\left(\sum_{j\in\mathcal{N}_{p}\backslash\{i\}}h_{2}\left(i,j\right)\pi_{mp}\right)^{2}\\
 & =\sum_{p}\sum_{i\in\mathcal{N}_{p}}\left(\sum_{j\in\mathcal{N}_{p}\backslash\{i\}}h_{2}\left(i,j\right)\pi_{mp}\right)^{2}\pi_{mp}^{2}\leq C\sum_{p}\sum_{i\in\mathcal{N}_{p}}\left(c-1\right)^{2}\pi_{mp}^{2}=C\sum_{i}\tilde{R}_{mi}^{2}.
\end{align*}

The same argument applies to the other components. 
For instance, 
\begin{align*}
 & \sum_{i}\left(\sum_{j\ne i}\sum_{k\ne i,j}\sum_{l\ne i,j,k}h_{4}\left(i,j,k,l\right)R_{ml}\right)^{2}
 =\sum_{p}\pi_{mp}^{2}\sum_{i\in\mathcal{N}_{p}}\left(\sum_{j\in\mathcal{N}_{p}\backslash\{i\}}\sum_{k\in\mathcal{N}_{p}\backslash\{i,j\}}\sum_{l\in\mathcal{N}_{p}\backslash\{i,j,k\}}h_{4}\left(i,j,k,l\right)\right)^{2}\\
 & \leq C\sum_{p}\sum_{i\in\mathcal{N}_{p}}\pi_{mp}^{2}\left(c-1\right)^{2}\left(c-2\right)^{2}\left(c-3\right)^{2}\leq C\sum_{i}\tilde{R}_{mi}^{2}.
\end{align*}
The upper bound is fairly loose because it merely counts the number of nonzero entries in $h_4$. 
When every judge has a large number of cases, since $h_4$ contains only entries from the projection matrix, the inner sum is still bounded and the assumption is satisfied.

\section{Supplement for Section 2} \label{sec:main_supp_sect2}

\begin{lemma} \label{lem:MS_estimand}
Consider the model of \Cref{sec:challenges}.
Suppose $h\ne0$ and $Ks^2>0$.
Then, $E\left[T_{AR}\right]\ne0$ for all real $\beta_{0}$.
\end{lemma}

\textbf{MS Variance Estimand.}
The proposed variance estimator is:
\begin{align*}
\hat{\Phi}_{MS} := \frac{2}{K} \sum_{i} \sum_{j \ne i} \frac{P_{ij}^2}{M_{ii} M_{jj} + M_{ij}^2} (e_i M_{i \cdot} e) (e_j M_{j \cdot} e).
\end{align*}
By substituting $e_i = R_{\Delta i} + \nu_i$ and taking expectations,
\begin{align*}
E[e_i M_{i \cdot} e e_j M_{j \cdot} e] &= R_{\Delta i} R_{\Delta j} \left(\sum_k M_{ik} M_{jk} E[\nu_k^2] \right) + (M_{ii} M_{jj} + M_{ij}^2) \left(E[\nu_i^2] E[\nu_j^2] \right).
\end{align*}
This estimand can be positive or negative, but observe that $\sum_i \sum_{j \ne i}  R_{\Delta i} R_{\Delta j} = \left( \sum_i R_{\Delta i} \right)^2 - \sum_i R_{\Delta i}^2 = - \sum_i R_{\Delta i}^2 = - (n-c)h^2$ in the model of \Cref{sec:simple_setting}. 
Consequently, the negative heterogeneity component can far outweigh the positive components, resulting in a negative estimand when $h$ does not converge to 0.


\section{Main Proofs} \label{sec:main_text_proofs}

A quadratic CLT is used for \Cref{thm:normality}. Let
\[
T=\sum_{i}s_{i}^{\prime}v_{i}+\sum_{i}\sum_{j\ne i}G_{ij}v_{i}^{\prime}Av_{j},
\]
where $v_{i}$ is a finite-dimensional random vector independent over $i=1,\dots,n$ with bounded 4th moments, $s_{i}$ is a nonstochastic vector, and $A$ is a conformable matrix. 

\begin{lemma} \label{lem:CLT}
Suppose:
\begin{enumerate}
\item $\Var\left(T\right)^{-1/2}$ is bounded;
\item $\sum_{i}s_{il}^{4}\rightarrow0$; and 
\item $||G_{L}G_{L}^{\prime}||_{F}+||G_{U}G_{U}^{\prime}||_{F}\rightarrow0$,
where $G_{L}$ is a lower-triangular matrix with elements $G_{L,ij}=G_{ij}1\left\{ i>j\right\} $
and $G_{U}$ is an upper-triangular matrix with elements $G_{U,ij}=G_{ij}1\left\{ i<j\right\}$.
\end{enumerate}
Then, $\Var\left(T\right)^{-1/2}T\xrightarrow{d}N(0,1)$.
\end{lemma}

\begin{proof}[Proof of \Cref{thm:normality}]
By substituting $Y_i = R_{Yi}+ \zeta_i$, $X_i = R_i + \eta_i$, $\zeta_{i}=\nu_{i}+\beta\eta_{i}$ and $R_{Yi}=R_{\Delta i}-R_{i}\beta$ into the expression for $\hat{\beta}_{JIVE}$, 
\begin{align*}
\hat{\beta}_{JIVE}-\beta_{JIVE} & =\frac{\left(\sum_{i}\sum_{j\ne i}G_{ij}\left(R_{\Delta i}\eta_{j}+ \nu_{i}R_{j}+\nu_{i}\eta_{j}\right)\right)}{\sum_{i}\sum_{j\ne i}G_{ij}R_{i}R_{j}+\sum_{i}\sum_{j\ne i}G_{ij}\left(R_{i}\eta_{j}+R_{j}\eta_{i}+\eta_{i}\eta_{j}\right)}.
\end{align*}

To see the equivalence with the $T$ objects,
\begin{align*}
\frac{1}{\sqrt{K}}\sum_{i}\sum_{j\ne i}G_{ij}e_{i}X_{j} & =\frac{1}{\sqrt{K}}\sum_{i}\sum_{j\ne i}G_{ij}\left(\nu_{i}R_{j}+\nu_{i}\eta_{j}+R_{\Delta i}R_{j}+R_{\Delta i}\eta_{j}\right), \text{ and } \\
\sum_{i}\sum_{j\ne i}G_{ij}R_{\Delta i}R_{j} & =\sum_{i}\sum_{j\ne i}G_{ij}R_{Yi}R_{j}-\sum_{i}\sum_{j\ne i}G_{ij}R_{i}R_{j}\left(\frac{\sum_{i}\sum_{j\ne i}G_{ij}R_{Yi}R_{j}}{\sum_{i}\sum_{j\ne i}G_{ij}R_{i}R_{j}}\right)=0,
\end{align*}
while $T_{FS}$ is immediate. 

Next, I show that the joint distribution of $\sqrt{\frac{K}{r_n}} \left(T_{AR},T_{LM},T_{FS} \right)$ is asymptotically normal. 
Using the Cramer-Wold device, it suffices to show that $\sqrt{\frac{K}{r_n}} (c_{1}T_{AR}+c_{2}T_{LM}+c_{3}T_{FS})$ is normal for fixed $c$'s, where
\begin{align*}
 & \sqrt{\frac{K}{r_n}} (c_{1}T_{AR}+c_{2}T_{LM}+c_{3}T_{FS})
 =c_{1}\frac{1}{\sqrt{r_n}}\sum_{i}\sum_{j\ne i}G_{ij}\left(\nu_{i}R_{j}+\nu_{i}\nu_{j}+R_{\Delta i}R_{\Delta j}+R_{\Delta i}\nu_{j}\right)\\
 & \quad+c_{2}\frac{1}{\sqrt{r_n}}\sum_{i}\sum_{j\ne i}G_{ij}\left(\nu_{i}R_{j}+\nu_{i}\eta_{j}+R_{\Delta i}\eta_{j}\right)
 +c_{3}\frac{1}{\sqrt{r_n}}\sum_{i}\sum_{j\ne i}G_{ij}\left(\eta_{i}R_{j}+\eta_{i}\eta_{j}+R_{i}R_{j}+R_{i}\eta_{j}\right).
\end{align*}

The object $T=\sqrt{\frac{K}{r_n}} (c_{1}T_{AR}+c_{2}T_{LM}+c_{3}T_{FS})-c_{1}\frac{1}{\sqrt{r_n}}\sum_{i}\sum_{j\ne i}G_{ij}R_{\Delta i}R_{\Delta j}-c_{3}\frac{1}{\sqrt{r_n}}\sum_{i}\sum_{j\ne i}G_{ij}R_{i}R_{j}$ can be written in the CLT form by setting:
\begin{align*}
v_{i} & =\left(\eta_{i},\nu_{i}\right)^{\prime}, A =\left[\begin{array}{cc}
c_{3} & 0\\
c_{2} & c_{1}
\end{array}\right], \text{ and } \\
s_{i} & =\left[\begin{array}{c}
c_{3}\sum_{j\ne i}\left(G_{ij}+G_{ji}\right)R_{j}+c_{2}\sum_{j\ne i}G_{ji}R_{\Delta j}\\
c_{1}\sum_{j\ne i}\left(G_{ij}+G_{ji}\right)R_{\Delta j}+c_{2}\sum_{j\ne i}G_{ij}R_{j}
\end{array}\right], 
\end{align*}
so that
\[
T=\frac{1}{\sqrt{r_n}}\sum_{i}s_{i}^{\prime}v_{i}+\frac{1}{\sqrt{r_n}}\sum_{i}\sum_{j\ne i}G_{ij}v_{i}^{\prime}Av_{j}.
\]

Bounded 4th moments hold by \Cref{asmp:normality}(a). 
To apply the CLT from \Cref{lem:CLT}, I verify the following:
\begin{enumerate}
\item $\Var\left(T\right)^{-1/2}$ is bounded;
\item $\frac{1}{r_n^2}\sum_{i}s_{il}^{4}\rightarrow0$ for all $l$; and 
\item $||G_{L}G_{L}^{\prime}||_{F}+||G_{U}G_{U}^{\prime}||_{F}\rightarrow0$,
where $G_{L}$ is a lower-triangular matrix with elements $G_{L,ij}=\frac{1}{\sqrt{r_n}}G_{ij}1\left\{ i>j\right\} $
and $G_{U}$ is an upper-triangular matrix with elements $G_{U,ij}=\frac{1}{\sqrt{r_n}}G_{ij}1\left\{ i<j\right\} $.
\end{enumerate}
Condition (2) follows from \Cref{asmp:normality}(d) and applying the Cauchy-Schwarz inequality. 
Condition (3) follows from \Cref{asmp:normality}(e). 
For Condition (1), I show that \Cref{asmp:normality}(b) and (c) imply that, for any nonstochastic scalars $c_1, c_2, c_3$ that are finite and not all 0, $\Var(T)^{-1/2}$ is bounded.
Since $Cov\left(\sum_{i}s_{i}^{\prime}v_{i},\sum_{i}\sum_{j\ne i}G_{ij}v_{i}^{\prime}Av_{j}\right)=0$,
\begin{equation}
\Var\left(T\right) = \frac{1}{r_n} \Var\left(\sum_{i}s_{i}^{\prime}v_{i}\right)+\frac{1}{r_n} 
 \Var\left(\sum_{i}\sum_{j\ne i}G_{ij}v_{i}^{\prime}Av_{j}\right),
\end{equation}
so it suffices to show that either term is bounded below.

The first term can be expanded as follows:
\begin{align*}
\Var \left( \sum_i s_i^\prime v_i \right) &= \sum_i s_i^\prime \Var\left(  v_i \right) s_i 
= \sum_is_{i1}^2 E[\eta_i]^2+ 2 s_{i1} s_{i2} E[\eta_i \nu_i] + s_{i2}^2 E[\nu_i^2] \\
&= \sum_i(1-\rho_i)^2 E[\eta_i^2] s_{i1}^2 + \left( \rho_i s_{i1} \sqrt{E[\eta_i^2]} + s_{i2} \sqrt{E[\nu_i^2]} \right)^2 \geq \sum_i (1-\rho_i)^2 E[\eta_i^2] s_{i1}^2.
\end{align*}
A similar argument yields $\Var \left( \sum_i s_i^\prime v_i \right) \geq \sum_i (1-\rho_i)^2 E[\eta_i^2] s_{i2}^2$.
Due to \Cref{asmp:normality}(c), at least one of the following must hold: (i) $\Var\left(\sum_{i}\sum_{j\ne i}G_{ij}v_{i}^{\prime}Av_{j}\right) \geq \underline{c}$ (ii) $\frac{1}{r_n} \sum_i s_{i1}^2 \geq \underline{c}$, or (iii) $\frac{1}{r_n} \sum_i s_{i2}^2 \geq \underline{c}$. 
Hence, $\Var (T)^{-1/2}$ is bounded.

Finally, since $\nu_i,\eta_i$ are mean zero, the expectations are immediate: $E\left[T_{AR}\right]  =\sum_{i}\sum_{j\ne i}G_{ij}R_{\Delta j}R_{\Delta i}$ and $E\left[T_{FS}\right]=\sum_{i}\sum_{j\ne i}G_{ij}R_{j}R_{i}$. 
\end{proof}

The proof of \Cref{thm:var_consistent} involves several lemmas whose proofs are relegated to the online appendix.
The proof strategy is to bound the variances above by components that are in the $h(.)$ form so that \Cref{asmp:Vconst_reg} inequalities can be applied. 

Let $V_{mi}=R_{mi}+v_{mi}$ where $R_{mi}$ denotes the nonstochastic component while $v_{mi}$ denotes the mean zero stochastic component.
Following \Cref{eqn:rn}, $r_{n}:=\sum_{i}\tilde{R}_{i}^{2}+\sum_{i}\tilde{R}_{\Delta i}^{2}+\sum_{i}\sum_{j\ne i}G_{ij}^{2}$.
Let $C_{i},C_{ij},C_{ijk}$ denote nonstochastic objects that are non-negative and are bounded above by $C$. 
I use $h_{4}^{A}(.)$ and $h_{4}^{B}(.)$ to denote two different functions that satisfy
the definition for $h_{4}$. 

\begin{lemma} \label{lem:a0}
Under \Cref{asmp:Vconst_reg}, the following hold:
\begin{enumerate} [topsep=0pt,label=(\alph*)]
\item $\left|\sum_{i\ne j\ne k}^{n}C_{ijk}\left(\sum_{l\ne i,j,k}h_{4}^{A}\left(i,j,k,l\right)R_{ml}\right)\left(\sum_{l\ne i,j,k}h_{4}^{B}\left(i,j,k,l\right)R_{ml}\right)\right|\leq C\sum_{i}\tilde{R}_{mi}^{2}$, \\
$\left|\sum_{i\ne j}^{n}C_{ij}\left(\sum_{k\ne i,j}\sum_{l\ne i,j,k}h_{4}^{A}\left(i,j,k,l\right)R_{ml}\right)\left(\sum_{k\ne i,j}\sum_{l\ne i,j,k}h_{4}^{B}\left(i,j,k,l\right)R_{ml}\right)\right|\leq C\sum_{i}\tilde{R}_{mi}^{2}$, and \\ 
$\left|\sum_{i}C_{i}\left(\sum_{j\ne i}\sum_{k\ne i,j}\sum_{l\ne i,j,k}h_{4}^{A}\left(i,j,k,l\right)R_{ml}\right)\left(\sum_{j\ne i}\sum_{k\ne i,j}\sum_{l\ne i,j,k}h_{4}^{B}\left(i,j,k,l\right)R_{ml}\right)\right|\leq C\sum_{i}\tilde{R}_{mi}^{2}.$
\item $\left|\sum_{i\ne j}^{n}C_{ij}\left(\sum_{k\ne i,j}h_{3}^{A}\left(i,j,k\right)R_{mk}\right)\left(\sum_{k\ne i,j}h_{3}^{B}\left(i,j,k\right)R_{mk}\right)\right|\leq C\sum_{i}\tilde{R}_{mi}^{2}$ \\
and $\left|\sum_{i}C_{i}\left(\sum_{j\ne i}\sum_{k\ne i,j}h_{3}^{A}\left(i,j,k\right)R_{mk}\right)\left(\sum_{j\ne i}\sum_{k\ne i,j}h_{3}^{B}\left(i,j,k\right)R_{mk}\right)\right|\leq C\sum_{i}\tilde{R}_{mi}^{2}$.
\item $\left|\sum_{i}C_{i}\left(\sum_{j\ne i}h_{2}^{A}\left(i,j\right)R_{mj}\right)\left(\sum_{j\ne i}h_{2}^{B}\left(i,j\right)R_{mj}\right)\right|\leq C\sum_{i}\tilde{R}_{mi}^{2}$.
\end{enumerate}
\end{lemma}

\begin{lemma} \label{lem:a1}
Under \Cref{asmp:Vconst_reg}, the following hold:
\begin{enumerate} [topsep=0pt,label=(\alph*)]
\item $\Var\left(\sum_{i\ne j}^{n}G_{ij}F_{ij}V_{1i}V_{2i}V_{3j}V_{4j}\right)\leq Cr_{n}$.
\item $\Var\left(\sum_{i\ne j\ne k}^{n}G_{ij}F_{ij}\check{M}_{ik,-ij}V_{1i}V_{2k}V_{3j}V_{4j}\right)\leq Cr_{n}$.
\item $\Var\left(\sum_{i\ne j\ne l}^{n}G_{ij}F_{ij}\check{M}_{jl,-ij}V_{1i}V_{2i}V_{3j}V_{4l}\right)\leq Cr_{n}$.
\item $\Var\left(\sum_{i\ne j\ne k\ne l}^{n}G_{ij}F_{ij}V_{1i}\check{M}_{ik,-ij}V_{2k}V_{3j}\check{M}_{jl,-ijk}V_{4l}\right)\leq Cr_{n}$.
\end{enumerate}
\end{lemma}

\begin{lemma} \label{lem:a2}
Under \Cref{asmp:Vconst_reg}, the following hold:
\begin{enumerate} [topsep=0pt,label=(\alph*)]
\item $\Var\left(\sum_{i\ne j\ne k}^{n}G_{ij}F_{ik}V_{1j}V_{2k}V_{3i}V_{4i}\right)\leq Cr_{n}$.
\item $\Var\left(\sum_{i\ne j\ne k\ne l}^{n}G_{ij}F_{ik}\check{M}_{il,-ijk}V_{1j}V_{2k}V_{3i}V_{4l}\right)\leq Cr_{n}$.
\end{enumerate}
\end{lemma}

\begin{lemma} \label{lem:a3}
Under \Cref{asmp:Vconst_reg}, the following hold:
\begin{enumerate} [topsep=0pt,label=(\alph*)]
\item $\Var\left(\sum_{i\ne j}^{n}G_{ji}^{2}V_{1i}V_{2i}V_{3j}V_{4j}\right)\leq Cr_{n}$;
\item $\Var\left(\sum_{i\ne j\ne k}^{n}G_{ji}^{2}\check{M}_{ik,-ij}V_{1i}V_{2k}V_{3j}V_{4j}\right)\leq Cr_{n}$;
\item $\Var\left(\sum_{i\ne j\ne l}^{n}G_{ji}^{2}\check{M}_{jl,-ij}V_{1i}V_{2i}V_{3j}V_{4l}\right)\leq Cr_{n}$;
\item $\Var\left(\sum_{i\ne j\ne k\ne l}^{n}G_{ji}^{2}V_{1i}\check{M}_{ik,-ij}V_{2k}V_{3j}\check{M}_{jl,-ijk}V_{4l}\right)\leq Cr_{n}$;
\item $\Var\left(\sum_{i\ne j\ne k}^{n}G_{ji}F_{ki}V_{1j}V_{2k}V_{3i}V_{4i}\right)\leq Cr_{n}$;
\item $\Var\left(\sum_{i\ne j\ne k\ne l}^{n}G_{ji}F_{ki}\check{M}_{il,-ijk}V_{1j}V_{2k}V_{3i}V_{4l}\right)\leq Cr_{n}$.
\end{enumerate}
\end{lemma}

\begin{proof} [Proof of \Cref{thm:var_consistent}]
\textbf{Proof of Unbiasedness.}
The variance expression is:
\small
\begin{equation} \label{eqn:VLM_expr_proof}
\begin{split}
V_{LM} & =\sum_{i}\left(E\left[\nu_{i}^{2}\right]\left(\sum_{j\ne i}G_{ij}R_{j}\right)^{2}+2\left(\sum_{j\ne i}G_{ij}R_{j}\right)\left(\sum_{j\ne i}G_{ji}R_{\Delta j}\right)E\left[\nu_{i}\eta_{i}\right]+E\left[\eta_{i}^{2}\right]\left(\sum_{j\ne i}G_{ji}R_{\Delta j}\right)^{2}\right)\\
 & \quad+\sum_{i}\sum_{j\ne i}G_{ij}^{2}E\left[\nu_{i}^{2}\right]E\left[\eta_{j}^{2}\right]+\sum_{i}\sum_{j\ne i}G_{ij}G_{ji}E\left[\eta_{i}\nu_{i}\right]E\left[\eta_{j}\nu_{j}\right].
\end{split}
\end{equation}
\normalsize

To ease notation, let:
\begin{align*}
A_{1i} & :=\sum_{j\ne i}\sum_{k\ne i}G_{ij}X_{j}G_{ik}X_{k}e_{i}\left(e_{i}-Q_{i}^{\prime}\hat{\tau}_{\Delta,-ijk}\right), 
&A_{2i}  :=\sum_{j\ne i}\sum_{k\ne i}G_{ij}X_{j}G_{ki}e_{k}e_{i}\left(X_{i}-Q_{i}^{\prime}\hat{\tau}_{-ijk}\right),\\
A_{3i} & :=\sum_{j\ne i}\sum_{k\ne i}G_{ji}e_{j}G_{ki}e_{k}X_{i}\left(X_{i}-Q_{i}^{\prime}\hat{\tau}_{-ijk}\right),
&A_{4ij}  :=X_{i}\sum_{k\ne j}\check{M}_{ik,-ij}X_{k}e_{j}\left(e_{j}-Q_{j}^{\prime}\hat{\tau}_{\Delta,-ijk}\right), \text{ and }\\
A_{5ij} & :=e_{i}\sum_{k\ne j}\check{M}_{ik,-ij}X_{k}e_{j}\left(X_{j}-Q_{j}^{\prime}\hat{\tau}_{-ijk}\right).
\end{align*}

Take expectation of $A_{1}$:
\small
\begin{align*}
E & \left[\sum_{i}\sum_{j\ne i}\sum_{k\ne i}G_{ij}X_{j}G_{ik}X_{k}e_{i}\left(e_{i}-Q_{i}^{\prime}\hat{\tau}_{\Delta,-ijk}\right)\right]\\
 &= \sum_{i}\sum_{j\ne i}\sum_{k\ne i}G_{ij}R_{j}G_{ik}R_{k}E\left[\nu_{i}^{2}\right]+\sum_{i}\sum_{j\ne i}G_{ij}^{2}E\left[\nu_{i}^{2}\right]E\left[\eta_{j}^{2}\right].
\end{align*}
\normalsize

Similarly, $E\left[A_{2i}\right] =\left(\sum_{j\ne i}G_{ij}R_{j}\right)\left(\sum_{j\ne i}G_{ji}R_{\Delta j}\right)E\left[\nu_{i}\eta_{i}\right]+\sum_{j\ne i}G_{ij}G_{ji}E\left[\eta_{i}\nu_{i}\right]E\left[\eta_{j}\nu_{j}\right]$ and $E\left[A_{3i}\right] =E\left[\eta_{i}^{2}\right]\left(\sum_{j\ne i}G_{ji}R_{\Delta j}\right)^{2}+\sum_{j\ne i}G_{ji}^{2}E\left[\eta_{i}^{2}\right]E\left[\nu_{j}^{2}\right]$.

For the $A_4$ and $A_5$ terms, observe that $X_{i}-Q_{i}^{\prime}\hat{\tau}_{-ij} =X_{i}-Q_{i}^{\prime}\sum_{k\ne i,j}\left(\sum_{l\ne i,j}Q_{l}Q_{l}^{\prime}\right)^{-1}Q_{k}X_{k}
=X_{i}+\sum_{k\ne i,j}\check{M}_{ik,-ij}X_{k} =\sum_{k\ne j}\check{M}_{ik,-ij}X_{k}$, where the final equality follows from $\check{M}_{ii,-ij}=1$. Then, 
\begin{align*}
E\left[A_{4ij}\right] & =E\left[X_{i}\sum_{k\ne j}\check{M}_{ik,-ij}X_{k}e_{j}\left(X_{j}-Q_{j}^{\prime}\hat{\tau}_{\Delta,-ijk}\right)\right]
 =\sum_{k\ne j}E\left[X_{i}\check{M}_{ik,-ij}X_{k}e_{j}\left(X_{j}-Q_{j}^{\prime}\hat{\tau}_{\Delta,-ijk}\right)\right]\\
 & =E\left[X_{i}\sum_{k\ne j}\check{M}_{ik,-ij}X_{k}\right]E\left[e_{j}\left(e_{j}-Q_{j}^{\prime}\hat{\tau}_{\Delta,-ijk}\right)\right] 
 =E\left[\eta_{i}^{2}\right]E\left[\nu_{j}^{2}\right].
\end{align*}
Similarly, $E\left[A_{5ij}\right]=E\left[\eta_{i}\nu_{i}\right]E\left[\eta_{j}\nu_{j}\right]$.
Combining these expressions yields the unbiasedness result.

\textbf{Proof of Consistency. } 
By Chebyshev's inequality,
\begin{align*}
\Pr & \left(\left|\frac{\hat{V}_{LM}-\Var\left(\sum_{i}\sum_{j\ne i}G_{ij}e_{i}X_{j}\right)}{\Var\left(\sum_{i}\sum_{j\ne i}G_{ij}e_{i}X_{j}\right)}\right|>\epsilon\right)\\
 & \leq\frac{1}{\epsilon^{2}}\frac{\Var\left(\sum_{i}\left(A_{1i}+2A_{2i}+A_{3i}\right)-\sum_{i}\sum_{j\ne i}G_{ji}^{2}A_{4ij}-\sum_{i}\sum_{j\ne i}G_{ij}G_{ji}A_{5ij}\right)}{\left[\Var\left(\sum_{i}\sum_{j\ne i}G_{ij}e_{i}X_{j}\right)\right]^{2}}
\end{align*}

Observe that the numerator can be written as the variance of the estimator only because $\hat{V}_{LM}$ is unbiased. 
I first establish the order of the denominator. 
Let $\tilde{R}_i := \sum_{j\ne i} G_{ij} R_j$, $\tilde{R}_{\Delta i} := \sum_{j \ne i} G_{ji} R_{\Delta j}$ and $\rho_{i}:=corr(\eta_{i}\nu_{i})$.

Since $E[\nu_i^2]$ and $E[\eta_i^2]$ are bounded away from zero and $|\textup{corr}(\eta_i\nu_i)|$ is bounded away from one by \Cref{asmp:normality}(b), the first line of the $V_{LM}$ expression in \Cref{eqn:VLM_expr_proof} has order at least $\sum_i \tilde{R}_i^2 + \sum_i \tilde{R}_{\Delta i}^2$, and the second line has order at least $\sum_i \sum_{j \ne i} G_{ij}^2$. To see this, for some $\underline{c}>0$, the first line is:
\small
\begin{align*}
 & \sum_{i}E\left[\nu_{i}^{2}\right]\tilde{R}_{i}^{2}+2\tilde{R}_{\Delta i}\tilde{R}_{i}E\left[\nu_{i}\eta_{i}\right]+\tilde{R}_{\Delta i}^{2}E\left[\eta_{i}^{2}\right] 
 =\sum_{i}E\left[\nu_{i}^{2}\right]\tilde{R}_{i}^{2}+2\tilde{R}_{\Delta i}\tilde{R}_{i}\rho_{i}\sqrt{E\left[\nu_{i}^{2}\right]E\left[\eta_{i}^{2}\right]}+\tilde{R}_{\Delta i}^{2}E\left[\eta_{i}^{2}\right]\\
 & \geq\sum_{i}\left(E\left[\nu_{i}^{2}\right]\tilde{R}_{i}^{2}+\tilde{R}_{\Delta i}^{2}E\left[\eta_{i}^{2}\right]\right)\left(1-|\rho_{i}|\right)
 +\sum_{i}|\rho_{i}|\left(E\left[\nu_{i}^{2}\right]\tilde{R}_{i}^{2}+\tilde{R}_{\Delta i}^{2}E\left[\eta_{i}^{2}\right]-2\tilde{R}_{\Delta i}\tilde{R}_{i}\sqrt{E\left[\nu_{i}^{2}\right]E\left[\eta_{i}^{2}\right]}\right)\\
 & =\sum_{i}\left(E\left[\nu_{i}^{2}\right]\tilde{R}_{i}^{2}+\tilde{R}_{\Delta i}^{2}E\left[\eta_{i}^{2}\right]\right)\left(1-|\rho_{i}|\right)+\sum_{i}|\rho_{i}|\left(\sqrt{E\left[\nu_{i}^{2}\right]\tilde{R}_{i}^{2}}-\sqrt{\tilde{R}_{\Delta i}^{2}E\left[\eta_{i}^{2}\right]}\right)^{2}\\
 & \geq\sum_{i}\left(E\left[\nu_{i}^{2}\right]\tilde{R}_{i}^{2}+\tilde{R}_{\Delta i}^{2}E\left[\eta_{i}^{2}\right]\right)\left(1-|\rho_{i}|\right) \geq \underline{c} \sum_{i}\left(\tilde{R}_{i}^{2}+\tilde{R}_{\Delta i}^{2}\right),
\end{align*}
\normalsize

and, similarly, the second line is:
\begin{align*}
 & \sum_{i}\sum_{j\ne i}G_{ij}^{2}E\left[\nu_{i}^{2}\right]E\left[\eta_{j}^{2}\right]+\sum_{i}\sum_{j\ne i}G_{ij}G_{ji}E\left[\eta_{i}\nu_{i}\right]E\left[\eta_{j}\nu_{j}\right]\\
 & =\frac{1}{2}\sum_{i}\sum_{j\ne i}G_{ij}^{2}E\left[\nu_{i}^{2}\right]E\left[\eta_{j}^{2}\right]\left(1-\rho_{i}^{2}\right)+\frac{1}{2}\sum_{i}\sum_{j\ne i}G_{ji}^{2}E\left[\nu_{j}^{2}\right]E\left[\eta_{i}^{2}\right]\left(1-\rho_{j}^{2}\right)\\
 & \quad+\frac{1}{2}\sum_{i}\sum_{j\ne i} E\left[\nu_{i}^{2}\right]E\left[\eta_{j}^{2}\right](G_{ij}^{2}\rho_{i}^{2} + G_{ji}^2 \rho_{j}^{2})+\sum_{i}\sum_{j\ne i}G_{ij}G_{ji}\rho_{i}\rho_{j}\sqrt{E\left[\nu_{i}^{2}\right]E\left[\eta_{j}^{2}\right]}\sqrt{E\left[\nu_{j}^{2}\right]E\left[\eta_{i}^{2}\right]} \\
 & \geq\frac{1}{2}\sum_{i}\sum_{j\ne i}G_{ij}^{2}E\left[\nu_{i}^{2}\right]E\left[\eta_{j}^{2}\right]\left(1-\rho_{i}^{2}\right)+\frac{1}{2}\sum_{i}\sum_{j\ne i}G_{ji}^{2}E\left[\nu_{j}^{2}\right]E\left[\eta_{i}^{2}\right]\left(1-\rho_{j}^{2}\right) \geq \underline{c} \sum_{i}\sum_{j\ne i}G_{ij}^{2}.
\end{align*}
Consequently, 
\begin{equation} \label{eqn:VLM_order}
V_{LM} \succeq \sum_i \tilde{R}_i^2 + \sum_i \tilde{R}_{\Delta i}^2 + \sum_i \sum_{j \ne i} G_{ij}^2 =: r_n.
\end{equation}
Hence, since $r_n \rightarrow \infty$ by \Cref{asmp:normality}(d), $V_{LM}$ also diverges. 
By repeated application of the Cauchy-Schwarz inequality, it suffices to show that the variance of each of the 5 $A$ terms above has order at most $r_n$ (i.e., bounded by any of the three terms in \Cref{eqn:VLM_order}). 
If this is true, then since the denominator has order at least $r_n^2$, the variance estimator is consistent. 
The A1 and A2 terms have the form:
\begin{align*}
\sum_{i}\sum_{j\ne i} & G_{ij}F_{ik}V_{1j}\sum_{k\ne i}V_{2k}V_{3i}\left(V_{4i}-Q_{i}^{\prime}\hat{\tau}_{4,-ijk}\right)=  \sum_{i}\sum_{j\ne i}\sum_{k\ne i}\sum_{l\ne j,k}G_{ij}F_{ik}V_{1j}V_{2k}V_{3i}\check{M}_{il,-ijk}V_{4l}\\
= & \sum_{i}\sum_{j\ne i}\sum_{k\ne i,j}\sum_{l\ne i,j,k}G_{ij}F_{ik}\check{M}_{il,-ijk}V_{1j}V_{2k}V_{3i}V_{4l} +\sum_{i}\sum_{j\ne i}\sum_{k\ne i,j}G_{ij}F_{ik}V_{1j}V_{2k}V_{3i}V_{4i}\\
 & +\sum_{i}\sum_{j\ne i}\sum_{l\ne i,j}G_{ij}F_{ij}\check{M}_{il,-ij}V_{1j}V_{2j}V_{3i}V_{4l} +\sum_{i}\sum_{j\ne i}G_{ij}F_{ij}V_{1j}V_{2j}V_{3i}V_{4i}.
\end{align*}
In particular, A1 uses $F=G,V_{1}=X,V_{2}=X,V_{3}=e,V_{4}=e$, while A2 uses $F=G^{\prime},V_{1}=X,V_{2}=e,V_{3}=e,V_{4}=X$ . 
By applying the Cauchy-Schwarz inequality, it suffices to show that the variance of each of the sums has order at most $r_{n}$. 
The terms $\sum_{i}\sum_{j\ne i}G_{ij}F_{ij}V_{1j}V_{2j}V_{3i}V_{4i}$ and $\sum_{i}\sum_{j\ne i}\sum_{l\ne i,j}G_{ij}F_{ij}\check{M}_{il,-ij}V_{1j}V_{2j}V_{3i}V_{4l}$ are identical to the result in \Cref{lem:a1}, with the latter result being obtained by switching the $i$ and $j$ indices. 
The remaining terms have a variance that has a bounded order by \Cref{lem:a2}. 
For A3, we can use $G_{ji}$ in place of $G_{ij}$ above, and use $F=G^{\prime},V_{1}=e,V_{2}=e,V_{3}=X,V_{4}=X$ so that the order is bounded above due to \Cref{lem:a3}. 
A4 and A5 can be written as:
\begin{align*}
\sum_{i}\sum_{j\ne i} & G_{ji}F_{ij}V_{1i}\sum_{k\ne j}\check{M}_{ik,-ij}V_{2k}V_{3j}\left(V_{4j}-Q_{j}^{\prime}\hat{\tau}_{4,-ijk}\right) \\
&=  \sum_{i}\sum_{j\ne i}\sum_{k\ne j}\sum_{l\ne i,k}G_{ji}F_{ij}V_{1i}\check{M}_{ik,-ij}V_{2k}V_{3j}\check{M}_{jl,-ijk}V_{4l}\\
= & \sum_{i}\sum_{j\ne i}\sum_{k\ne i,j}\sum_{l\ne i,j,k}G_{ji}F_{ij}\check{M}_{ik,-ij}\check{M}_{jl,-ijk}V_{1i}V_{2k}V_{3j}V_{4l} +\sum_{i}\sum_{j\ne i}\sum_{k\ne i,j}G_{ji}F_{ij}\check{M}_{ik,-ij}V_{1i}V_{2k}V_{3j}V_{4j}\\
 & +\sum_{i}\sum_{j\ne i}\sum_{l\ne i,j}G_{ji}F_{ij}\check{M}_{jl,-ij}V_{1i}V_{2i}V_{3j}V_{4l} +\sum_{i}\sum_{j\ne i}G_{ji}F_{ij}V_{1i}V_{2i}V_{3j}V_{4j}.
\end{align*}
In particular, A4 uses $F=G^{\prime},V_{1}=X,V_{2}=X,V_{3}=e,V_{4}=e$, while A5 uses $F=G,V_{1}=e,V_{2}=X,V_{3}=e,V_{4}=X$ . 
By applying the Cauchy-Schwarz inequality, it suffices to show that the variance of each of the sums has order at most $r_{n}$. 
This result is immediate from \Cref{lem:a1} and \Cref{lem:a3}. 
\end{proof}

\begin{proof} [Proof of \Cref{prop:umpu}]

Let $\mu\in\mathcal{M}= \{\mu:  \mu_1>0, \mu_3 >0, \mu_2^2 < \mu_1 \mu_3 \} $. I first show that $\mathcal{M}$ is convex.
For $\lambda\in(0,1)$, it suffices to show, for $\mu_{a}$ and $\mu_{b}$
that satisfy $\mu_{2a}^{2} < \mu_{1a}\mu_{3a}$ and $\mu_{2b}^{2}<\mu_{1b}\mu_{3b}$,
that $\left(\lambda\mu_{2a}+\left(1-\lambda\right)\mu_{2b}\right)^{2}<\left(\lambda\mu_{1a}+\left(1-\lambda\right)\mu_{1b}\right)\left(\lambda\mu_{3a}+\left(1-\lambda\right)\mu_{3b}\right)$.
This set is intersected with the set that satisfies $\mu_{1}>0$ and $\mu_{3}>0$, which is clearly convex. The following is negative:
\small
\begin{align*}
 & \left(\lambda\mu_{2a}+\left(1-\lambda\right)\mu_{2b}\right)^{2}-\left(\lambda\mu_{1a}+\left(1-\lambda\right)\mu_{1b}\right)\left(\lambda\mu_{3a}+\left(1-\lambda\right)\mu_{3b}\right)\\
= & \lambda^{2}\left(\mu_{2a}^{2}-\mu_{1a}\mu_{3a}\right)+\left(1-\lambda\right)^{2}\left(\mu_{2b}^{2}-\mu_{1b}\mu_{3b}\right)+\lambda\left(1-\lambda\right)\left(2\mu_{2a}\mu_{2b}-\mu_{1b}\mu_{3a}-\mu_{1a}\mu_{3b}\right)\\
< & \lambda\left(1-\lambda\right)\left(2\sqrt{\mu_{1a}\mu_{1b}\mu_{1b}\mu_{3b}}-\mu_{1b}\mu_{3a}-\mu_{1a}\mu_{3b}\right)
<  -\lambda\left(1-\lambda\right)\left(\sqrt{\mu_{1b}\mu_{3a}}-\sqrt{\mu_{1a}\mu_{3b}}\right)^{2}\leq0.
\end{align*}
\normalsize
The first inequality occurs from applying $\mu_{2a}^{2}<\mu_{1a}\mu_{3a}$ and $\mu_{2b}^{2}<\mu_{1b}\mu_{3b}$, so $\mathcal{M}$ is convex.
Let $m\sim N(\mu,\Sigma)$ denote a statistic drawn from the asymptotic distribution, with $m_{i}$ being a component of the vector $m$, so that $m_{2}$ is the LM statistic. 
Using the linear transformation from \citet{lehmann2005testing} Example 3.9.2 Case 3, we can transform the statistics and parameter such that $m_{2}$ is orthogonal to all other components. 
In particular, consider the following transformation $L$:
\begin{align*}
L:=\left(\begin{array}{ccc}
\sqrt{\frac{\sigma_{22}}{\sigma_{11}\sigma_{22}-\sigma_{12}^{2}}} & -\frac{\sigma_{12}}{\sigma_{22}}\sqrt{\frac{\sigma_{22}}{\sigma_{11}\sigma_{22}-\sigma_{12}^{2}}} & 0\\
0 & \frac{1}{\sqrt{\sigma_{22}}} & 0\\
0 & -\frac{\sigma_{23}}{\sigma_{22}}\sqrt{\frac{\sigma_{22}}{\sigma_{33}\sigma_{22}-\sigma_{23}^{2}}} & \sqrt{\frac{\sigma_{22}}{\sigma_{33}\sigma_{22}-\sigma_{23}^{2}}}
\end{array}\right).
\end{align*}
Then,
\begin{align*}
Lm \sim N \left( L\mu, \left(\begin{array}{ccc}
1 & 0 & \frac{\sigma_{13}\sigma_{22}-\sigma_{12}\sigma_{23}}{\left(\sigma_{11}\sigma_{22}-\sigma_{12}^{2}\right)\left(\sigma_{33}\sigma_{22}-\sigma_{23}^{2}\right)}\\
0 & 1 & 0\\
\frac{\sigma_{13}\sigma_{22}-\sigma_{12}\sigma_{23}}{\left(\sigma_{11}\sigma_{22}-\sigma_{12}^{2}\right)\left(\sigma_{33}\sigma_{22}-\sigma_{23}^{2}\right)} & 0 & 1
\end{array}\right) \right).
\end{align*}
The parameter space of $L\mu\in\mathcal{L}$ is also convex because $L$ is a linear transformation: take any $\mu_{a},\mu_{b}\in\mathcal{M}$, then observe that $\lambda L\mu_{a}+(1-\lambda)L\mu_{b}=L\left(\lambda\mu_{a}+\left(1-\lambda\right)\mu_{b}\right)$.
Since $\mathcal{M}$ is convex, and every element in $\mathcal{M}$ is linearly transformed into the space on $\mathcal{L}$, we have $\lambda\mu_{a}+\left(1-\lambda\right)\mu_{b}\in\mathcal{M}$ and hence $L\left(\lambda\mu_{a}+\left(1-\lambda\right)\mu_{b}\right)\in\mathcal{L}$.
Since $Lm$ is normally distributed and $\mathcal{L}$ is convex with rank 3, the problem is in the exponential class, using the definition from \citet{lehmann2005testing} Section 4.4. 
Since the joint distribution is in the exponential class and the restriction to the interior ensures that there are points in the parameter space that are above and below the null, the uniformly most powerful unbiased test follows the form of \citet{lehmann2005testing} Theorem 4.4.1(iv), by using $U=m_2$ and $T=\left( \sqrt{\frac{\sigma_{22}}{\sigma_{33}\sigma_{22}-\sigma_{23}^{2}}}m_{3}-\frac{\sigma_{23}}{\sqrt{\sigma_{22}\left(\sigma_{33}\sigma_{22}-\sigma_{23}^{2}\right)}}m_{2}, \sqrt{\frac{\sigma_{22}}{\sigma_{11}\sigma_{22}-\sigma_{12}^{2}}}m_{1}-\frac{\sigma_{12}}{\sqrt{\sigma_{22}\left(\sigma_{11}\sigma_{22}-\sigma_{12}^{2}\right)}}m_{2}\right)^\prime$ in their notation.
To calculate the critical values of the \citet{lehmann2005testing} Theorem 4.4.1(iv) result, observe that $\left[Lm\right]_{2}$ is orthogonal to $\left[Lm\right]_{1}$ and $\left[Lm\right]_{3}$, so the distribution of $\left[Lm\right]_{2}$ conditional on $\left[Lm\right]_{1}$ and $\left[Lm\right]_{3}$ is standard normal. 
Since $\left[Lm\right]_{2}$ is standard normal, it is symmetric around $0$ under the null, so the solution to the critical value is $\pm1.96$ for a 5\% test, due to simplification in \citet{lehmann2005testing} Section 4.2. 
The resulting test is hence identical to the two-sided LM test. 
\end{proof}

\bibliographystyle{ecta}
\bibliography{mwiv}

\clearpage
\section{Online Appendix}

\subsection{Data Generating Process} \label{sec:sect2DGP}
Data is generated from an environment with $E[\varepsilon_i] =0$, and $\int_0^1 f(v) dv = \beta$. 
To run a regression on judge indicators (without an intercept) in the reduced-form system, I make a transformation $\check{X} = 2X - 1$ so that the reduced-form equations can be written as:
\begin{align*}
\check{X}_{i} & =Z_{i}^{\prime}\pi+\eta_{i}, \text{ and } Y_{i} =Z_{i}^{\prime}\pi_{Y}+\zeta_{i},
\end{align*}
so $\pi_k = \pi_{Yk}=0$ for the base judge. 
The reduced-form errors are: $\eta_{i} =I\left\{ \lambda_{k(i)}-v_{i}\geq0\right\} -\pi_{k(i)}$ and $\zeta_i = 1\left\{ \lambda_{k(i)}-v_{i}\geq0\right\} f\left(v_{i}\right)+\varepsilon_{i}-\pi_{Yk(i)}$ respectively.
With $\pi_{\Delta k} = \pi_{Yk}- \pi_k \beta$, the reduced-form parameters for the groups of judges are derived in \Cref{tab:parRFmte}.
Since the coefficient of the base judge is normalized to zero, the implementation without covariates in simulations excludes the intercept and uses indicators for all judges, instead of omitting the base judge and having an intercept. 
This implementation results in a block diagonal projection matrix, which aids computational speed, while retaining the interpretation of $\pi$'s in the reduced-form model.
The $f(v)$ that delivers the parameters in \Cref{tab:parRFmte} is
\begin{equation}
f\left(v\right)=\begin{cases}
-s\beta+h & v\in[0,\frac{1}{2}-s]\\
\frac{1}{s}\left(1-s\right)\left(-\frac{1}{2}s\beta-h\right)-\frac{1}{s}\left(1-2s\right)\left(-s\beta+h\right) & v\in(\frac{1}{2}-s,\frac{1}{2}-\frac{1}{2}s]\\
\frac{1}{s}\left(1-s\right)\left(\frac{1}{2}s\beta+h\right) & v\in(\frac{1}{2}-\frac{1}{2}s,\frac{1}{2}]\\
\frac{1}{s}\left(1+s\right)\left(\frac{1}{2}s\beta-h\right) & v\in(\frac{1}{2},\frac{1}{2}+\frac{1}{2}s]\\
\frac{1}{s}\left(1+2s\right)\left(s\beta+h\right)-\frac{1}{s}\left(1+s\right)\left(\frac{1}{2}s\beta-h\right) & v\in(\frac{1}{2}+\frac{1}{2}s,\frac{1}{2}+s]\\
\frac{\beta-\left(\frac{1}{2}+s\right)\left(s\beta+h\right)}{\frac{1}{2}-s} & v\in(\frac{1}{2}+s,1]
\end{cases}.
\end{equation}

\begin{table}
    \centering
    \caption{Parameters for Simple Example} \label{tab:parRFmte}
\begin{tabular}{cccccc}
\toprule
$\lambda_{k}$ & $\frac{1}{2}-s$ & $\frac{1}{2}-\frac{1}{2}s$ & $\frac{1}{2}$ & $\frac{1}{2}+\frac{1}{2}s$ & $\frac{1}{2}+s$ \\
\midrule
$\beta_{k}$ & $\beta-\frac{h}{s}$ & $\beta+2\frac{h}{s}$ & NA & $\beta-2\frac{h}{s}$ & $\beta+\frac{h}{s}$ \\ 
\midrule
$\pi_{k}$ & $-s$ & $-\frac{1}{2}s$ & $0$ & $\frac{1}{2}s$ & $s$ \\
\midrule
$\pi_{Yk}$ & $-s\beta+h$ & $-\frac{1}{2}s\beta-h$ & $0$ & $\frac{1}{2}s\beta-h$ & $s\beta+h$ \\
\midrule
$\pi_{\Delta k}$ & $h$ & $-h$ & $0$ & $-h$ & $h$ \\
\bottomrule 
\end{tabular}
\end{table}

To generate the data in the simulation, I draw $v_i \sim U[0,1]$ as implied by the structural model, then generate $\zeta_i \mid v_i \sim N(\sigma_{\varepsilon v}v_i , \sigma_{\varepsilon \varepsilon})$.
Hence, $\sigma_{\varepsilon v}$ and $\sigma_{\varepsilon \varepsilon}$ control the correlation between $\eta_i$ and $\zeta_i$, with $\sigma_{\varepsilon \varepsilon}=0$ corresponding to perfect correlation.
In the base case, I set $\sigma_{\varepsilon \varepsilon} = 0.1$ and $\sigma_{\varepsilon v} = 0.3$.
With the given $\pi_k, \pi_{Yk}$, the observable variables are generated from $\check{X}_i = I\{ \pi_{k(i)} > v_i \}$ and $Y_i = \pi_{Yk(i)} + \zeta_i$.

\subsection{Comparing Variance Estimands} \label{sec:compare_var_estimands}

\textbf{Derivations for Constructed Instrument}
Using the notation for the just-identified IV AR test in \Cref{sec:constructedIV} when $\beta = \beta_0$,
\begin{align*}
\hat{\varepsilon}_{i} & =e_{i}-\tilde{X}_{i}\frac{\sum_{i}e_{i}\tilde{X}_{i}}{\sum_{i}\tilde{X}_{i}^{2}} 
=\frac{e_{i}\sum_{i}\tilde{X}_{i}^{2}-\tilde{X}_{i}\sum_{i}e_{i}\tilde{X}_{i}}{\sum_{i}\tilde{X}_{i}^{2}}, \text{ and } \\
\hat{V} & =\frac{\sum_{i}\tilde{X}_{i}^{2}\hat{\varepsilon}_{i}^{2}}{\left(\sum_{i}\tilde{X}_{i}^{2}\right)^{2}} 
=\frac{\sum_{i}\tilde{X}_{i}^{2}\left(e_{i}\sum_{j}\tilde{X}_{j}^{2}-\tilde{X}_{i}\sum_{j}e_{j}\tilde{X}_{j}\right)^{2}}{\left(\sum_{i}\tilde{X}_{i}^{2}\right)^{4}}\\
 & =\frac{\sum_{i}\tilde{X}_{i}^{2}e_{i}^{2}\left(\sum_{j}\tilde{X}_{j}^{2}\right)^{2}+\sum_{i}\tilde{X}_{i}^{4}\left(\sum_{j}e_{j}\tilde{X}_{j}\right)^{2}-2\sum_{i}\tilde{X}_{i}^{3}e_{i}\left(\sum_{j}\tilde{X}_{j}^{2}\right)\left(\sum_{j}e_{j}\tilde{X}_{j}\right)}{\left(\sum_{i}\tilde{X}_{i}^{2}\right)^{4}}.
\end{align*}

Applying the asymptotic result that $\frac{1}{n}\sum_{j}e_{j}\tilde{X}_{j}\xrightarrow{p}0$ from \Cref{thm:normality},
\begin{align*}
t_{AR}^{2} & =\frac{\frac{\left(\sum_{i}\tilde{X}_{i}e_{i}\right)^{2}}{\left(\sum_{i}\tilde{X}_{i}^{2}\right)^{2}}}{\frac{\sum_{i}\tilde{X}_{i}^{2}e_{i}^{2}\left(\sum_{j}\tilde{X}_{j}^{2}\right)^{2}+\sum_{i}\tilde{X}_{i}^{4}\left(\sum_{j}e_{j}\tilde{X}_{j}\right)^{2}-2\sum_{i}\tilde{X}_{i}^{3}e_{i}\left(\sum_{j}\tilde{X}_{j}^{2}\right)\left(\sum_{j}e_{j}\tilde{X}_{j}\right)}{\left(\sum_{i}\tilde{X}_{i}^{2}\right)^{4}}}\\
 & =\frac{\left(\frac{1}{\sqrt{n}}\sum_{i}\tilde{X}_{i}e_{i}\right)^{2}\left(\frac{1}{n}\sum_{i}\tilde{X}_{i}^{2}\right)^{2}}{\frac{1}{n}\sum_{i}\tilde{X}_{i}^{2}e_{i}^{2}\left(\frac{1}{n}\sum_{j}\tilde{X}_{j}^{2}\right)^{2}+\frac{1}{n}\sum_{i}\tilde{X}_{i}^{4}\left(\frac{1}{n}\sum_{j}e_{j}\tilde{X}_{j}\right)^{2}-2\frac{1}{n}\sum_{i}\tilde{X}_{i}^{3}e_{i}\left(\frac{1}{n}\sum_{j}\tilde{X}_{j}^{2}\right)\left(\frac{1}{n}\sum_{j}e_{j}\tilde{X}_{j}\right)}\\
 & =\frac{\left(\frac{1}{\sqrt{n}}\sum_{i}\tilde{X}_{i}e_{i}\right)^{2}\left(\frac{1}{n}\sum_{i}\tilde{X}_{i}^{2}\right)^{2}}{\frac{1}{n}\sum_{i}\tilde{X}_{i}^{2}e_{i}^{2}\left(\frac{1}{n}\sum_{j}\tilde{X}_{j}^{2}\right)^{2}+o_{P}(1)}=\frac{\left(\frac{1}{\sqrt{n}}\sum_{i}\tilde{X}_{i}e_{i}\right)^{2}}{\frac{1}{n}\sum_{i}\tilde{X}_{i}^{2}e_{i}^{2}}+o_{P}(1), \text { and }
\end{align*}
\begin{align*}
E\left[\sum_{i}\tilde{X}_{i}^{2}e_{i}^{2}\right] & =\sum_{i}E\left[\left(\sum_{j\ne i}P_{ij}\left(R_{j}+\eta_{j}\right)\right)^{2}\left(R_{\Delta i}+\nu_{i}\right)^{2}\right] \\
&=\sum_{i}E\left[\left(M_{ii}^{2}R_{i}^{2}+\left(\sum_{j\ne i}P_{ij}^{2}\eta_{j}^{2}\right)\right)\left(R_{\Delta i}^{2}+\nu_{i}^{2}\right)\right]\\
 & =\sum_{i}\left(M_{ii}^{2}R_{i}^{2}R_{\Delta i}^{2}+\sum_{j\ne i}P_{ij}^{2}R_{\Delta i}^{2}E\left[\eta_{j}^{2}\right]+M_{ii}^{2}R_{i}^{2}E\left[\nu_{i}^{2}\right]+\sum_{j\ne i}P_{ij}^{2}E\left[\nu_{i}^{2}\right]E\left[\eta_{j}^{2}\right]\right).
\end{align*}

\textbf{Clustering by Judges}. If we just use the JIVE t-ratio with clustering, weak identification is still a problem, and we should similarly get over-rejection. 
The just-identified AR with clustered standard errors uses the following estimator:
\[
\hat{V}_{clus}=\frac{\sum_{i}\sum_{j\in\mathcal{N}_{i}}\tilde{X}_{i}\hat{\varepsilon}_{i}\tilde{X}_{j}\hat{\varepsilon}_{j}}{\left(\sum_{i}\tilde{X}_{i}^{2}\right)^{2}},
\]
where $\mathcal{N}_{i}$ is the neighborhood of $i$ (i.e., the set of observations that share the same cluster as $i$). 
Expanding $\hat{V}_{clus}$ using the same steps as before,
\begin{align*}
\hat{V}_{clus} & =\frac{\sum_{i}\sum_{j\in\mathcal{N}_{i}}\tilde{X}_{i}\tilde{X}_{j}\left(e_{i}e_{j}\left(\sum_{k}\tilde{X}_{k}^{2}\right)^{2}-2\tilde{X}_{i}e_{j}\left(\sum_{k}e_{k}\tilde{X}_{k}\right)\left(\sum_{k}\tilde{X}_{k}^{2}\right)+\tilde{X}_{i}\tilde{X}_{j}\left(\sum_{k}e_{k}\tilde{X}_{k}\right)^{2}\right)}{\left(\sum_{i}\tilde{X}_{i}^{2}\right)^{4}}.
\end{align*}
Using the fact that $\frac{1}{n}\sum_{k}e_{k}\tilde{X}_{k}=o_{P}(1)$, the dominant term is: $\sum_{i}\sum_{j\in\mathcal{N}_{i}}\tilde{X}_{i}\tilde{X}_{j}e_{i}e_{j}$, which is analogous to the previous derivation. 
The expansion steps are analogous to that required to derive $V_{LM}$, so they are omitted.
Then,
\begin{align*}
E&\left[\sum_{i}\sum_{j\in\mathcal{N}_{i}}\tilde{X}_{i}\tilde{X}_{j}e_{i}e_{j}\right]  \\
&=E\left[\sum_{i}\sum_{j\in\mathcal{N}_{i}}\left(\sum_{k\ne i}P_{ik}\left(R_{k}+\eta_{k}\right)\right)\left(\sum_{k\ne j}P_{jk}\left(R_{k}+\eta_{k}\right)\right)\left(R_{\Delta i}+\nu_{i}\right)\left(R_{\Delta j}+\nu_{j}\right)\right]\\
 & =E\left[\sum_{i}\sum_{j\in\mathcal{N}_{i}}\left(M_{ii}R_{i}+\sum_{k\ne i}P_{ik}\eta_{k}\right)\left(M_{jj}R_{j}+\sum_{k\ne j}P_{jk}\eta_{k}\right)\left(R_{\Delta i}R_{\Delta j}+\nu_{i}R_{\Delta j}+R_{\Delta i}\nu_{j}+\nu_{i}\nu_{j}\right)\right]\\
 & =\sum_{i}\sum_{j\in\mathcal{N}_{i}}\left(M_{ii}M_{jj}R_{i}R_{j}R_{\Delta i}R_{\Delta j}+R_{\Delta i}R_{\Delta j}\sum_{k\ne i,j}P_{ik}P_{jk}E\left[\eta_{k}^{2}\right]\right)\\
 & \quad+2\sum_{i}\sum_{j\in\mathcal{N}_{i}}M_{ii}R_{i}R_{\Delta j}P_{ji}E\left[\eta_{i}\nu_{i}\right]+\sum_{i}M_{ii}^{2}R_{i}^{2}E\left[\nu_{i}^{2}\right]\\
 & \quad+\sum_{i}\sum_{k\ne i}P_{ik}^{2}E\left[\nu_{i}^{2}\right]E\left[\eta_{k}^{2}\right]+\sum_{i}\sum_{j\in\mathcal{N}_{i}\backslash\{i\}}P_{ij}^{2}E\left[\nu_{i}\eta_{i}\right]E\left[\nu_{j}\eta_{j}\right].
\end{align*}

By applying the fact that the entries of the projection matrix are nonzero only when the observations share the same judge, the expression simplifies further:
\begin{align*}
E\left[\sum_{i}\sum_{j\in\mathcal{N}_{i}}\tilde{X}_{i}\tilde{X}_{j}e_{i}e_{j}\right] & =\sum_{i}\sum_{j\in\mathcal{N}_{i}}M_{ii}M_{jj}R_{i}R_{j}R_{\Delta i}R_{\Delta j}+\sum_{i}M_{ii}^{2}R_{\Delta i}^{2}E\left[\eta_{i}^{2}\right]\\
 & \quad+2\sum_{i}M_{ii}R_{i}R_{\Delta i}E\left[\eta_{i}\nu_{i}\right]+\sum_{i}M_{ii}^{2}R_{i}^{2}E\left[\nu_{i}^{2}\right] \\
 &\quad +\sum_{i}\sum_{j\ne i}P_{ij}^{2}\left(E\left[\nu_{i}^{2}\right]E\left[\eta_{j}^{2}\right]+E\left[\nu_{i}\eta_{i}\right]E\left[\nu_{j}\eta_{j}\right]\right).
\end{align*}

Compared to the true variance in \Cref{cor:var_expr}, due to the own-observation bias, we have an extra $\sum_{i}\sum_{j\in\mathcal{N}_{i}}M_{ii}M_{jj}R_{i}R_{j}R_{\Delta i}R_{\Delta j}$ term, and the estimand here has $\sum_{i}M_{ii}R_{i}R_{\Delta i}E\left[\eta_{i}\nu_{i}\right]$
instead of $\sum_{i}M_{ii}^{2}R_{i}R_{\Delta i}E\left[\eta_{i}\nu_{i}\right]$.
Even though $\sum_{i}\sum_{j\in\mathcal{N}_{i}}M_{ii}M_{jj}R_{i}R_{j}R_{\Delta i}R_{\Delta j} \geq 0$, $\sum_i M_{ii} (1-M_{ii}) R_{i}R_{\Delta i}E\left[\eta_{i}\nu_{i}\right]$ could be positive or negative, so the clustered variance estimand could either over or underestimate the true variance.

\textbf{Comparing MO Variance Estimand with L3O.}
The \citet{matsushita2022jackknife} variance estimator presented in \Cref{eqn:VMO_def} is biased in general. 
The model of \Cref{sec:simple_setting} implies:
\begin{equation} \label{lem:VMO_expr}
\begin{split}
E\left[\hat{\Psi}_{MO}\right] & =\sum_{i}M_{ii}^{2}R_{i}^{2}R_{\Delta i}^{2}+\sum_{i}M_{ii}^{2}R_{i}^{2}E\left[\nu_{i}^{2}\right]+\sum_{i}\sum_{j\ne i}P_{ij}^{2}E\left[\nu_{i}^{2}\right]E\left[\eta_{j}^{2}\right]+\sum_{i}\sum_{j\ne i}P_{ij}^{2}R_{\Delta i}^{2}E\left[\eta_{j}^{2}\right]\\
 & \quad+\sum_{i}\sum_{j\ne i}P_{ij}^{2}\left(R_{i}R_{\Delta i}R_{j}R_{\Delta j}+E\left[\eta_{i}\nu_{i}\right]R_{j}R_{\Delta j}+R_{i}R_{\Delta i}E\left[\eta_{j}\nu_{j}\right]+E\left[\eta_{i}\nu_{i}\right]E\left[\eta_{j}\nu_{j}\right]\right).
\end{split}
\end{equation}
As a corollary of \Cref{lem:var_expr}, when $G=P$, by observing that $P$ is symmetric, and that since $PR=I$, we have $\sum_{j\ne i}P_{ij}R_{j}=\sum_{j\ne i}P_{ji}R_{j}=M_{ii}R_{i}$, so
\begin{equation} \label{cor:var_expr}
\begin{split}
\Var\left(\sum_{i}\sum_{j\ne i}P_{ij}e_{i}X_{j}\right) & =\sum_{i}E\left[\nu_{i}^{2}\right]M_{ii}^{2}R_{i}^{2}+\sum_{i}\sum_{j\ne i}P_{ij}^{2}\left(E\left[\nu_{i}^{2}\right]E\left[\eta_{j}^{2}\right]+E\left[\eta_{i}\nu_{i}\right]E\left[\eta_{j}\nu_{j}\right]\right)\\
 & \quad+2\sum_{i}E\left[\nu_{i}\eta_{i}\right]M_{ii}^{2}R_{i}R_{\Delta i}+\sum_{i}E\left[\eta_{i}^{2}\right]M_{ii}^{2}R_{\Delta i}^{2}.
\end{split}
\end{equation}
If the $R_{\Delta}$'s are zero, then $\hat{\Psi}_{MO}$ is unbiased.
Nonetheless, heterogeneity results in many excess terms in the expectation of the variance estimator, generating bias and inconsistency in general. 
However, $\hat{\Psi}_{MO}$ can be consistent when forcing weak identification and weak heterogeneity. 
If it is assumed that $\frac{1}{\sqrt{K}}\sum_{i}M_{ii}R_{i}^{2}\rightarrow C_{S}<\infty$ and $\frac{1}{\sqrt{K}}\sum_{i}M_{ii}R_{\Delta i}^{2}\rightarrow C<\infty$ with weak identification and weak heterogeneity, then the excess terms in $\frac{1}{K}E\left[\hat{\Psi}_{MO}\right]$ can be written as $\frac{1}{\sqrt{K}}\frac{1}{\sqrt{K}}\sum_{i}M_{ii}R_{i}^{2}=\frac{1}{\sqrt{K}}O(1)=o(1)$ and $\frac{1}{\sqrt{K}}\frac{1}{\sqrt{K}}\sum_{i}M_{ii}R_{\Delta i}^{2}=o(1)$.
However, when identification or heterogeneity is strong, $\frac{1}{K}\sum_{i}M_{ii}R_{i}^{2}$ or $\frac{1}{K}\sum_{i}M_{ii}R_{\Delta i}^{2}$ is nonnegligible and the variance estimator is inconsistent. 
The variance estimator adapted from MS22 has similar properties.
In contrast, the L3O variance estimator is robust regardless of whether the identification is weak or strong. 

In general, heterogeneity does not make the MO variance estimator any more or less conservative than L3O. In the simple case with judge instruments and $G=P$, we have:
\begin{align*}
E&\left[\hat{\Psi}_{MO}\right]-Var\left(\sum_{i}\sum_{j\ne i}P_{ij}e_{i}X_{j}\right) =\sum_{i}M_{ii}^{2}R_{i}^{2}R_{\Delta i}^{2}+\sum_{i}\sum_{j\ne i}P_{ij}^{2}R_{i}R_{\Delta i}R_{j}R_{\Delta j}\\
 & \quad+\sum_{i}\sum_{j\ne i}P_{ij}^{2}R_{\Delta i}^{2}E\left[\eta_{j}^{2}\right]-\sum_{i}M_{ii}^{2}R_{\Delta i}^{2}E\left[\eta_{i}^{2}\right]
 +2\sum_{i}\sum_{j\ne i}P_{ij}^{2}E\left[\eta_{i}\nu_{i}\right]R_{j}R_{\Delta j} \\
 &\quad -2\sum_{i}E\left[\nu_{i}\eta_{i}\right]M_{ii}^{2}R_{i}R_{\Delta i} \\
 &= \sum_{i}M_{ii} R_{\Delta i}^{2} R_{i}^{2} -\sum_{i}M_{ii} (1- 2P_{ii}) R_{\Delta i}^{2}E\left[\eta_{i}^{2}\right] -2\sum_{i} M_{ii} (1- 2P_{ii}) E\left[\eta_{i}\nu_{i}\right]R_{i}R_{\Delta i} \\
 &= \sum_{i}M_{ii} R_{\Delta i}^{2} \left(R_{i}^{2} - (1- 2P_{ii}) E\left[\eta_{i}^{2}\right] \right) -2\sum_{i} M_{ii} (1- 2P_{ii}) E\left[\eta_{i}\nu_{i}\right]R_{i}R_{\Delta i}
\end{align*}
which can be positive or negative. 
The second equality uses the fact that $P$ and $M$ are non-zero only for observations that share the same judge, and when that occurs, they have the same $R, R_Y, E[\eta_i^2]$, and $E[\zeta_i^2]$, and that $\sum_{j \ne i}P_{ij}^2 = P_{ii} M_{ii}$.

To compare the confidence sets of MO and L3O, observe that the shape of the confidence set depends on the coefficient on $\beta_{0}^2$. 
In particular, for $\Psi_{2}:=\sum_{i}\left(\sum_{j\ne i}P_{ij}X_{j}\right)^{2}X_{i}^{2}+\sum_{i}\sum_{j\ne i}P_{ij}^{2}X_{i}^{2}X_{j}^{2}$, MO is unbounded when $\left(\sum_{i\ne j}^{n}P_{ij}X_{i}X_{j}\right)^{2}-q\Psi_{2}<0$ and L3O is unbounded when $\left(\sum_{i\ne j}^{n}P_{ij}X_{i}X_{j}\right)^{2}-qB_{2}<0$. 
In the simple judges case without covariates, the expected coefficients can be compared. With
\begin{align*}
E\left[\Psi_{2}\right]&=\sum_{i}\left(\left(\sum_{j\ne i}P_{ij}R_{j}\right)^{2}+\sum_{j\ne i}P_{ij}^{2}E\left[\eta_{j}^{2}\right]\right) E[X_i^2]+\sum_{i}\sum_{j\ne i}P_{ij}^{2}E[X_i^2] E[X_j^2], 
\end{align*}
the difference is:
\begin{align*}
E\left[\Psi_{2}\right]-E\left[B_{2}\right] &= \sum_{i}\left(\sum_{j\ne i}P_{ij}R_{j}\right)^{2}R_{i}^{2}+\sum_{i}\sum_{j\ne i}P_{ij}^{2}\left(R_{i}^{2}R_{j}^{2}+3E\left[\eta_{i}^{2}\right]R_{j}^{2}\right) \\
&\quad -3\sum_{i}\left(\sum_{j\ne i}P_{ij}R_{j}\right)^{2}E\left[\eta_{i}^{2}\right]\\
 &= \sum_{i}M_{ii}R_{i}^{2}\left(R_{i}^{2}-3\left(1-2P_{ii}\right)E\left[\eta_{i}^{2}\right]\right),
\end{align*}
where the second equality uses the same trick as before. 

\textbf{Derivation of LM Variance.}
\begin{proof}[Proof of \Cref{lem:var_expr}]
Expanding the variance,
\begin{align*}
\Var&\left(\sum_{i}\sum_{j\ne i}G_{ij}e_{i}X_{j}\right) =E\left[\left(\sum_{i}\sum_{j\ne i}G_{ij}e_{i}X_{j}\right)^{2}\right] =E\left[\sum_{i}\sum_{j\ne i}\sum_{k}\sum_{l\ne k}G_{ij}e_{i}X_{j}G_{kl}e_{k}X_{l}\right]\\
 & =\sum_{i}\sum_{j\ne i}\sum_{k}\sum_{l\ne k}G_{ij}G_{kl}E\left[\nu_{i}X_{j}\nu_{k}X_{l}\right]+\sum_{i}\sum_{j\ne i}\sum_{k}\sum_{l\ne k}G_{ij}G_{kl}E\left[\nu_{i}X_{j}R_{\Delta k}X_{l}\right]\\
 & \quad+\sum_{i}\sum_{j\ne i}\sum_{k}\sum_{l\ne k}G_{ij}G_{kl}E\left[R_{\Delta i}X_{j}\nu_{k}X_{l}\right]+\sum_{i}\sum_{j\ne i}\sum_{k}\sum_{l\ne k}G_{ij}G_{kl}E\left[R_{\Delta i}X_{j}R_{\Delta k}X_{l}\right]
\end{align*}

The first term is:
\begin{align*}
\sum_{i} & \sum_{j\ne i}\sum_{k}\sum_{l\ne k}G_{ij}G_{kl}E\left[\nu_{i}X_{j}\nu_{k}X_{l}\right]\\
 & =\sum_{i}\sum_{j\ne i}\sum_{k}\sum_{l\ne k}G_{ij}G_{kl}E\left[\nu_{i}R_{j}\nu_{k}R_{l}+\nu_{i}\eta_{j}\nu_{k}R_{l}+\nu_{i}R_{j}\nu_{k}\eta_{l}+\nu_{i}\eta_{j}\nu_{k}\eta_{l}\right]\\
 & =\sum_{i}\sum_{j\ne i}\sum_{k}\sum_{l\ne k}G_{ij}G_{kl}E\left[\nu_{i}R_{j}\nu_{k}R_{l}+\nu_{i}\eta_{j}\nu_{k}\eta_{l}\right]\\
 & =\sum_{i}\sum_{j\ne i}\sum_{k\ne i}E\left[\nu_{i}^{2}\right]G_{ij}G_{ik}R_{j}R_{k}+\sum_{i}\sum_{j\ne i}\left(\sum_{l\ne i}G_{ij}G_{il}E\left[\nu_{i}\eta_{j}\nu_{i}\eta_{l}\right]+\sum_{l\ne j}G_{ij}G_{jl}E\left[\nu_{i}\eta_{j}\nu_{j}\eta_{l}\right]\right)\\
 & \quad+\sum_{i}\sum_{j\ne i}\left(\sum_{k\ne i,j}G_{ij}G_{ki}E\left[\nu_{i}\eta_{j}\nu_{k}\eta_{i}\right]+\sum_{k\ne i,j}G_{ij}G_{kj}E\left[\nu_{i}\eta_{j}\nu_{k}\eta_{j}\right]\right)\\
 & =\sum_{i}\sum_{j\ne i}\sum_{k\ne i}E\left[\nu_{i}^{2}\right]G_{ij}G_{ik}R_{j}R_{k}+\sum_{i}\sum_{j\ne i}\left(G_{ij}^{2}E\left[\nu_{i}^{2}\eta_{j}^{2}\right]+G_{ij}G_{ji}E\left[\nu_{i}\eta_{i}\eta_{j}\nu_{j}\right]\right)\\
 & =\sum_{i}\sum_{j\ne i}\sum_{k\ne i}E\left[\nu_{i}^{2}\right]G_{ij}G_{ik}R_{j}R_{k}+\sum_{i}\sum_{j\ne i}\left(G_{ij}^{2}E\left[\nu_{i}^{2}\right]E\left[\eta_{j}^{2}\right]+G_{ij}G_{ji}E\left[\nu_{i}\eta_{i}\right]E\left[\eta_{j}\nu_{j}\right]\right)
\end{align*}

In the next few terms, the expansion steps are analogous, so intermediate steps are omitted for brevity. 
The second to fourth terms can be expressed as:
\begin{align*}
\sum_{i} & \sum_{j\ne i}\sum_{k}\sum_{l\ne k}G_{ij}G_{kl}E\left[\nu_{i}X_{j}R_{\Delta k}X_{l}\right]
=\sum_{i}E\left[\nu_{i}\eta_{i}\right]\sum_{j\ne i}G_{ij}R_{j}\sum_{k\ne i}G_{ki}R_{\Delta k}, \\
\sum_{i} & \sum_{j\ne i}\sum_{k}\sum_{l\ne k}G_{ij}G_{kl}E\left[R_{\Delta i}X_{j}\nu_{k}X_{l}\right]
 = \sum_{i}\sum_{j\ne i}\sum_{l\ne i}G_{ji}G_{il}E\left[\eta_{i}\nu_{i}\right]R_{\Delta j}R_{l}, \text { and } \\
\sum_{i} & \sum_{j\ne i}\sum_{k}\sum_{l\ne k}G_{ij}G_{kl}E\left[R_{\Delta i}X_{j}R_{\Delta k}X_{l}\right]
  =\sum_{i}E\left[\eta_{i}^{2}\right]\sum_{j\ne i}\sum_{k\ne i}G_{ji}G_{ki}R_{\Delta j}R_{\Delta k}.
\end{align*}

The expression stated in the equation combines these expressions. 
\end{proof}

\subsection{Further Details for Power} \label{sec:details_power}

\subsubsection{Further Analytic Results} \label{sec:further_analytics_power}

\begin{lemma} \label{lem:max_inv}
$\left(s_{1}^{\prime},s_{2}^{\prime}\right)^{\prime}$ are sufficient statistics for $\left(\pi_{Y}^{\prime},\pi^{\prime}\right)^{\prime}$.
Further, for transformations of the form $Z \rightarrow ZF^\prime$ where $F$ is a $K \times K$ orthogonal matrix, $\left(s_{1}^{\prime}s_{1},s_{1}^{\prime}s_{2},s_{2}^{\prime}s_{2}\right)$ is a maximal invariant, and 
\[
\left(\begin{array}{c}
s_{1}\\
s_{2}
\end{array}\right)\sim N\left(\left(\begin{array}{c}
\left(Z^{\prime}Z\right)^{1/2}\pi_{Y}\\
\left(Z^{\prime}Z\right)^{1/2}\pi
\end{array}\right),\Omega\otimes I_{K}\right).
\]
\end{lemma}

\begin{proposition} \label{prop:max_inv_distr}
With \Cref{eqn:rf_hom_mod}, $K\rightarrow\infty$ and $\frac{1}{\sqrt{K}}\left(\pi_{Y}^{\prime}Z^{\prime}Z\pi_{Y},\pi^{\prime}Z^{\prime}Z\pi_{Y},\pi^{\prime}Z^{\prime}Z\pi\right)\rightarrow\left(C_{YY},C_{Y},C_{S}\right)$,
\begin{equation} \label{eqn:max_inv_distr}
\frac{1}{\sqrt{K}}\left(\begin{array}{c}
s_{1}^{\prime}s_{1}-K\omega_{\zeta \zeta}-C_{YY}\\
s_{1}^{\prime}s_{2}-K\omega_{\zeta\eta}-C_{Y}\\
s_{2}^{\prime}s_{2}-K\omega_{\eta \eta} -C_{S}
\end{array}\right)\xrightarrow{d}N\left(\left(\begin{array}{c}
0\\
0\\
0
\end{array}\right),\Sigma\right)
\end{equation}
for some variance matrix $\Sigma$. If $C_{YY}, C_Y, C_S < \infty$,
\[
\Sigma=\left(\begin{array}{ccc}
2\omega_{\zeta \zeta}^{2} & 2\omega_{\zeta\eta}\omega_{\zeta \zeta} & 2\omega_{\zeta\eta}^{2}\\
2\omega_{\zeta\eta}\omega_{\zeta \zeta} & \omega_{\zeta \zeta} \omega_{\eta \eta} +\omega_{\zeta\eta}^{2} & 2\omega_{\zeta\eta}\omega_{\eta \eta}\\
2\omega_{\zeta\eta}^{2} & 2\omega_{\zeta\eta} \omega_{\eta\eta} & 2 \omega_{\eta \eta}^2
\end{array}\right).
\]
\end{proposition}

The proof of \Cref{prop:max_inv_distr} relies on $K\rightarrow\infty$ because objects like $s_{1}^{\prime}s_{1}$ can be written as a sum of $K$ objects.
With an appropriate representation to obtain independence, a CLT can be applied to yield normality. 
Compared to MS22, \Cref{prop:max_inv_distr} does not require constant treatment effects and characterizes the distribution without orthogonalizing the sufficient statistics. 
Nonetheless, the form of the covariance matrix is similar to MS22.

Instead of using the maximal invariant, we use the L1O analog
\begin{align*}
(T_{YY}, T_{YX}, T_{XX}) := \frac{1}{\sqrt{K}} \sum_i \sum_{j\ne i} P_{ij} (Y_i Y_j, Y_i X_j, X_i X_j),
\end{align*} 
which relates to the JIVE directly as $\hat{\beta}_{JIVE} = T_{YX}/T_{XX}$, so the asymptotic problem is: 
\begin{equation} \label{eqn:asymp_prob}
\left(\begin{array}{c}
T_{YY}\\
T_{YX}\\
T_{XX}
\end{array}\right)\sim N\left(\mu,\Sigma\right),\mu=\left(\begin{array}{c}
\frac{1}{\sqrt{K}}\sum_{i}\sum_{j\ne i}P_{ij}R_{Yi}R_{Yj}\\
\frac{1}{\sqrt{K}}\sum_{i}\sum_{j\ne i}P_{ij}R_{Yi}R_{j}\\
\frac{1}{\sqrt{K}}\sum_{i}\sum_{j\ne i}P_{ij}R_{i}R_{j}
\end{array}\right),\Sigma=\left(\begin{array}{ccc}
\sigma_{11} & \sigma_{12} & \sigma_{13}\\
\cdot & \sigma_{22} & \sigma_{23}\\
\cdot & \cdot & \sigma_{33}
\end{array}\right).
\end{equation}

With abuse of notation, $\mu$ and $\Sigma$ refer to the asymptotic mean and variances of $(T_{YY}, T_{YX}, T_{XX})$ instead of $(T_{AR}, T_{LM}, T_{FS})$ so that the statistics do not depend on the hypothesized null, but these are identical when $\beta_0 =0$. 
Using the same argument as \Cref{sec:power}, $\mu_1, \mu_3 \geq 0$ and $\mu_2^2 \leq \mu_1 \mu_3$.

Even with covariates, if the regression is fully saturated with $G$ given by UJIVE, \Cref{prop:UJIVE_restr} below shows that the same inequality restrictions hold. 
To be precise, saturation is defined in Section 2 of \citet{evdokimov2018inference}.
All individuals can be partitioned into $L$ covariate groups, so with group index $G_i \in \{1, \cdots, L \}$, we have covariates $W_{i,l} = 1\{ G_i  = l\}$. 
We also have an instrument $S_i$ that takes $M+1$ possible values in each group, and these values for every group $l$ are labeled $s_{l0}, \cdots, s_{lM}$.
Then, the vector of instruments has dimension $K = M L$ and $Z_{i, lm} = 1\{ S_i  = s_{lm}\}$. 
Adapting \Cref{eqn:asymp_prob} to the case with covariates,
\begin{align*}
\mu=\left(
\frac{1}{\sqrt{K}}\sum_{i}\sum_{j\ne i}G_{ij}R_{Yi}R_{Yj},
\frac{1}{\sqrt{K}}\sum_{i}\sum_{j\ne i}G_{ij}R_{Yi}R_{j},
\frac{1}{\sqrt{K}}\sum_{i}\sum_{j\ne i}G_{ij}R_{i}R_{j}\right)^\prime.
\end{align*}
\begin{proposition} \label{prop:UJIVE_restr}
If $\mu$ is defined such that $G=\left(I-\diag\left(H_{Q}\right)\right)^{-1}H_{Q}-\left(I-\diag\left(H_{W}\right)\right)^{-1}H_{W}$, and the regression is fully saturated, then $\mu_{1} \geq 0,\mu_{3}\geq0$ and $\mu_{2}^{2}\leq\mu_{1}\mu_{3}$. 
\end{proposition}

Using the asymptotic problem of \Cref{eqn:asymp_prob}, testing $H_{0}:\mu_{2}/\mu_{3}=\beta^{*}$ is identical to testing $H_{0}:\mu_{2}-\beta^{*}\mu_{3}=0$. 
Since $\beta^{*}$ is fixed, and I consider alternatives of the form: $H_{A}:\mu_{2}-\beta^{*}\mu_{3}=h_{A}$.
The LM statistic corresponds to $T_{YX}-\beta^{*}T_{XX}$, so it can be used to test the null directly. 
I focus on the most common case of $\beta^{*}=0$, and it is analogous to extend the argument for $\beta^{*}\ne0$. 
Let $\mu^{A}$ denote the mean under the alternative and $\mu^{H}$ under the null.
The remainder of this section presents theoretical results for power, and numerical results beyond the environment covered by theory are relegated to \Cref{sec:more_power_curves}.

The one-sided test is the most powerful test for testing against a particular subset of alternatives $\mathcal{S} :=\left\{ \left(\mu_1^A, \mu_2^A, \mu_3^A \right):\mu_1^A-\frac{\sigma_{12}}{\sigma_{22}}\mu_2^A\geq0, \mu_3^A-\frac{\sigma_{23}}{\sigma_{22}}\mu_2^A\geq0\right\}$.
While $\mathcal{S}$ may not be empirically interpretable, this set is constructed so that standard \citet{lehmann2005testing} arguments can be applied to conclude that the one-sided LM test is the most powerful test. 
The proposition makes no statement about alternative hypotheses that are not in $\mathcal{S}$. 
A more powerful test can be constructed when $\mu_{2}^{A}$ is large and covariance $\sigma_{23},\sigma_{12}$ are large. 

\begin{proposition} \label{prop:lm_opt}
The one-sided LM test is the most powerful test for testing any alternative hypothesis $\left(\mu_1^A, \mu_2^A, \mu_3^A \right)\in\mathcal{S}$ in the asymptotic problem of \Cref{eqn:asymp_prob}.
\end{proposition}

For a given $\left(\mu_{1}^{A},\mu_{2}^{A},\mu_{3}^{A}\right)$ in the alternative space, LM (which just uses the second element) is justified as being most powerful because it is identical to the Neyman-Pearson test when testing against a point null $\mu^H$ with $\mu_{1}^{H}=\mu_{1}^{A}-\frac{\sigma_{12}}{\sigma_{22}}\mu_{2}^{A}$, $\mu_2^H=0$ and $\mu_{3}^{H}=\mu_{3}^{A}-\frac{\sigma_{23}}{\sigma_{22}}\mu_{2}^{A}$.
The inequalities in $\mathcal{S}$ are imposed so that $\mu_1^H, \mu_3^H \geq 0$, ensuring that $\mu^H$ is in the null space, so LM is the most powerful test. 
In contrast, if the inequalities fail in the alternative space, then $(\mu_{1}^{A}-\frac{\sigma_{12}}{\sigma_{22}}\mu_{2}^{A},0,\mu_{3}^{A}-\frac{\sigma_{23}}{\sigma_{22}}\mu_{2}^{A})$ is not in the null space, and the \citet{lehmann2005testing} argument cannot be applied.

\subsubsection{Existence of Structural Model}

This section presents a structural model, then argues that any reduced-form model in the form of \Cref{eqn:asymp_prob} can be justified by this structural model.

\begin{example} \label{ex:contX}
Consider a linear potential outcomes model with an instrument $Z$ that is a vector of indicators for judges, each with $c=5$ cases, a continuous endogenous variable $X$, and outcome $Y$:
\begin{equation} \label{eqn:sf}
\begin{split}
X_{i}(z)  =z^{\prime}\pi+v_{i}, \qquad 
Y_{i}(x)  =x\left(\beta+\xi_{i}\right)+\varepsilon_{i}, \text{ and } \\
\left(\begin{array}{c}
\varepsilon_{i}\\
\xi_{i}\\
v_{i}
\end{array}\right) \mid k(i) = k
\sim N\left(\left(\begin{array}{c}
0\\
0\\
0
\end{array}\right),\left(\begin{array}{ccc}
\sigma_{\varepsilon\varepsilon} & \sigma_{\varepsilon\xi} & \sigma_{\varepsilon v}\\
\cdot & \sigma_{\xi\xi} & \sigma_{\xi v k}\\
\cdot & \cdot & \sigma_{vv}
\end{array}\right)\right).
\end{split}
\end{equation}
Due to the judge design, $X_{i}=\pi_{k(i)}+v_{i}$, where $k(i)$ is the judge that observation $i$ is assigned to. 
The strength of the instrument is $C_S=\frac{1}{\sqrt{K}}\sum_{k} (c-1) \pi_{k}^{2}$.
The $\pi_k$'s are constructed as such: with $s = \sqrt{C_S/\sqrt{K}/(c-1)}$, set $\pi_k=0$ for the base judge, $\pi_k=-s$ for half the judges and $\pi_k=s$ for the other half.
The heterogeneity covariances $\sigma_{\xi vk}$ are constructed so that $\sum_{k}\pi_{k}=0,\sum_{k}\sigma_{\xi vk}=0$, and $\sum_{k}\pi_{k}\sigma_{\xi vk}=0$.
With $C_H$ characterizing the heterogeneity in the model, and $h = \sqrt{C_H/\sqrt{K}/(c-1)}$, set $\sigma_{\xi v k}=0$ of the base judge; among judges with $\pi_{k}=s$, half of them have $\sigma_{\xi vk}=h$ and the other half $\sigma_{\xi vk}=-h$. 
The same construction of $\sigma_{\xi vk}$ applies for judges with $\pi_{k}=-s$. 
\end{example}

In this model, the individual treatment effect is $\beta_{i}=\beta+\xi_{i}$.
We can interpret $v_{i}$ as the noise associated with the first-stage regression, $\varepsilon_{i}$ as the noise in the intercept of the outcome equation, and $\xi_{i}$ as the individual-level treatment effect heterogeneity.
Further, $\sigma_{\xi vk}$ characterizes the extent of treatment effect heterogeneity. 
The observed outcome in a model with constant treatment effects is $Y_i(X_i) = X_i\beta + \check{\varepsilon}_i$, with $E[\check{\varepsilon}_i]$=0.
When $\sigma_{\xi vk}=0$, regardless of the values of $\sigma_{\varepsilon \xi}, \sigma_{\xi \xi}$, the observed outcome of \Cref{eqn:sf} can be written as $Y_i(X_i) =  X_i\beta + \check{\varepsilon}_i$ where $E[\check{\varepsilon}_i] = E[X_i \xi_i +\varepsilon_i] = E[X_i E[\xi_i\mid X_i]] =0$, which resembles the constant treatment effect case.

\begin{lemma} \label{lem:RF_in_SF}
Consider the model of \Cref{ex:contX}. If $\sqrt{K}s^{2}\rightarrow \tilde{C}_{S}<\infty$
and $\sqrt{K}h^{2}\rightarrow \tilde{C}_{H}<\infty$, then
\begin{align*}
\sigma_{11} & =\frac{4}{\sigma_{33}}\left(\sigma_{22}-\frac{\sigma_{23}^{2}}{2\sigma_{33}}\right)^{2}+o(1),\quad
\sigma_{12} =2\frac{\sigma_{23}}{\sigma_{33}}\left(\sigma_{22}-\frac{\sigma_{23}^{2}}{2\sigma_{33}}\right)+o(1), \quad
\sigma_{13}  =\frac{\sigma_{23}^{2}}{\sigma_{33}}+o(1), \\
\sigma_{22} & =\frac{c-1}{c}\left(\sigma_{vv}\left(\sigma_{\varepsilon\varepsilon}+\sigma_{vv}\beta^{2}+\sigma_{vv}\sigma_{\xi\xi}+2\sigma_{\varepsilon v}\beta\right)+\left(\sigma_{vv}\beta+\sigma_{\varepsilon v}\right)^{2}\right)+o(1),\\
\sigma_{33}&=2\frac{c-1}{c}\sigma_{vv}^{2}+o(1), \quad 
\sigma_{23}=2\frac{c-1}{c}\sigma_{vv}\left(\sigma_{vv}\beta+\sigma_{\varepsilon v}\right)+o(1), \text{ and } \\
\left(\begin{array}{c}
\mu_{1}\\
\mu_{2}\\
\mu_{3}
\end{array}\right) & =\left(\begin{array}{c}
\sqrt{K}\left(c-1\right)\left(s^{2}\beta^{2}+h^{2}\right)\\
\sqrt{K}\left(c-1\right)s^{2}\beta\\
\sqrt{K}\left(c-1\right)s^{2}
\end{array}\right)=\left(c-1\right)\left(\begin{array}{c}
C_{S}\beta^{2}+C_{H}\\
C_{S}\beta\\
C_{S}
\end{array}\right).
\end{align*}
\end{lemma}

\begin{proposition} \label{prop:exist_SF}
In the model of \Cref{ex:contX} with $\sqrt{K}s^{2}\rightarrow \tilde{C}_{S}<\infty$
and $\sqrt{K}h^{2}\rightarrow \tilde{C}_{H}<\infty$, for any $\sigma_{22},\sigma_{23},\sigma_{33}$
such that $\sigma_{22},\sigma_{33}>0$, $\sigma_{23}^{2}\leq\sigma_{22}\sigma_{33}$
and $\mu$ such that $\mu_{1} \geq 0,\mu_{3}>0$, $\mu_{2}^{2}\leq\mu_{1}\mu_{3}$, the following values of structural parameters:
\begin{align*}
\tilde{C}_{S} & =\mu_{3}/\left(c-1\right), \quad 
\beta =\mu_{2}/\mu_{3}, \quad 
h  =\sqrt{\frac{1}{\sqrt{K}}\frac{1}{c-1}\left(\mu_{1}-\frac{\mu_{2}^{2}}{\mu_{3}}\right)},\\
\Sigma_{SF} & =\left(\begin{array}{ccc}
\sigma_{\varepsilon\varepsilon} & \sigma_{\varepsilon\xi} & \sigma_{\varepsilon v}\\
. & \sigma_{\xi\xi} & \sigma_{\xi vk}\\
. & . & \sigma_{vv}
\end{array}\right)=\left(\begin{array}{ccc}
\frac{1}{\sigma_{vv}}\frac{c}{c-1}\left(\sigma_{22}-\frac{\sigma_{23}^{2}}{\sigma_{33}}\right)+\frac{\sigma_{\varepsilon v}^{2}}{\sigma_{vv}} & 0 & \sigma_{\varepsilon v}\\
. & \frac{h}{\sigma_{vv}} & \pm h\\
. & . & \sigma_{vv}
\end{array}\right),\\
\sigma_{vv} & =\sqrt{\frac{\sigma_{33}c}{2\left(c-1\right)}}, \text{ and } \quad 
\sigma_{\varepsilon v} =\frac{1}{\sigma_{vv}}\left(\frac{\sigma_{23}c}{2\left(c-1\right)}-\sigma_{vv}^{2}\beta\right),
\end{align*}
satisfy the equations in \Cref{lem:RF_in_SF}, and $det\left(\Sigma_{SF}\right)/h\rightarrow C_{D}\geq0$. 
\end{proposition}

Due to \Cref{prop:exist_SF}, since the principal submatrices of $\Sigma_{SF}$ are positive semidefinite asymptotically, $\Sigma_{SF}$ is a symmetric positive semidefinite matrix asymptotically. 
The proposition thus implies that when the $\sigma$'s and $\mu$ satisfy the conditions, there exists structural parameters that can generate the given $\mu$ and $\Sigma$ asymptotically. 
Hence, there are no further restrictions on $\mu$ from the observed $\Sigma$ in \Cref{ex:contX}.

\subsubsection{Numerical Results for Power} \label{sec:more_power_curves}
This section presents numerical results for power in environments not covered by the theory. 
I first consider one-sided tests beyond the set $\mathcal{S}$ covered by the theory, then weighted average power for two-sided tests rather than the class of unbiased tests. 

The power envelope is achieved by a test that is valid across the entire composite null space, and is most powerful for testing against a particular point in the alternative space.
To obtain this test, I implement the algorithm from \citet{elliott2015nearly} (EMW) where all weight on the alternative are placed on a single point while being valid across a composite null. 
Then, testing against every point in the alternative space requires a different critical value.
For the numerical exercises in this subsection, I use a $\Sigma$ matrix of the form:
\begin{equation}
\Sigma=\left(\begin{array}{ccc}
2 & 2\rho & 2\rho^2\\
\cdot & 1+\rho^2 & 2\rho\\
\cdot & \cdot & 2
\end{array}\right),
\end{equation}
which corresponds to the $\Sigma$ matrix in \Cref{prop:max_inv_distr} with $\omega_{\zeta\zeta}=\omega_{\eta\eta}=1,\omega_{\zeta \eta} = \rho$.

In the numerical exercises, I display the rejection rate across 500 independent draws from $X^* \sim N(\mu,\Sigma)$ at each point on the $\mu_2$ axis, across several $\mu_1,\mu_3$ values for a 5\% test. 
The composite null uses a grid of $\mu_1 \in [0,5], \mu_3\in [0,5]$ in 0.5 increments, and assumes the variance is known. 

\Cref{fig:os_adv} uses a one-sided LM test, with a large covariance at $\rho=0.9$. 
When data is generated from the null, since LM and EMW are valid tests, their rejection rate is at most 0.05.
EMW has exact size when testing a weighted average of values in the null space and is valid across the entire space, so when data is generated from one particular point in the null, EMW can be conservative. 
Consistent with \Cref{prop:lm_opt}, when $\mu_{2}$ is small enough for $\mu_1=1,\mu_3=4$, LM achieves the power envelope, but as $\mu_{2}$ gets larger, the gap widens substantially. 
This phenomenon occurs because EMW still uses the same null grid, but now it no longer needs to have correct size for testing against the point $(\mu_{1}^{A}-\frac{\sigma_{12}}{\sigma_{22}}\mu_{2}^{A},0,\mu_{3}^{A}-\frac{\sigma_{23}}{\sigma_{22}}\mu_{2}^{A})$, as that point is no longer in the null space. 

In \Cref{fig:os_emp}, $\Sigma$ is calibrated by using the $\Sigma$ matrix calculated from the \citet{angrist1991does} application, so after appropriate normalizations, $\rho=0.37$.
With such a low covariance, LM is basically indistinguishable from the EMW bound.
Hence, even though there are gains to be made theoretically, in the empirical application considered, the gains are small. 

Instead of considering a point alternative, we may be more interested in testing against a composite alternative. 
Here, the alternative grid for EMW places equal weight on alternatives $(\mu_1^A, \mu_2^A, \mu_3^A) \in [0,5] \times [-2,2] \times [0,5]$ in increments of 0.5 (excluding $\mu_2=0$) subject to inequality constraints. 
Figures \ref{fig:uw_adv} and \ref{fig:uw_emp} present one such possibility by allowing EMW to place equal weight on several points within the alternative space.
The resulting test is the nearly optimal test for a weighted average of values the null space against the uniformly weighted average of alternative values.
Hence, there is no guarantee that its power is necessarily higher than the LM test at every point in the alternative space.
While there are weighted-average power curves that substantially outperform LM, this result is compatible with \Cref{prop:umpu}. 
EMW is a biased test as there are points in the alternative space that are not a part of the grid where LM outperforms EMW. 
Nonetheless, \Cref{fig:uw_emp} suggests that, when using the empirical covariance, LM does not perform substantially worse than EMW.

\begin{figure}
    \centering
    \caption{One-sided test with $\rho=0.9$}
    \label{fig:os_adv}
    \includegraphics[scale=0.45]{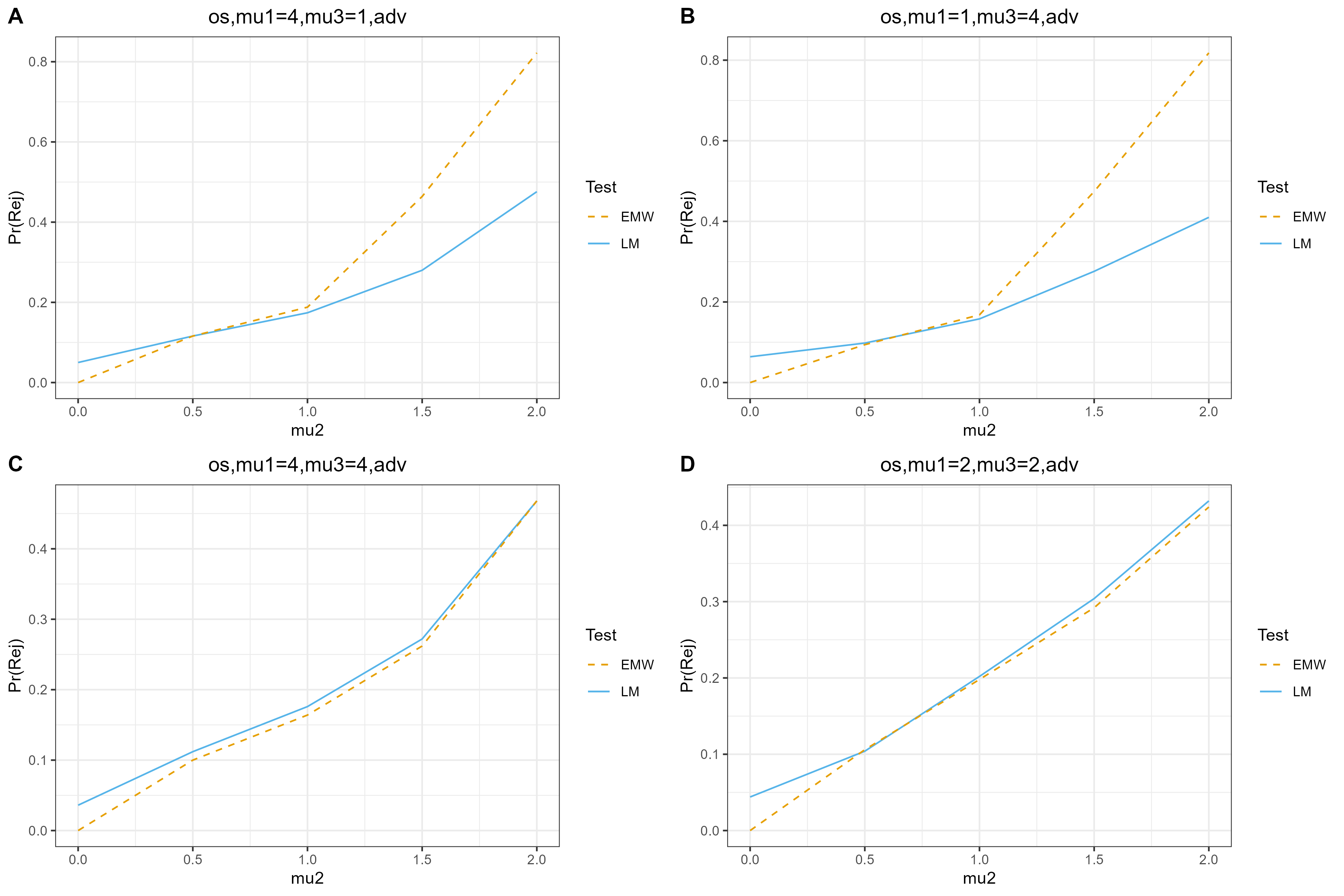}\\
\end{figure}

\begin{figure}
    \centering
    \caption{One-sided test with $\rho=0.37$}
    \label{fig:os_emp}
    \includegraphics[scale=0.45]{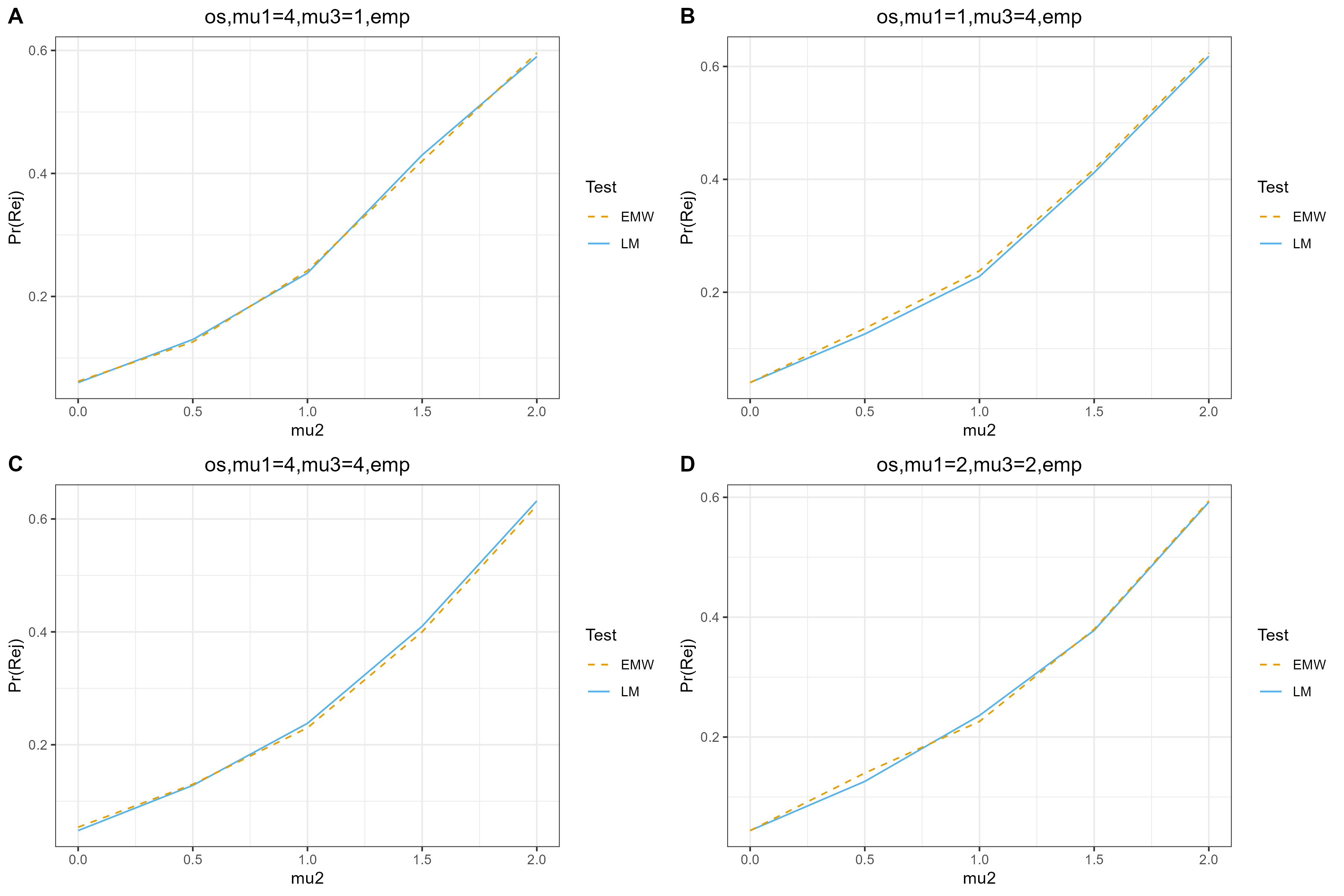}\\
\end{figure}



\begin{figure}
    \centering
    \caption{Uniform Weighting on grid of alternatives with $\rho=0.9$}
    \label{fig:uw_adv}
    \includegraphics[scale=0.45]{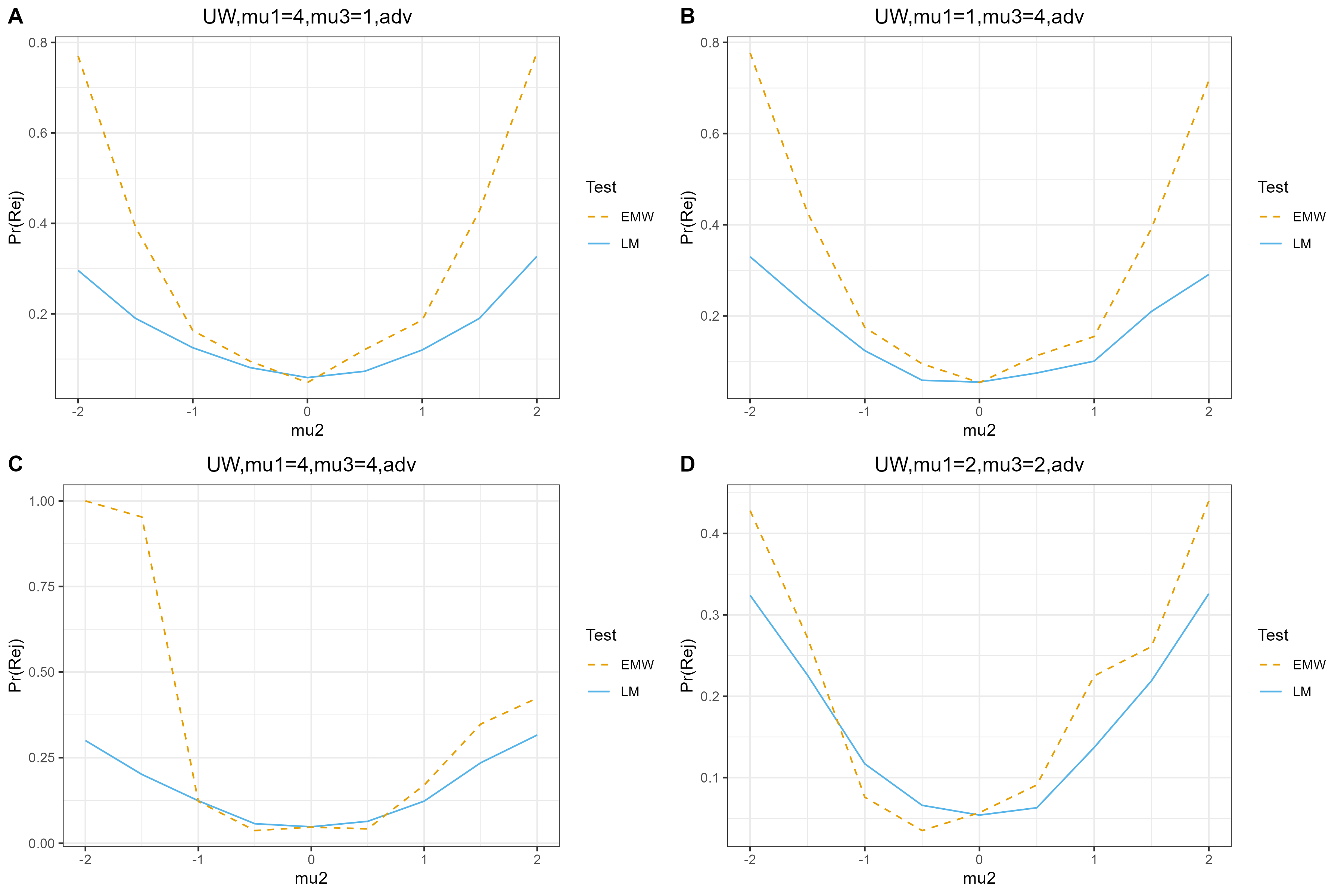}\\
\end{figure}

\begin{figure}
    \centering
    \caption{Uniform Weighting on grid of alternatives with $\rho=0.37$}
    \label{fig:uw_emp}
    \includegraphics[scale=0.45]{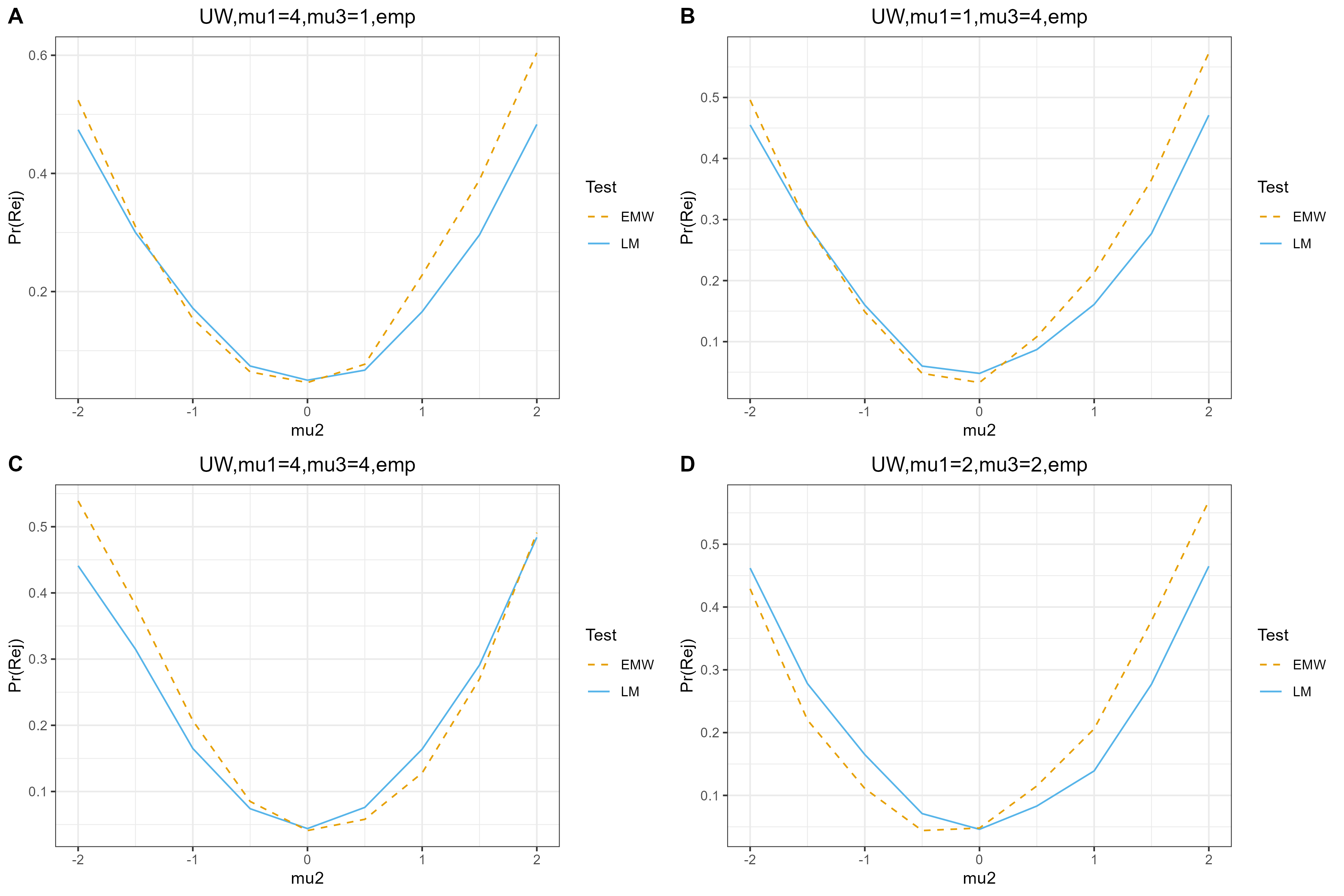}\\
\end{figure}

\subsection{Constructing Confidence Sets} \label{sec:implementation}

Expressions for the test are given in \Cref{sec:valid_inference}, which can be efficiently implemented using matrix operations.
Inverting the test to obtain a confidence set is also straightforward in this procedure, as the bounds of the confidence set are derived in closed-form in this section.

To invert the LM test to obtain a confidence set, use $e_{i}=Y_{i}-X_{i}\beta_{0}$ and expand the $A$ expressions in \Cref{eqn:Vhat_LM} so that they are written in terms of $X$ and $Y$. 
The two-sided test rejects: $\left(\sum_{i}\sum_{j\ne i}G_{ij}e_{i}X_{j}\right)^{2}/\hat{V}_{LM}\geq q = \Phi(1-\alpha/2)^2$.
Let $T_{YX}:=\frac{1}{\sqrt{K}}\sum_{i}\sum_{j\ne i}G_{ij}Y_{i}X_{j}$. 
Then, $\sum_{i}\sum_{j\ne i}G_{ij}e_{i}X_{j}= \sqrt{K} \left( T_{YX}- T_{XX}\beta_{0} \right)$, so squaring it results in a term that is quadratic in $\beta_{0}^{2}$.
With $\hat{V}_{LM}=B_{0}+C_{1}\beta_{0}+B_{2}\beta_{0}^{2}$ quadratic in $\beta_0$, the analysis for the shape of the confidence intervals is similar to the AR procedure for just-identified IV. 
Calculations for coefficients is similar to that of the L3O variance.

\begin{proposition} \label{lem:CI}
The two-sided LM test does not reject $\beta_0$ when $\left(K T_{XX}^{2}-qB_{2}\right)\beta_{0}^{2}-\left(2KT_{YX}T_{XX}+qB_{1}\right)\beta_{0}+\left(KT_{YX}^{2}-qB_{0}\right)\leq0$.
Let $D:=\left(2KT_{YX}T_{XX}+qB_{1}\right)^{2}-4\left(KT_{XX}^{2}-qB_{2}\right)\left(KT_{YX}^{2}-qB_{0}\right)$.
If $D\geq0$ and $KT_{XX}^{2}-qB_{2}\geq0$, then the upper and lower
bounds of confidence set are:
\[
\frac{\left(2KT_{YX}T_{XX}+qB_{1}\right)\pm\sqrt{D}}{2\left(KT_{XX}^{2}-qB_{2}\right)}.
\]
If $D<0$ and $KT_{XX}^{2}-qB_{2}<0$, then the confidence set is empty.
Otherwise, the confidence set is unbounded. 
\end{proposition}

Due to $+qB_{1},-qB_{2}$ in the expression of the upper and lower bounds, the confidence set is not necessarily centered around $\hat{\beta}_{JIVE}=T_{YX}/T_{XX}$. 

\subsection{Further Simulation Results} \label{sec:further_sim}
This section reports simulation results from several structural models to assess how well various procedures control for size.
Since the nominal size is 0.05, and data is generated under the null, the target rejection rate is 0.05.
Across the board, the L3O method performs well, and for all existing procedures, there is at least one design where they perform badly. 

\subsubsection{Continuous Treatment}
This subsection reports results for a simulation based on \Cref{ex:contX} that has a continuous $X$. 
\Cref{tab:sim_cont} reports results with $K=500$ and \Cref{tab:sim_cont_K40} reports results for $K=40$.
The L3O rejection rates are closer to the nominal rate than the existing procedures in the literature, albeit worse with a smaller $K$. 
ARorc has high rejection rates with strong heterogeneity and EK has high rejection rates with weak instruments. 
Notably, with perfect correlation and an irrelevant instrument, EK can achieve 100\% rejection in the simulation with $K=500$.
The procedures that use the LM statistic are MO, $\tilde{X}$-AR, L3O and LMorc; they differ only in variance estimation.
Hence, while $\tilde{X}$-AR and MO over-reject, the extent of over-rejection is smaller than ARorc and EK in the adversarial cases.

\begin{table}
    \centering
    \caption{Rejection rates under the null for nominal size 0.05 test for continuous $X$}
    \label{tab:sim_cont}
    \include{fig/simf_contX_res.tex}
    \justifying \small
    Notes: Data generating process corresponds to \Cref{ex:contX}. 
    Unless mentioned otherwise, simulations use $K=500,c=5,\beta=0,\sigma_{\varepsilon \varepsilon}= \sigma_{vv}=1,  \sigma_{\varepsilon \xi} =0, =\sigma_{\varepsilon v} = 0.8, \sigma_{\xi \xi} = 1+h$ for $h^2 <1$ with 1000 simulations. 
    The table displays rejection rates of various procedures (in columns) for various designs (in rows).
    $C_H=0$ uses $\xi_i=0$ for all $i$, which uses $\sigma_{\xi\xi}=\sigma_{\varepsilon\xi}=\sigma_{\xi v}=0$, corresponding to constant treatment effects.
    Procedures are described in \Cref{tab:sim_mte_base}.
\end{table}

\begin{table}
    \centering
    \caption{Rejection Rates under the null for nominal size 0.05 test for Continuous $X$ with $K=40$}
    \include{fig/simf_contX_res_K40}
    \label{tab:sim_cont_K40}
    \justifying \small
    Note: Designs are identical to \Cref{tab:sim_cont}, but $K=40$ here.
\end{table}

\subsubsection{Binary Treatment} \label{sec:binary_sim}
This subsection presents a structural model with a binary $X$. 
Data is generated from a judge model with $J=K+1$ judges, each with $c=5$ cases, and cases are indexed by $i$. 
The structural model is:
\begin{align*}
Y_{i}(x) & =x(\beta + \xi_i)+\varepsilon_{i}, \text{ and }\\
X_{i}(z) & =I\left\{ z^{\prime}\pi-v_{i}\geq0\right\}.
\end{align*}
Our unobservables are generated as follows.
Draw $v_{i}\sim U[-1,1]$, then generate residuals from:
\[
\varepsilon_{i}\mid v_{i}\sim\begin{cases}
\begin{array}{c}
N\left(\sigma_{\varepsilon v},\sigma_{\varepsilon \varepsilon}\right)\\
N\left(-\sigma_{\varepsilon v},\sigma_{\varepsilon \varepsilon} \right)
\end{array} & \begin{array}{c}
if\\
if
\end{array}\begin{array}{c}
v_{i}\geq0\\
v_{i}<0
\end{array}\end{cases},
\]
\[
\xi_{i} \mid v_i \geq 0=\begin{cases}
\begin{array}{c}
\sigma_{\xi vk}\\
-\sigma_{\xi vk}
\end{array} & \begin{array}{c}
w.p.\\
w.p.
\end{array}\begin{array}{c}
p\\
1-p
\end{array}\end{cases}, \text{ and} \quad  
\xi_{i} \mid v_i <0=\begin{cases}
\begin{array}{c}
\sigma_{\xi vk}\\
-\sigma_{\xi vk}
\end{array} & \begin{array}{c}
w.p.\\
w.p.
\end{array}\begin{array}{c}
1-p\\
p
\end{array}\end{cases}.
\]

The process for determining $s,h$ and $\pi_k \in \{ 0,-s,s \},\sigma_{\xi vk} \in \{ 0,-h,h \}$ are identical to \Cref{ex:contX}, as $s$ controls the strength of the instrument, $h$ the extent of heterogeneity, and $\beta$ is the object of interest.
Then, the problem's variances and covariances are determined by $\left(p,\sigma_{\varepsilon v},\sigma_{\varepsilon \varepsilon}\right)$.
The JIVE estimand is shown to be $\beta$ in Supplementary Appendix E.2.
A simulation is run with $K=100$, so the sample size is smaller than the normal experiment in \Cref{ex:contX}.


\begin{table}
    \centering
    \caption{Rejection Rates under the null for nominal size 0.05 test for binary $X$}
    \include{fig/simf_binX_res}
    \label{tab:sim_bin}
    \justifying \small
    Note: The data generating process corresponds to \Cref{sec:binary_sim}.
    Unless stated otherwise, designs use $K=100,c=5,\beta=0,p=7/8$,
$\sigma_{\varepsilon \varepsilon}=0.1,\sigma_{\varepsilon v}=0.5$ with 1000 simulations. 
\end{table}

Results are presented in \Cref{tab:sim_bin}, and are qualitatively similar to \Cref{sec:challenges}. 
The oracle test consistently obtains rejection rates close to the nominal 5\% rate across all designs, in accordance with the normality result, even with heterogeneous treatment effects and non-normality of errors due to the binary setup. 
The L3O rejection rate is close to the nominal rate even with a smaller sample size. 
EK, ARorc and MO continue to have high rejection rates in the adversarial designs.

\subsubsection{Incorporating Covariates} \label{sec:covar_sim}
This section presents a data-generating process that involves covariates. 
Instead of judges, consider a model where there are $K$ states. 
Let $t=1,\cdots, K$ index the state and let $W$ denote the control vector that is an indicator for states.
With a binary exogenous variable (say an indicator for birth being in the fourth quarter) $B \in \{ 0,1 \}$, the value of the instrument is given by $k = t \times B$.
Then, the instrument vector $Z$ is an indicator for all possible values of $k$.
The structural model is:
\begin{align*}
Y_{i}(x) & =x(\beta + \xi_i)+w^\prime \gamma + \varepsilon_{i}, \text{ and }\\
X_{i}(z) & =I\left\{ z^{\prime}\pi + w^\prime \gamma  -v_{i}\geq0\right\}. 
\end{align*}

\begin{table}
    \centering
    \caption{Rejection Rates under the null for nominal size 0.05 test for binary $X$ with covariates}
    \include{fig/simf_binXcov_res}
    \label{tab:sim_bincov}
    \justifying \small
    Note: The data generating process corresponds to \Cref{sec:covar_sim}. 
    Unless stated otherwise, designs use $K=48,c=5,\beta=0,p=7/8$,
$\sigma_{\varepsilon \varepsilon}=0.5,\sigma_{\varepsilon v}=0.1$, and $g=0.1$ with 1000 simulations. 
\end{table}

In the simulation, every state has 10 observations, of which 5 have $B=1$ and the other 5 have $B=0$. 
The process for generating $(v_i, \varepsilon_i, \xi_i)$, $\pi_{k}, \sigma_{\xi vk}$, and $s,h$ is identical to the binary case. 
Hence, $\pi_0=\sigma_{\xi v0}$ for the base group, which constitutes half the observations. 
For $k\ne 0$, $\pi_k$ is the coefficient for observations from state $t=k$ and have $B=1$, and $\sigma_{\xi vk}$ is the corresponding heterogeneity term. 
Whenever $\pi_{t}=s$, set $\gamma_t=g$; whenever $\pi_{t}=-s$, set $\gamma_t=-g$.
In this setup, it can be shown that the UJIVE estimand is $\beta$, and the proof is in Supplementary Appendix E.2.
\Cref{tab:sim_bincov} reports the associated simulation results, which are qualitatively similar to the results described before.

\subsection{Proofs of Lemmas 1 and 2}

\begin{proof}[Proof of \Cref{lem:MS_estimand}]
Suppose not. Then, for some real $\beta_{0}$,
\begin{align*}
E\left[T_{AR}\right] & =\sum_{i}\sum_{j\ne i}P_{ij}R_{\Delta i}R_{\Delta j}
=\sum_{i}\sum_{j\ne i}P_{ij}\left(R_{Yi}R_{Yj}-R_{i}R_{Yj}\beta_{0}-R_{Yi}R_{j}\beta_{0}+R_{i}R_{j}\beta_{0}^{2}\right)=0.
\end{align*}

Solving for $\beta_{0}$,
\small
\[
\beta_{0}=\frac{2\sum_{i}\sum_{j\ne i}P_{ij}R_{i}R_{Yj}\pm\sqrt{4\left(\sum_{i}\sum_{j\ne i}P_{ij}R_{i}R_{Yj}\right)^{2}-4\left(\sum_{i}\sum_{j\ne i}P_{ij}R_{i}R_{j}\right)\left(\sum_{i}\sum_{j\ne i}P_{ij}R_{Yi}R_{Yj}\right)}}{2\left(\sum_{i}\sum_{j\ne i}P_{ij}R_{i}R_{j}\right)}.
\]
\normalsize

In our structural model, $R_{i}=\pi_{k(i)}$ and $R_{Yi}=\pi_{Yk(i)}$.
The term in the square root can be written as:
\[
D=4\left(\sum_{k}\pi_{k}\pi_{Yk}\right)^{2}-4\left(\sum_{k}\pi_{k}^{2}\right)\left(\sum_{k}\pi_{Yk}^{2}\right)
\]

Using \Cref{tab:parRFmte}, $\sum_k \pi_k^2 = \frac{5}{8}s^2 K$, $\sum_k \pi_{Yk}^2 = \left( \frac{5}{8} s^2\beta^2 +h^2 \right) K$, and $\sum_k \pi_k \pi_{Yk} = \frac{5}{8}s^2 \beta K$, we obtain 
\begin{align*}
\frac{1}{4} D = \left( \frac{5}{8}s^2 \beta K \right)^2 - \left( \frac{5}{8}s^2 K \right) \left( \frac{5}{8} s^2\beta^2 +h^2 \right) K = - \frac{5}{8}s^2 h^2 K^2 \leq 0.
\end{align*}
Since $h\ne0$ and $Ks^2 >0$, there are no real roots of $\beta_{0}$, a contradiction. 
\end{proof}

\begin{proof} [Proof of \Cref{lem:CLT}]
I rewrite the quadratic term to produce a martingale difference array: 
\begin{align*}
\sum_{i}\sum_{j\ne i}G_{ij}v_{i}^{\prime}Av_{j} & =\sum_{i}\sum_{j<i}G_{ij}v_{i}^{\prime}Av_{j}+\sum_{i}\sum_{j>i}G_{ij}v_{i}^{\prime}Av_{j}\\
 & =\sum_{i}\sum_{j<i}\left(G_{ij}v_{i}^{\prime}Av_{j}+G_{ji}v_{j}^{\prime}Av_{i}\right).
\end{align*}

Hence, $\sum_{i}s_{i}^{\prime}v_{i}+\sum_{i}\sum_{j\ne i}G_{ij}v_{i}^{\prime}Av_{j} =\sum_{i}y_{i}$, where
\begin{align*}
y_{i} &=s_{i}^{\prime}v_{i}+\sum_{j<i}\left(G_{ij}v_{i}^{\prime}Av_{j}+G_{ji}v_{j}^{\prime}Av_{i}\right)
=s_{i}^{\prime}v_{i}+v_{i}^{\prime}A\left(\sum_{j<i}G_{ij}v_{j}\right)+\left(\sum_{j<i}G_{ji}v_{j}^{\prime}\right)Av_{i}\\
 & =s_{i}^{\prime}v_{i}+v_{i}^{\prime}A\left(G_{L}v\right)_{i\cdot}^{\prime}+\left(G_{U}^{\prime}v\right)_{i\cdot}Av_{i}.
\end{align*}

Let $\mathcal{F}_{i}$ denote the filtration of $y_{1},\dots,y_{i-1}$.
To apply the martingale CLT, we require:
\begin{enumerate}
\item $\sum_{i}E\left[|y_{i}|^{2+\epsilon}\right]\rightarrow0$.
\item Conditional variance converges to 1, i.e., $P\left(|\sum_{i}E\left[B^{2}y_{i}^{2}\mid\mathcal{F}_{i}\right]-1|>\eta\right)\rightarrow0$,
where $B=\Var\left(T\right)^{-1/2}$.
\end{enumerate}
The 4th moments of $v_{i}$ are bounded.
With $\epsilon=2$, we want $\sum_{i}E\left[y_{i}^{4}\right]\rightarrow0$.
Using Loeve's $c_{r}$ inequality, it suffices that, for any element $l$ of the $v_{i}$ vector,
\begin{align*}
\sum_{i}s_{il}^{4}E\left[v_{il}^{4}\right] & \rightarrow0, \text{ and } \quad 
\sum_{i}E\left[v_{il}^{4}\left(G_{L}v\right)_{il}^{4}\right] \rightarrow0.
\end{align*}

The first condition is immediate from condition (2). 
The second condition holds by condition (3) using the proof in EK18. 
To be precise,
\begin{align*}
\sum_{i}E\left[v_{il}^{4}\left(G_{L}v\right)_{il}^{4}\right] & =\sum_{i}E\left[v_{il}^{4}\right]E\left[\left(G_{L}v\right)_{il}^{4}\right]\preceq\sum_{i}E\left[\left(G_{L}v\right)_{il}^{4}\right]\\
 & =\sum_{i}\sum_{j}G_{L,ij}^{4}E\left[v_{il}^{4}\right]+3\sum_{i}\sum_{j}\sum_{k\ne j}G_{L,ij}^{2}G_{L,ik}^{2}E\left[v_{il}^{2}\right]E\left[v_{jl}^{2}\right]\\
 & \preceq\sum_{i}\sum_{j}\sum_{k}G_{L,ij}^{2}G_{L,ik}^{2}
 =\sum_{i}\left(G_{L}G_{L}^{\prime}\right)_{ii}^{2} \\
 &\leq\sum_{i}\sum_{j}\left(G_{L}G_{L}^{\prime}\right)_{ij}^{2}=||G_{L}G_{L}^{\prime}||_{F}^{2}.
\end{align*}

The argument for $G_{U}$ is analogous. 
Now, I turn to showing convergence of the conditional variance. 
With abuse of notation, let $W_{i}=s_{i}^{\prime}v_{i}$ and $X_{i}=v_{i}^{\prime}A\left(G_{L}v\right)_{i\cdot}^{\prime}+v_{i}^{\prime}A\left(G_{U}^{\prime}v\right)_{i\cdot}^{\prime}$.
Since $\Var\left(BT\right)=B^{2}\sum_{i}E\left[W_{i}^{2}\right]+B^{2}\sum_{i}E\left[X_{i}^{2}\right]=1$,
\begin{align*}
\sum_{i}E\left[B^{2}y_{i}^{2}\mid\mathcal{F}_{i}\right]-1&=B^{2}\sum_{i}\left(E\left[X_{i}^{2}\mid\mathcal{F}_{i}\right]-E\left[X_{i}^{2}\right]\right) +2B^{2}\sum_{i}E\left[W_{i}X_{i}\mid\mathcal{F}_{i}\right] \\
&\quad +B^{2}\sum_{i}\left(E\left[W_{i}^{2}\mid\mathcal{F}_{i}\right]-E\left[W_{i}^{2}\right]\right).
\end{align*}

The previous observations in the filtration do not feature, so $E\left[W_{i}^{2}\mid\mathcal{F}_{i}\right]-E\left[W_{i}^{2}\right]=0$.
It suffices to show that the RHS converges to 0.
For the $\sum_{i}E\left[W_{i}X_{i}\mid\mathcal{F}_{i}\right]$ term,
\begin{align*}
B^{2}\sum_{i}E\left[W_{i}X_{i}\mid\mathcal{F}_{i}\right] & =B^{2}\sum_{i}E\left[W_{i}\left(v_{i}^{\prime}A\left(G_{L}v\right)_{i\cdot}^{\prime}+v_{i}^{\prime}A\left(G_{U}^{\prime}v\right)_{i\cdot}^{\prime}\right)\mid\mathcal{F}_{i}\right]\\
 & =B^{2}\sum_{i}E\left[W_{i}v_{i}^{\prime}A\right]\left(G_{L}v\right)_{i\cdot}^{\prime}+B^{2}\sum_{i}E\left[W_{i}v_{i}^{\prime}A\right]\left(G_{U}^{\prime}v\right)_{i\cdot}^{\prime}.
\end{align*}

It suffices to show that the respective squares converge to 0. 
Due to bounded fourth moments, and applying the Cauchy-Schwarz inequality repeatedly, for some n-vector $\delta_v$ with $||\delta_v||_2 \leq C$,
\[
E\left[\left(\sum_{i} E[W_i v_i^\prime] A \left(G_{L}v\right)_{i\cdot}^{\prime}\right)^{2}\right]\preceq
\delta_v^\prime G_L G_L^\prime \delta_v \leq ||\delta_v||_2^2 ||G_L G_L^\prime||_2 \preceq ||G_{L}G_{L}^{\prime}||_{F},
\]
and the same argument can be applied to the $G_{U}$ term. 
Finally, 
\begin{align*}
\sum_{i}\left(E\left[X_{i}^{2}\mid\mathcal{F}_{i}\right]-E\left[X_{i}^{2}\right]\right) & =\sum_{i}\bigg(E\left[\left(v_{i}^{\prime}A\left(G_{L}v\right)_{i\cdot}^{\prime}+v_{i}^{\prime}A\left(G_{U}^{\prime}v\right)_{i\cdot}^{\prime}\right)^{2}\mid\mathcal{F}_{i}\right] \\
&\quad-E\left[\left(v_{i}^{\prime}A\left(G_{L}v\right)_{i\cdot}^{\prime}+v_{i}^{\prime}A\left(G_{U}^{\prime}v\right)_{i\cdot}^{\prime}\right)^{2}\right]\bigg).
\end{align*}

It suffices to consider the $G_{L}$ term, as the $G_{U}$ and cross terms are analogous: 
\begin{align*}
\sum_{i} & \left(E\left[\left(v_{i}^{\prime}A\left(G_{L}v\right)_{i\cdot}^{\prime}\right)^{2}\mid\mathcal{F}_{i}\right]-E\left[\left(v_{i}^{\prime}A\left(G_{L}v\right)_{i\cdot}^{\prime}\right)^{2}\right]\right)\\
 & =\sum_{i}\left(\left(G_{L}v\right)_{i\cdot}A^{\prime}E\left[v_{i}v_{i}^{\prime}\right]A\left(G_{L}v\right)_{i\cdot}^{\prime}-E\left[\left(G_{L}v\right)_{i\cdot}A^{\prime}v_{i}v_{i}^{\prime}A\left(G_{L}v\right)_{i\cdot}^{\prime}\right]\right).
\end{align*}

Since $\sum_{i}\left(G_{L}v\right)_{i\cdot}A^{\prime}E\left[v_{i}v_{i}^{\prime}\right]A\left(G_{L}v\right)_{i\cdot}^{\prime}$ is demeaned, it suffices to show that its variance converges to 0.
Due to bounded moments, 
\[
\Var\left(\sum_{i}\left(G_{L}v\right)_{i\cdot}A^{\prime}E\left[v_{i}v_{i}^{\prime}\right]A\left(G_{L}v\right)_{i\cdot}^{\prime}\right)\preceq \sum_{i}\sum_{j} \left(G_L G_L^\prime \right)^{2}=||G_{L}G_{L}^{\prime}||_{F}^{2},
\]
which suffices for the result. 
\end{proof}

The full and latest version of the paper, including proofs of results in the appendices, can be found at \url{https://lutheryap.github.io/files/mwiv_het_wp.pdf}.

\clearpage 
\section{Supplementary Appendix, Not for Publication}
\small
\singlespacing

\subsection{Proofs for Appendix C}

\begin{proof} [Proof of \Cref{lem:a0}]
I begin with part (c). By applying the Cauchy-Schwarz inequality,
\begin{align*}
 & \left|\sum_{i}C_{i}\left(\sum_{j\ne i}h_{2}^{A}\left(i,j\right)R_{mj}\right)\left(\sum_{j\ne i}h_{2}^{B}\left(i,j\right)R_{mj}\right)\right|\\
 & \leq\left(\sum_{i}C_{i}\left(\sum_{j\ne i}h_{2}^{A}\left(i,j\right)R_{mj}\right)^{2}\right)^{1/2}\left(\sum_{i}C_{i}\left(\sum_{j\ne i}h_{2}^{B}\left(i,j\right)R_{mj}\right)^{2}\right)^{1/2}\\
 & \leq\max_{i}C_{i}\left(\sum_{i}\left(\sum_{j\ne i}h_{2}^{A}\left(i,j\right)R_{mj}\right)^{2}\right)^{1/2}\left(\sum_{i}\left(\sum_{j\ne i}h_{2}^{B}\left(i,j\right)R_{mj}\right)^{2}\right)^{1/2}\\
 & \leq\max_{i}C_{i}\left(\sum_{i}\tilde{R}_{mi}^{2}\right)^{1/2}\left(\sum_{i}\tilde{R}_{mi}^{2}\right)^{1/2}\leq C\sum_{i}\tilde{R}_{mi}^{2}.
\end{align*}

The proof of all other parts are entirely analogous. 
\end{proof}

\begin{proof} [Proof of \Cref{lem:a1}]

\textbf{Proof of \Cref{lem:a1}(a)}.

Using the decomposition from AS23,
\begin{align*}
\Var &  \left(\sum_{i}\sum_{j\ne i}G_{ij}F_{ij}V_{1i}V_{2i}V_{3j}V_{4j}\right)\\
 & =\sum_{i\ne j}^{n}G_{ij}^{2}F_{ij}^{2}\Var\left(V_{1i}V_{2i}V_{3j}V_{4j}\right)+\sum_{i\ne j}^{n}G_{ij}F_{ij}G_{ji}F_{ji}Cov\left(V_{1i}V_{2i}V_{3j}V_{4j},V_{1j}V_{2j}V_{3i}V_{4i}\right)\\
 & \quad+\sum_{i\ne j\ne k}^{n}G_{ij}F_{ij}G_{kj}F_{kj}Cov\left(V_{1i}V_{2i}V_{3j}V_{4j},V_{1k}V_{2k}V_{3j}V_{4j}\right)
 \\
 &\quad +\sum_{i\ne j\ne k}^{n}G_{ij}F_{ij}G_{jk}F_{jk}Cov\left(V_{1i}V_{2i}V_{3j}V_{4j},V_{1j}V_{2j}V_{3k}V_{4k}\right)\\
 & \quad+\sum_{i\ne j\ne k}^{n}G_{ij}F_{ij}G_{ik}F_{ik}Cov\left(V_{1i}V_{2i}V_{3j}V_{4j},V_{1i}V_{2i}V_{3k}V_{4k}\right)
  \\
  &\quad +\sum_{i\ne j\ne k}^{n}G_{ij}F_{ij}G_{ki}F_{ki}Cov\left(V_{1i}V_{2i}V_{3j}V_{4j},V_{1k}V_{2k}V_{3i}V_{4i}\right)\\
 & \leq2\left[\max_{i,j}\Var\left(V_{1i}V_{2i}V_{3j}V_{4j}\right)\right]\sum_{i}\left(\left(\sum_{j\ne i}G_{ij}F_{ij}\right)^{2}+\left(\sum_{j\ne i}G_{ij}F_{ij}\right)\left(\sum_{j\ne i}G_{ji}F_{ji}\right)\right).
\end{align*}

Notice that the terms in $\sum_{i\ne j}^{n}$ are absorbed into the sum over $k$ so that the final expression can be written as $\sum_{i}\sum_{j\ne i}\sum_{k\ne i}$.
Then, due to \Cref{asmp:Vconst_reg}(a) and the Cauchy-Schwarz inequality,
\begin{align*}
\sum_{i}\left(\sum_{j\ne i}G_{ij}F_{ij}\right)^{2} & \leq\sum_{i}\left(\sum_{j\ne i}G_{ij}^{2}\right)\left(\sum_{j\ne i}F_{ij}^{2}\right)\leq C\sum_{i}\sum_{j\ne i}G_{ij}^{2},
\end{align*}

and
\begin{align*}
\left|\sum_{i}\left(\sum_{j\ne i}G_{ij}F_{ij}\right)\left(\sum_{j\ne i}G_{ji}F_{ji}\right)\right| & \leq\left(\sum_{i}\left(\sum_{j\ne i}G_{ij}F_{ij}\right)^{2}\right)^{1/2}\left(\sum_{i}\left(\sum_{j\ne i}G_{ji}F_{ji}\right)^{2}\right)^{1/2}\\
 & \leq C\left(\sum_{i}\sum_{j\ne i}G_{ij}^{2}\right)^{1/2}\left(\sum_{i}\sum_{j\ne i}G_{ji}^{2}\right)^{1/2}=C\sum_{i}\sum_{j\ne i}G_{ij}^{2}.
\end{align*}

\textbf{Proof of \Cref{lem:a1}(b).} 
Expand the term:
\begin{align*}
\sum_{i\ne j\ne k}^{n}G_{ij}F_{ij}\check{M}_{ik,-ij}V_{1i}V_{2k}V_{3j}V_{4j}
= \sum_{i\ne j\ne k}^{n}G_{ij}F_{ij}\check{M}_{ik,-ij}\left(R_{1i}R_{2k}+v_{1i}R_{2k}+R_{1i}v_{2k}+v_{1i}v_{2k}\right)V_{3j}V_{4j}.
\end{align*}
Consider the final sum with 4 stochastic terms. 
The 6-sums have zero covariances due to independent sampling. 
The 5-sums also have zero covariances, because at least one of $v_{1}$ or $v_{2}$ needs to have different indices. 
Within the 4-sum, the covariance is nonzero only for $j_{2}\ne j$. We require $i_{2}$ to be equal to either $i$ or $k$ and $k_{2}$ the other index. 
Hence, by bounding covariances above by Cauchy-Schwarz,
\begin{align*}
\Var &  \left(\sum_{i\ne j\ne k}^{n}G_{ij}F_{ij}\check{M}_{ik,-ij}v_{1i}v_{2k}V_{3j}V_{4j}\right)\\
\leq & \max_{i,j,k}\Var\left(v_{1i}v_{2k}V_{3j}V_{4j}\right)\sum_{i}\sum_{j\ne i}\sum_{k\ne i,j}\sum_{l\ne i,j,k}\left(G_{ij}F_{ij}G_{il}F_{il}\check{M}_{ik,-ij}\check{M}_{ik,-il}+G_{ij}F_{ij}G_{kl}F_{kl}\check{M}_{ik,-ij}\check{M}_{ki,-kl}\right)\\
 & +\max_{i,j,k}\Var\left(v_{1i}v_{2k}V_{3j}V_{4j}\right)3!\sum_{i\ne j\ne k}^{n}G_{ij}^{2}F_{ij}^{2}\check{M}_{ik,-ij}^{2}\\
\leq & \max_{i,j,k}\Var\left(v_{1i}v_{2k}V_{3j}V_{4j}\right)\left(\sum_{i\ne j\ne k\ne l}^{n}G_{ij}^{2}G_{il}^{2}\check{M}_{ik,-ij}^{2}\right)^{1/2}\left(\sum_{i\ne j\ne k\ne l}^{n}F_{ij}^{2}F_{il}^{2}\check{M}_{ik,-ij}^{2}\right)^{1/2}\\
 & +\max_{i,j,k}\Var\left(v_{1i}v_{2k}V_{3j}V_{4j}\right)\left(\sum_{i\ne j\ne k\ne l}^{n}G_{ij}^{2}G_{kl}^{2}\check{M}_{ik,-ij}^{2}\right)^{1/2}\left(\sum_{i\ne j\ne k\ne l}^{n}F_{ij}^{2}F_{kl}^{2}\check{M}_{ik,-ij}^{2}\right)^{1/2}\\
 & +\max_{i,j,k}\Var\left(v_{1i}v_{2k}V_{3j}V_{4j}\right)3!\sum_{i\ne j\ne k}^{n}G_{ij}^{2}F_{ij}^{2}\check{M}_{ik,-ij}^{2}.
\end{align*}

To obtain the first inequality, observe that once we have fixed 3
indices, there are $3!$ permutations of the $v_{1i}v_{2k}V_{3j}V_{4j}$
that we can calculate covariances for. They are all bounded above
by the variance. In the various combinations, we may have different
combinations of $G$ and $F$, but they are bounded above by the expression.
To be precise, the 3-sum is:
\begin{align*}
 & \sum_{i\ne j\ne k}^{n}G_{ij}F_{ij}\check{M}_{ik,-ij}\left(G_{ij}F_{ij}\check{M}_{ik,-ij}+G_{ik}F_{ik}\check{M}_{ij,-ik}+G_{ji}F_{ji}\check{M}_{jk,-ji}\right)\\
 & +\sum_{i\ne j\ne k}^{n}G_{ij}F_{ij}\check{M}_{ik,-ij}\left(G_{jk}F_{jk}\check{M}_{ji,-jk}+G_{ki}F_{ki}\check{M}_{kj,-ki}+G_{kj}F_{kj}\check{M}_{ki,-kj}\right).
\end{align*}
Apply Cauchy-Schwarz to the sum and apply the commutative property
of summations to obtain the upper bound. For instance,
\begin{align*}
 \left(\sum_{i\ne j\ne k}^{n}G_{ij}F_{ij}\check{M}_{ik,-ij}G_{jk}F_{jk}\check{M}_{ji,-jk}\right)^{2}
\leq \left(\sum_{i\ne j\ne k}^{n}G_{ij}^{2}F_{ij}^{2}\check{M}_{ik,-ij}^{2}\right)\left(\sum_{i\ne j\ne k}^{n}G_{jk}^{2}F_{jk}^{2}\check{M}_{ji,-jk}^{2}\right).
\end{align*}
Then, observe that $\sum_{i}\sum_{j\ne i}\sum_{k\ne i,j}G_{jk}^{2}F_{jk}^{2}\check{M}_{ji,-jk}^{2}=\sum_{j}\sum_{k\ne j}\sum_{i\ne j,k}G_{jk}^{2}F_{jk}^{2}\check{M}_{ji,-jk}^{2} \\=\sum_{i}\sum_{j\ne i}\sum_{k\ne i,j}G_{ij}^{2}F_{ij}^{2}\check{M}_{ik,-ij}^{2}$.
Due to AS23 Equation (22), $\sum_{l}\check{M}_{il-ijk}^{2}=O(1)$,
so $\sum_{i\ne j\ne k}^{n}G_{ij}^{2}F_{ij}^{2}\check{M}_{ik,-ij}^{2}\leq C\sum_{i}\sum_{j\ne i}G_{ij}^{2}F_{ij}^{2}\leq C\sum_{i}\sum_{j\ne i}G_{ij}^{2}$.
Similarly, $\sum_{i\ne j\ne k\ne l}^{n}G_{ij}^{2}G_{kl}^{2}\check{M}_{ik,-ij}^{2}=O(1)\sum_{i\ne j\ne k}^{n}G_{ij}^{2}\check{M}_{ik,-ij}^{2}=O(1)\sum_{i\ne j}^{n}G_{ij}^{2}$,
which delivers the order required. 

To deal with 3 stochastic terms,
\begin{align*}
\Var &  \left(\sum_{i\ne j\ne k}^{n}G_{ij}F_{ij}\check{M}_{ik,-ij}R_{1i}v_{2k}V_{3j}V_{4j}\right)
=\Var\left(\sum_{i\ne j}^{n}v_{2i}V_{3j}V_{4j}\left(\sum_{k\ne i,j}G_{kj}F_{kj}\check{M}_{ki,-kj}R_{1k}\right)\right)\\
 & \leq\sum_{i\ne j}^{n}\Var\left(v_{2i}V_{3j}V_{4j}\right)\left(\sum_{k\ne i,j}G_{kj}F_{kj}\check{M}_{ki,-kj}R_{1k}\right)\left[\sum_{k\ne i,j}G_{kj}F_{kj}\check{M}_{ki,-kj}R_{1k}+\sum_{k\ne i,j}G_{ki}F_{ki}\check{M}_{kj,-ki}R_{1k}\right]\\
 & +\max_{i,j}\Var\left(v_{2i}V_{3j}V_{4j}\right)\sum_{i\ne j}^{n}\sum_{l\ne i,j}\left(\sum_{k\ne i,j}G_{kj}F_{kj}\check{M}_{ki,-kj}R_{1k}\right)\left(\sum_{k\ne i,l}G_{kl}F_{kl}\check{M}_{ki,-kl}R_{1k}\right)\\
 & \leq\sum_{i}\sum_{j\ne i}\Var\left(v_{2i}V_{3j}V_{4j}\right)\left(\sum_{k\ne i,j}G_{kj}F_{kj}\check{M}_{ki,-kj}R_{1k}\right) \\
 &\quad \left[\sum_{k\ne i,j}G_{kj}F_{kj}\check{M}_{ki,-kj}R_{1k}+\sum_{k\ne i,j}G_{ki}F_{ki}\check{M}_{kj,-ki}R_{1k}\right]\\
 & +\max_{i,j}\Var\left(v_{2i}V_{3j}V_{4j}\right)\sum_{i}\left(\sum_{j\ne i}\sum_{k\ne i,j}G_{kj}F_{kj}\check{M}_{ki,-kj}R_{1k}\right)^{2} \\
 &\quad -\max_{i,j}\Var\left(v_{2i}V_{3j}V_{4j}\right)\sum_{i}\sum_{j\ne i}\left(\sum_{k\ne i,j}G_{kj}F_{kj}\check{M}_{ki,-kj}R_{1k}\right)^{2}\\
 & \leq\max_{i,j}\Var\left(v_{2i}V_{3j}V_{4j}\right)\sum_{i}\left(\sum_{j\ne i}\sum_{k\ne i,j}G_{kj}F_{kj}\check{M}_{ki,-kj}R_{1k}\right)^{2} \\
 &+\sum_{i\ne j}^{n}\Var\left(v_{2i}V_{3j}V_{4j}\right)\left(\sum_{k\ne i,j}G_{kj}F_{kj}\check{M}_{ki,-kj}R_{1k}\right)\left(\sum_{k\ne i,j}G_{ki}F_{ki}\check{M}_{kj,-ki}R_{1k}\right)
\end{align*}

To get the first inequality, observe that, if for $l\ne i,j$, we have $v_{2l}$ instead of $V_{3l}V_{4l}$, the covariance must be 0. 
We can then bound the order by using \Cref{asmp:Vconst_reg} and \Cref{lem:a0}.
Similarly,
\begin{align*}
\Var &  \left(\sum_{i\ne j\ne k}^{n}G_{ij}F_{ij}\check{M}_{ik,-ij}v_{1i}R_{2k}V_{3j}V_{4j}\right)=\Var\left(\sum_{i\ne j}^{n}v_{1i}V_{3j}V_{4j}\left(\sum_{k\ne i,j}G_{ij}F_{ij}\check{M}_{ik,-ij}R_{2k}\right)\right)\\
 & \leq\max_{i,j}\Var\left(v_{1i}V_{3j}V_{4j}\right)\sum_{i}\left(\sum_{j\ne i}\sum_{k\ne i,j}G_{ij}F_{ij}\check{M}_{ik,-ij}R_{2k}\right)^{2} \\
 &+\sum_{i}\sum_{j\ne i}\Var\left(v_{1i}V_{3j}V_{4j}\right)\left(\sum_{k\ne i,j}G_{ij}F_{ij}\check{M}_{ik,-ij}R_{2k}\right)\left(\sum_{k\ne i,j}G_{ji}F_{ji}\check{M}_{jk,-ij}R_{2k}\right).
\end{align*}
since the expansion in the intermediate steps are entirely analogous. 

Turning to the sum with two stochastic objects,
\begin{align*}
\Var &  \left(\sum_{i\ne j\ne k}^{n}G_{ij}F_{ij}\check{M}_{ik,-ij}R_{1i}R_{2k}V_{3j}V_{4j}\right)
=\Var\left(\sum_{i}V_{3i}V_{4i}\left(\sum_{j\ne i}\sum_{k\ne i,j}G_{ji}F_{ji}\check{M}_{jk,-ij}R_{1j}R_{2k}\right)\right)\\
 & =\sum_{i}\Var\left(V_{3i}V_{4i}\right)\left(\sum_{j\ne i}\sum_{k\ne i,j}G_{ji}F_{ji}\check{M}_{jk,-ij}R_{1j}R_{2k}\right)^{2} \\
 &\leq\max_{i}\Var\left(V_{3i}V_{4i}\right)\sum_{i}\left(\sum_{j\ne i}\sum_{k\ne i,j}G_{ji}F_{ji}\check{M}_{jk,-ij}R_{1j}R_{2k}\right)^{2}.
\end{align*}
With these inequalities, applying \Cref{asmp:Vconst_reg} suffices for the result.

\textbf{Proof of \Cref{lem:a1}(c)}. Expand the term:
\begin{align*}
\sum_{i\ne j\ne l}^{n}G_{ij}F_{ij}\check{M}_{jl,-ij}V_{1i}V_{2i}V_{3j}V_{4l}
= \sum_{i\ne j\ne l}^{n}G_{ij}F_{ij}\check{M}_{jl,-ij}V_{1i}V_{2i}\left(R_{3j}R_{4l}+R_{3j}v_{4l}+v_{3j}R_{4l}+v_{3j}v_{4l}\right).
\end{align*}

With four stochastic objects, 
\begin{align*}
\Var &  \left(\sum_{i\ne j\ne l}^{n}G_{ij}F_{ij}\check{M}_{jl,-ij}V_{1i}V_{2i}v_{3j}v_{4l}\right)\\
\leq & \max_{i,j,k}\Var\left(V_{1i}V_{2i}v_{3j}v_{4l}\right)\sum_{i\ne j\ne l}^{n}\sum_{i_{2}\ne i,j,l}\left(G_{ij}F_{ij}\check{M}_{jl,-ij}G_{i_{2}j}F_{i_{2}j}\check{M}_{jl,-i_{2}j}+G_{ij}F_{ij}\check{M}_{jl,-ij}G_{i_{2}l}F_{i_{2}l}\check{M}_{lj,-i_{2}l}\right)\\
 & +\max_{i,j,k}\Var\left(V_{1i}V_{2i}v_{3j}v_{4l}\right)3!\sum_{i\ne j\ne l}^{n}G_{ij}^{2}F_{ij}^{2}\check{M}_{jl,-ij}^{2}.
\end{align*}

Simplifying the first line,
\begin{align*}
 & \sum_{i\ne j\ne l}^{n}\sum_{i_{2}\ne i,j,l}\left(G_{ij}F_{ij}\check{M}_{jl,-ij}G_{i_{2}j}F_{i_{2}j}\check{M}_{jl,-i_{2}j}+G_{ij}F_{ij}\check{M}_{jl,-ij}G_{i_{2}l}F_{i_{2}l}\check{M}_{lj,-i_{2}l}\right)\\
 & \leq\left(\sum_{i\ne j\ne l\ne i_{2}}^{n}G_{ij}^{2}G_{i_{2}j}^{2}\check{M}_{jl,-ij}^{2}\right)^{1/2}\left(\sum_{i\ne j\ne l\ne i_{2}}^{n}F_{ij}^{2}F_{i_{2}j}^{2}\check{M}_{jl,-i_{2}j}^{2}\right)^{1/2} \\
 &+\left(\sum_{i\ne j\ne l\ne i_{2}}^{n}G_{ij}^{2}G_{i_{2}l}^{2}\check{M}_{jl,-ij}^{2}\right)^{1/2}\left(\sum_{i\ne j\ne l\ne i_{2}}^{n}F_{ij}^{2}F_{i_{2}l}^{2}\check{M}_{lj,-i_{2}j}^{2}\right)^{1/2}.
\end{align*}

These terms have the required order due to a proof analogous to \Cref{lem:a1}(b). Next,
\begin{align*}
\Var & \left(\sum_{i\ne j\ne l}^{n}G_{ij}F_{ij}\check{M}_{jl,-ij}V_{1i}V_{2i}R_{3j}v_{4l}\right)
=\Var\left(\sum_{i}\sum_{j\ne i}V_{1i}V_{2i}v_{4j}\left(\sum_{l\ne i,j}G_{il}F_{il}\check{M}_{lj,-il}R_{3l}\right)\right)\\
 & \leq\sum_{i}\sum_{j\ne i}\Var\left(V_{1i}V_{2i}v_{4j}\right)\left(\sum_{l\ne i,j}G_{il}F_{il}\check{M}_{lj,-il}R_{3l}\right)\left[\sum_{l\ne i,j}G_{il}F_{il}\check{M}_{lj,-il}R_{3l}+\sum_{l\ne i,j}G_{jl}F_{jl}\check{M}_{li,-jl}R_{3l}\right]\\
 & +\max_{i,j}\Var\left(V_{1i}V_{2i}v_{4j}\right)\sum_{i}\sum_{j\ne i}\sum_{i_{2}\ne i,j}\left(\sum_{l\ne i,j}G_{il}F_{il}\check{M}_{lj,-il}R_{3l}\right)\left(\sum_{k\ne i_{2},l}G_{kl}F_{kl}\check{M}_{ki_{2},-kl}R_{1k}\right)\\
 & \leq\max_{i,j}\Var\left(V_{1i}V_{2i}v_{4j}\right)\sum_{i}\left(\sum_{j\ne i}\sum_{l\ne i,j}G_{il}F_{il}\check{M}_{lj,-il}R_{3l}\right)^{2} \\
 &+\sum_{i}\sum_{j\ne i}\Var\left(V_{1i}V_{2i}v_{4j}\right)\left(\sum_{l\ne i,j}G_{il}F_{il}\check{M}_{lj,-il}R_{3l}\right)\left(\sum_{l\ne i,j}G_{jl}F_{jl}\check{M}_{li,-jl}R_{3l}\right).
\end{align*}

Further, $\Var\left(\sum_{i\ne j\ne l}^{n}G_{ij}F_{ij}\check{M}_{jl,-ij}V_{1i}V_{2i}v_{3j}R_{4l}\right)$ can be bounded by a similar argument. 
Turning to the sum with two stochastic objects,
\begin{align*}
\Var &  \left(\sum_{i\ne j\ne l}^{n}G_{ij}F_{ij}\check{M}_{jl,-ij}V_{1i}V_{2i}R_{3j}R_{4l}\right)
 =\sum_{i}\Var\left(V_{1i}V_{2i}\right)\left(\sum_{j\ne i}\sum_{l\ne i,j}G_{ij}F_{ij}\check{M}_{jl,-ij}R_{3j}R_{4l}\right)^{2}.
\end{align*}
These inequalities suffice for the result due to \Cref{asmp:Vconst_reg}. 

\textbf{Proof of \Cref{lem:a1}(d).}
Expand the term: 
\begin{align*}
 & \sum_{i\ne j\ne k\ne l}^{n}G_{ij}F_{ij}\check{M}_{ik,-ij}\check{M}_{jl,-ijk}V_{1i}V_{2k}V_{3j}V_{4l}\\
 & =\sum_{i\ne j\ne k\ne l}^{n}G_{ij}F_{ij}\check{M}_{ik,-ij}\check{M}_{jl,-ijk}V_{1i}R_{2k}\left(R_{3j}R_{4l}+R_{3j}v_{4l}+v_{3j}R_{4l}+v_{3j}v_{4l}\right)\\
 & \quad+\sum_{i\ne j\ne k\ne l}^{n}G_{ij}F_{ij}\check{M}_{ik,-ij}\check{M}_{jl,-ijk}V_{1i}v_{2k}\left(R_{3j}R_{4l}+R_{3j}v_{4l}+v_{3j}R_{4l}+v_{3j}v_{4l}\right).
\end{align*}

Consider the $v_{2k}$ line first. 
We only have the 4-sum to contend with. 
For 5-sum and above, at least one of the errors can be factored out as a zero expectation. 
Hence, by using Cauchy-Schwarz and the same argument as above,
\begin{align*}
 & \Var\left(\sum_{i\ne j\ne k\ne l}^{n}G_{ij}F_{ij}\check{M}_{ik,-ij}\check{M}_{jl,-ijk}V_{1i}v_{2k}v_{3j}v_{4l}\right)\\
 & \leq\max_{i,j,k,l}\Var\left(V_{1i}v_{2k}v_{3j}v_{4l}\right)4!\sum_{i\ne j\ne k\ne l}^{n}G_{ij}^{2}F_{ij}^{2}\check{M}_{ik,-ij}^{2}\check{M}_{jl,-ijk}^{2}\\
 & \leq C\sum_{i\ne j\ne k}^{n}G_{ij}^{2}F_{ij}^{2}\check{M}_{ik,-ij}^{2}\leq C\sum_{i\ne j}^{n}G_{ij}^{2}F_{ij}^{2}\leq C\left(\sum_{i\ne j}^{n}G_{ij}^{2}\right)^{1/2}\left(\sum_{i\ne j}^{n}F_{ij}^{2}\right)^{1/2}.
\end{align*}

By using the same expansion step as before,
\begin{align*}
\Var &  \left(\sum_{i\ne j\ne k}^{n}G_{ij}F_{ij}\check{M}_{ik,-ij}V_{1i}v_{2k}v_{3j}\left(\sum_{l\ne i,j,k}\check{M}_{jl,-ijk}R_{4l}\right)\right)\\
\leq & \max_{i,j,k}\Var\left(V_{1i}v_{2k}v_{3j}\left(\sum_{l\ne i,j,k}\check{M}_{jl,-ijk}R_{4l}\right)\right) \\
&\quad \sum_{i \ne j \ne k \ne i_2}^n \left(G_{ij}F_{ij}G_{i_{2}j}F_{i_{2}j}\check{M}_{ik,-ij}\check{M}_{i_{2}k,-ij}+G_{ij}F_{ij}G_{i_{2}k}F_{i_{2}k}\check{M}_{ij,-ik}\check{M}_{i_{2}j,-ik}\right)\\
 & +\max_{i,j,k}\Var\left(V_{1i}v_{2k}v_{3j}\left(\sum_{l\ne i,j,k}\check{M}_{jl,-ijk}R_{4l}\right)\right)3!\sum_{i}\sum_{j\ne i}\sum_{k\ne i,j}G_{ij}^{2}F_{ij}^{2}\check{M}_{ik,-ij}^{2}.
\end{align*}

The $\sum_{i\ne j\ne k\ne i_{2}}^{n}\left(G_{ij}F_{ij}G_{i_{2}j}F_{i_{2}j}\check{M}_{ik,-ij}\check{M}_{i_{2}k,-ij}+G_{ij}F_{ij}G_{i_{2}k}F_{i_{2}k}\check{M}_{ij,-ik}\check{M}_{i_{2}j,-ik}\right)$
term has the required order due to the same argument as the proof
of \Cref{lem:a1}(b). Next, 
\footnotesize
\begin{align*}
\Var & \left(\sum_{i\ne j\ne k\ne l}^{n}G_{ij}F_{ij}\check{M}_{ik,-ij}\check{M}_{jl,-ijk}V_{1i}v_{2k}R_{3j}v_{4l}\right)=\Var\left(\sum_{i\ne j\ne k}^{n}V_{1i}v_{2k}v_{4j}\left(\sum_{l\ne i,j,k}G_{il}F_{il}\check{M}_{ik,-il}\check{M}_{lj,-ilk}R_{3l}\right)\right)\\
\leq & \max_{i,j,k}\Var\left(V_{1i}v_{2k}v_{4j}\right)\sum_{i\ne j\ne k}^{n}\sum_{i_{2}\ne i,j,k}\left(\sum_{l\ne i,j,k}G_{il}F_{il}\check{M}_{ik,-il}\check{M}_{lj,-ilk}R_{3l}\right)\left(\sum_{l\ne i_{2},j,k}G_{i_{2}l}F_{i_{2}l}\check{M}_{i_{2}k,-i_{2}l}\check{M}_{lj,-i_{2}lk}R_{3l}\right)\\
 & +\max_{i,j,k}\Var\left(V_{1i}v_{2k}v_{4j}\right)\sum_{i\ne j\ne k}^{n}\sum_{i_{2}\ne i,j,k}\left(\sum_{l\ne i,j,k}G_{il}F_{il}\check{M}_{ik,-il}\check{M}_{lj,-ilk}R_{3l}\right)\left(\sum_{l\ne i_{2},j,k}G_{i_{2}l}F_{i_{2}l}\check{M}_{i_{2}j,-i_{2}l}\check{M}_{lk,-i_{2}lj}R_{3l}\right)\\
 & +\max_{i,j,k}\Var\left(V_{1i}v_{2k}v_{4j}\right)3!\sum_{i\ne j\ne k}^{n}\check{M}_{ik,-ij}^{2}\left(\sum_{l\ne i,j,k}G_{il}F_{il}\check{M}_{ik,-il}\check{M}_{lj,-ilk}R_{3l}\right)^{2}\\
\leq & \max_{i,j,k}\Var\left(V_{1i}v_{2k}v_{4j}\right) \sum_{k}\sum_{j\ne k}\left(\sum_{i\ne k,j}\sum_{l\ne i,j,k}G_{il}F_{il}\check{M}_{ik,-il}\check{M}_{lj,-ilk}R_{3l}\right)^{2} \\
&\quad -\max_{i,j,k}\Var\left(V_{1i}v_{2k}v_{4j}\right) \sum_{k}\sum_{j\ne k}\sum_{i\ne k,j}\left(\sum_{l\ne i,j,k}G_{il}F_{il}\check{M}_{ik,-il}\check{M}_{lj,-ilk}R_{3l}\right)^{2} \\
 & +\max_{i,j,k}\Var\left(V_{1i}v_{2k}v_{4j}\right)\sum_{k}\sum_{j\ne k}\left(\sum_{i\ne k,j}\sum_{l\ne i,j,k}G_{il}F_{il}\check{M}_{ik,-il}\check{M}_{lj,-ilk}R_{3l}\right)\left(\sum_{i\ne k,j}\sum_{l\ne i,j,k}G_{il}F_{il}\check{M}_{ij,-il}\check{M}_{lk,-ilj}R_{3l}\right)\\
 & -\max_{i,j,k}\Var\left(V_{1i}v_{2k}v_{4j}\right)\sum_{k}\sum_{j\ne k}\sum_{i\ne k,j}\left(\sum_{l\ne i,j,k}G_{il}F_{il}\check{M}_{ik,-il}\check{M}_{lj,-ilk}R_{3l}\right)\left(\sum_{l\ne i,j,k}G_{il}F_{il}\check{M}_{ij,-il}\check{M}_{lk,-ilj}R_{3l}\right)\\
 & +\max_{i,j,k}\Var\left(V_{1i}v_{2k}v_{4j}\right)3!\sum_{i\ne j\ne k}^{n}\check{M}_{ik,-ij}^{2}\left(\sum_{l\ne i,j,k}G_{il}F_{il}\check{M}_{ik,-il}\check{M}_{lj,-ilk}R_{3l}\right)^{2}.
\end{align*}
\small

The first term in the $v_{2k}$ line is then:
\footnotesize
\begin{align*}
\Var &  \left(\sum_{i\ne j\ne k\ne l}^{n}G_{ij}F_{ij}\check{M}_{ik,-ij}\check{M}_{jl,-ijk}V_{1i}v_{2k}R_{3j}R_{4l}\right) =\Var\left(\sum_{i\ne j}^{n}G_{ij}F_{ij}\check{M}_{ij,-ik}V_{1i}v_{2j}\sum_{k\ne i,j}\sum_{l\ne i,j,k}\check{M}_{kl,-ijk}R_{3k}R_{4l}\right)\\
 & \leq\max_{i,j}\Var\left(V_{1i}v_{2j}\right)\sum_{i\ne j}^{n}\left(G_{ij}F_{ij}\sum_{k\ne i,j}\sum_{l\ne i,j,k}\check{M}_{ij,-ik}\check{M}_{kl,-ijk}R_{3k}R_{4l}\right)^{2}\\
 & +\max_{i,j}\Var\left(V_{1i}v_{2j}\right)\sum_{i\ne j}^{n}\left(G_{ij}F_{ij}\sum_{k\ne i,j}\sum_{l\ne i,j,k}\check{M}_{ij,-ik}\check{M}_{kl,-ijk}R_{3k}R_{4l}\right) \\
&\quad \left(G_{ji}F_{ji}\sum_{k\ne i,j}\sum_{l\ne i,j,k}\check{M}_{ji,-jk}\check{M}_{kl,-ijk}R_{3k}R_{4l}\right)\\
 & +\max_{i,j}\Var\left(V_{1i}v_{2j}\right)\sum_{i \ne j \ne i_2}^n \left(G_{ij}F_{ij}\sum_{k\ne i,j}\sum_{l\ne i,j,k}\check{M}_{ij,-ik}\check{M}_{kl,-ijk}R_{3k}R_{4l}\right) \\
 &\quad \left(G_{i_{2}j}F_{i_{2}j}\sum_{k\ne i_{2},j}\sum_{l\ne i_{2},j,k}\check{M}_{i_{2}j,-i_{2}k}\check{M}_{kl,-i_{2}jk}R_{3k}R_{4l}\right)\\
 & \leq\max_{i,j}\Var\left(V_{1i}v_{2j}\right)\sum_{i\ne j}^{n}\left(G_{ij}F_{ij}\sum_{k\ne i,j}\sum_{l\ne i,j,k}\check{M}_{ij,-ik}\check{M}_{kl,-ijk}R_{3k}R_{4l}\right)^{2}\\
 & +\max_{i,j}\Var\left(V_{1i}v_{2j}\right)\sum_{i\ne j}^{n}\left(G_{ij}F_{ij}\sum_{k\ne i,j}\sum_{l\ne i,j,k}\check{M}_{ij,-ik}\check{M}_{kl,-ijk}R_{3k}R_{4l}\right) \\
&\quad\left(G_{ji}F_{ji}\sum_{k\ne i,j}\sum_{l\ne i,j,k}\check{M}_{ji,-jk}\check{M}_{kl,-ijk}R_{3k}R_{4l}\right)\\
 & +\max_{i,j}\Var\left(V_{1i}v_{2j}\right)\sum_{j}\left(\sum_{i\ne j}\sum_{k\ne i,j}\sum_{l\ne i,j,k}G_{ij}F_{ij}\check{M}_{ij,-ik}\check{M}_{kl,-ijk}R_{3k}R_{4l}\right)^{2} \\
 &\quad - \max_{i,j}\Var\left(V_{1i}v_{2j}\right)\sum_j \sum_{i\ne j}\left(\sum_{k\ne i,j}\sum_{l\ne i,j,k}G_{ij}F_{ij}\check{M}_{ij,-ik}\check{M}_{kl,-ijk}R_{3k}R_{4l}\right)^{2}.
\end{align*}
\small

Now, we turn back to the $R_{2k}$ expression to complete the proof:

\[
\sum_{i\ne j\ne k\ne l}^{n}G_{ij}F_{ij}\check{M}_{ik,-ij}\check{M}_{jl,-ijk}V_{1i}R_{2k}\left(R_{3j}R_{4l}+R_{3j}v_{4l}+v_{3j}R_{4l}+v_{3j}v_{4l}\right).
\]

Consider the term with three stochastic terms first, and simplify it using the same strategy as before:
\footnotesize
\begin{align*}
 & \Var\left(\sum_{i\ne j\ne k\ne l}^{n}G_{ij}F_{ij}\check{M}_{ik,-ij}\check{M}_{jl,-ijk}V_{1i}R_{2k}v_{3j}v_{4l}\right) =\Var\left(\sum_{i\ne j\ne k}^{n}G_{ij}F_{ij}V_{1i}v_{3j}v_{4k}\sum_{l\ne i,j,k}\check{M}_{il,-ij}\check{M}_{jk,-ijl}R_{2l}\right)\\
 & \leq \max_{i,j,k}\Var\left(V_{1i}v_{3j}v_{4k}\right)\bigg(\sum_{k \ne j}^n \left(\sum_{i\ne k,j}\sum_{l\ne i,j,k}G_{ij}F_{ij}\check{M}_{il,-ij}\check{M}_{jk,-ijl}R_{2l}\right)^{2} \\
&\quad-\sum_{k \ne j}^n \sum_{i\ne k,j}\left(\sum_{l\ne i,j,k}G_{ij}F_{ij}\check{M}_{il,-ij}\check{M}_{jk,-ijl}R_{2l}\right)^{2}\bigg)\\
 & +\max_{i,j,k}\Var\left(V_{1i}v_{3j}v_{4k}\right)\sum_{k}\sum_{j\ne k}\left(\sum_{i\ne k,j}\sum_{l\ne i,j,k}G_{ij}F_{ij}\check{M}_{il,-ij}\check{M}_{jk,-ijl}R_{2l}\right)\left(\sum_{i\ne k,j}\sum_{l\ne i,j,k}G_{ik}F_{ik}\check{M}_{il,-ik}\check{M}_{kj,-ikl}R_{2l}\right)\\
 & -\max_{i,j,k}\Var\left(V_{1i}v_{3j}v_{4k}\right)\sum_{k}\sum_{j\ne k}\sum_{i\ne k,j}\left(\sum_{l\ne i,j,k}G_{ij}F_{ij}\check{M}_{il,-ij}\check{M}_{jk,-ijl}R_{2l}\right)\left(\sum_{l\ne i,j,k}G_{ik}F_{ik}\check{M}_{il,-ik}\check{M}_{kj,-ikl}R_{2l}\right)\\
 & +\max_{i,j,k}\Var\left(V_{1i}v_{3j}v_{4k}\right)3!\sum_{i\ne j\ne k}^{n}\left(G_{ij}F_{ij}\sum_{l\ne i,j,k}\check{M}_{il,-ij}\check{M}_{jk,-ijl}R_{2l}\right)^{2}.
\end{align*}

Next,
\begin{align*}
 & \Var\left(\sum_{i}\sum_{j\ne i}\sum_{k\ne i,j}\sum_{l\ne i,j,k}G_{ij}F_{ij}\check{M}_{ik,-ij}\check{M}_{jl,-ijk}V_{1i}R_{2k}v_{3j}R_{4l}\right)\\
 & \leq \max_{i,j}\Var\left(V_{1i}v_{3j}\right)\sum_{i\ne j}^{n} \left(G_{ij}F_{ij}\sum_{k\ne i,j}\sum_{l\ne i,j,k}\check{M}_{ik,-ij}\check{M}_{jl,-ijk}R_{2k}R_{4l}\right)^{2} \\
 &+ \max_{i,j}\Var\left(V_{1i}v_{3j}\right)\sum_{i\ne j}^{n}\left(G_{ij}F_{ij}\sum_{k\ne i,j}\sum_{l\ne i,j,k}\check{M}_{ik,-ij}\check{M}_{jl,-ijk}R_{2k}R_{4l}\right)\left(G_{ji}F_{ji}\sum_{k\ne i,j}\sum_{l\ne i,j,k}\check{M}_{jk,-ij}\check{M}_{il,-ijk}R_{2k}R_{4l}\right) \\
 & +\max_{i,j}\Var\left(V_{1i}v_{3j}\right)\sum_{j}\left(\sum_{i\ne j}G_{ij}F_{ij}\sum_{k\ne i,j}\sum_{l\ne i,j,k}\check{M}_{ik,-ij}\check{M}_{jl,-ijk}R_{2k}R_{4l}\right)^{2}\\
 &-\max_{i,j}\Var\left(V_{1i}v_{3j}\right)\sum_{j\ne i}^n\left(G_{ij}F_{ij}\sum_{k\ne i,j}\sum_{l\ne i,j,k}\check{M}_{ik,-ij}\check{M}_{jl,-ijk}R_{2k}R_{4l}\right)^{2}.
\end{align*}
\small

Finally,
\begin{align*}
\Var&\left(\sum_{i\ne j\ne k\ne l}^{n}G_{ij}F_{ij}\check{M}_{ik,-ij}\check{M}_{jl,-ijk}V_{1i}R_{2k}R_{3j}R_{4l}\right) \\
&\quad=\sum_{i}\Var\left(V_{1i}\right)\left(\sum_{j\ne i}\sum_{k\ne i,j}\sum_{l\ne i,j,k}G_{ij}F_{ij}\check{M}_{ik,-ij}\check{M}_{jl,-ijk}R_{2k}R_{3j}R_{4l}\right)^{2}.
\end{align*}
\end{proof}

\begin{proof} [Proof of \Cref{lem:a2}]
\textbf{Proof of \Cref{lem:a2}(a).}
Expand the term:
\[
\sum_{i\ne j\ne k}^{n}G_{ij}F_{ik}V_{1j}V_{2k}V_{3i}V_{4i}=\sum_{i\ne j\ne k}^{n}G_{ij}F_{ik}V_{3i}V_{4i}\left(R_{1j}R_{2k}+R_{1j}v_{2k}+v_{1j}R_{2k}+v_{1j}v_{2k}\right).
\]

With four stochastic objects, 
\begin{align*}
Var  \left(\sum_{i\ne j\ne k}^{n}G_{ij}F_{ik}V_{3i}V_{4i}v_{1j}v_{2k}\right)
&\leq  \max_{i,j,k}\Var\left(V_{3i}V_{4i}v_{1j}v_{2k}\right)\sum_{i\ne j\ne k}^{n}\sum_{i_{2}\ne i,j,k}\left(G_{ij}F_{ik}G_{i_{2}j}F_{i_{2}k}+G_{ij}F_{ik}G_{i_{2}k}F_{i_{2}j}\right)\\
 & +\max_{i,j,k}\Var\left(V_{1i}V_{2i}v_{3j}v_{4l}\right)3!\sum_{i\ne j\ne k}^{n}G_{ij}^{2}F_{ik}^{2}.
\end{align*}

Observe that, due to \Cref{asmp:Vconst_reg}(a), 
\begin{align*}
\sum_{i\ne j\ne k\ne l}^{n}G_{ij}F_{ik}G_{lj}F_{lk} & =\sum_{j\ne k}^{n}\left(\sum_{i\ne j,k}G_{ij}F_{ik}\right)\left(\sum_{l\ne j,k}G_{lj}F_{lk}-G_{ij}F_{ik}\right)\\
 & =\sum_{j\ne k}^{n}\left(\sum_{i\ne j,k}G_{ij}F_{ik}\right)^{2}-\sum_{j\ne k\ne i}^{n}G_{ij}^{2}F_{ik}^{2}
\end{align*}
has the required order, which suffices for the bound. Next,
\begin{align*}
\Var &  \left(\sum_{i\ne j\ne k}^{n}G_{ij}F_{ik}V_{3i}V_{4i}R_{1j}v_{2k}\right)\\
 & =\Var\left(\sum_{i}\sum_{j\ne i}F_{ij}V_{3i}V_{4i}v_{2j}\left(\sum_{k\ne i,j}G_{ik}R_{1k}\right)\right)\\
 & \leq\sum_{i}\sum_{j\ne i}\Var\left(V_{3i}V_{4i}v_{2j}\right)\left(\sum_{k\ne i,j}F_{ij}G_{ik}R_{1k}\right)\left[\sum_{k\ne i,j}F_{ij}G_{ik}R_{1k}+\sum_{k\ne i,j}F_{ji}G_{jk}R_{1k}\right]\\
 & +\max_{i,j}\Var\left(V_{3i}V_{4i}v_{2j}\right)\sum_{i}\sum_{j\ne i}\sum_{i_{2}\ne i,j}\left(\sum_{k\ne i,j}F_{ij}G_{ik}R_{1k}\right)\left(\sum_{k\ne i_{2},l}F_{i_{2}j}G_{i_{2}k}R_{1k}\right)\\
 & \leq\max_{i,j}\Var\left(V_{3i}V_{4i}v_{2j}\right)\sum_{i}\left(\sum_{j\ne i}\sum_{k\ne i,j}F_{ij}G_{ik}R_{1k}\right)^{2}\\
&\quad+\sum_{i}\sum_{j\ne i}\Var\left(V_{3i}V_{4i}v_{2j}\right)\left(\sum_{k\ne i,j}F_{ij}G_{ik}R_{1k}\right)\left(\sum_{k\ne i,j}F_{ji}G_{jk}R_{1k}\right).
\end{align*}

Similarly,
\begin{align*}
\Var &  \left(\sum_{i\ne j\ne k}^{n}G_{ij}F_{ik}V_{3i}V_{4i}v_{1j}R_{2k}\right) =\Var\left(\sum_{i}\sum_{j\ne i}V_{3i}V_{4i}v_{1j}\left(\sum_{k\ne i,j}G_{ij}F_{ik}R_{2k}\right)\right)\\
 & \leq\max_{i,j}\Var\left(V_{3i}V_{4i}v_{1j}\right)\sum_{i}\left(\left(\sum_{j\ne i}\sum_{k\ne i,j}G_{ij}F_{ik}R_{2k}\right)^{2}+\sum_{j\ne i}\left(\sum_{k\ne i,j}G_{ij}F_{ik}R_{2k}\right)\left(\sum_{k\ne i,j}G_{ji}F_{jk}R_{2k}\right)\right).
\end{align*}

Turning to the sum with two stochastic objects,
\begin{align*}
\Var &  \left(\sum_{i\ne j\ne k}^{n}G_{ij}F_{ik}V_{3i}V_{4i}R_{1j}R_{2k}\right)=\Var\left(\sum_{i}V_{3i}V_{4i}\left(\sum_{j\ne i}\sum_{k\ne i,j}G_{ij}F_{ik}R_{1j}R_{2k}\right)\right)\\
 & \leq\max_{i}\Var\left(V_{3i}V_{4i}\right)\sum_{i}\left(\sum_{j\ne i}\sum_{k\ne i,j}G_{ij}F_{ik}R_{1j}R_{2k}\right)^{2}.
\end{align*}

\textbf{Proof of \Cref{lem:a2}(b).}

Decompose the term:
\begin{align*}
 & \sum_{i\ne j\ne k\ne l}^{n}G_{ij}F_{ik}\check{M}_{il,-ijk}V_{1j}V_{2k}V_{3i}V_{4l}\\
 & =\sum_{i\ne j\ne k\ne l}^{n}G_{ij}F_{ik}\check{M}_{il,-ijk}V_{3i}R_{1j}\left(R_{2k}R_{4l}+R_{2k}v_{4l}+v_{2k}R_{4l}+v_{2k}v_{4l}\right)\\
 & \quad+\sum_{i\ne j\ne k\ne l}^{n}G_{ij}F_{ik}\check{M}_{il,-ijk}V_{3i}v_{1j}\left(R_{2k}R_{4l}+R_{2k}v_{4l}+v_{2k}R_{4l}+v_{2k}v_{4l}\right).
\end{align*}

Consider the $v_{1j}$ line first. 

\begin{align*}
 & \Var\left(\sum_{i\ne j\ne k\ne l}^{n}G_{ij}F_{ik}\check{M}_{il,-ijk}V_{3i}v_{1j}v_{2k}v_{4l}\right)\leq\max_{i,j,k,l}\Var\left(V_{3i}v_{1j}v_{2k}v_{4l}\right)4!\sum_{i\ne j\ne k\ne l}^{n}G_{ij}^{2}F_{ik}^{2}\check{M}_{il,-ijk}^{2}.
\end{align*}

Next, by using the same expansion and simplification steps as before,
\begin{align*}
 & \Var\left(\sum_{i\ne j\ne k\ne l}^{n}G_{ij}F_{ik}\check{M}_{il,-ijk}V_{3i}v_{1j}v_{2k}R_{4l}\right)
 =\Var\left(\sum_{i\ne j\ne k}^{n}G_{ij}F_{ik}V_{3i}v_{1j}v_{2k}\sum_{l\ne i,j,k}\check{M}_{il,-ijk}R_{4l}\right)\\
 & \leq\max_{i,j,k}\Var\left(V_{3i}v_{1j}v_{2k}\right)\sum_{k}\sum_{j\ne k}\left(\left(\sum_{i\ne j,k}\sum_{l\ne i,j,k}G_{ij}F_{ik}\check{M}_{il,-ijk}R_{4l}\right)^{2}-\sum_{i\ne j,k}\left(\sum_{l\ne i,j,k}G_{ij}F_{ik}\check{M}_{il,-ijk}R_{4l}\right)^{2}\right)\\
 & +\max_{i,j,k}\Var\left(V_{3i}v_{1j}v_{2k}\right)\sum_{k}\sum_{j\ne k}\left(\sum_{i\ne j,k}\sum_{l\ne i,j,k}G_{ij}F_{ik}\check{M}_{il,-ijk}R_{4l}\right)\left(\sum_{i\ne j,k}\sum_{l\ne i,j,k}G_{ik}F_{ij}\check{M}_{il,-ijk}R_{4l}\right)\\
 & -\max_{i,j,k}\Var\left(V_{3i}v_{1j}v_{2k}\right)\sum_{k}\sum_{j\ne k}\sum_{i\ne j,k}\left(\sum_{l\ne i,j,k}G_{ij}F_{ik}\check{M}_{il,-ijk}R_{4l}\right)\left(\sum_{l\ne i,j,k}G_{ik}F_{ij}\check{M}_{il,-ijk}R_{4l}\right)\\
 & +\max_{i,j,k}\Var\left(V_{3i}v_{1j}v_{2k}\right)3!\sum_{i\ne j\ne k}^{n}G_{ij}^{2}F_{ik}^{2}\left(\sum_{l\ne i,j,k}\check{M}_{il,-ijk}R_{4l}\right)^{2}
\end{align*}
and
\begin{align*}
 & \Var\left(\sum_{i\ne j\ne k\ne l}^{n}G_{ij}F_{ik}\check{M}_{il,-ijk}V_{3i}v_{1j}R_{2k}v_{4l}\right)
 =\Var\left(\sum_{i\ne j\ne k}^{n}G_{ij}F_{ik}V_{3i}v_{1j}v_{4k}\sum_{l\ne i,j,k}\check{M}_{ik,-ijl}R_{2l}\right)\\
 & \leq\max_{i,j,k}\Var\left(V_{3i}v_{1j}v_{4k}\right)\sum_{k}\sum_{j\ne k}\left(\left(\sum_{i\ne j,k}\sum_{l\ne i,j,k}G_{ij}F_{ik}\check{M}_{ik,-ijl}R_{2l}\right)^{2}-\sum_{i\ne j,k}\left(\sum_{l\ne i,j,k}G_{ij}F_{ik}\check{M}_{ik,-ijl}R_{2l}\right)^{2}\right)\\
 & +\max_{i,j,k}\Var\left(V_{3i}v_{1j}v_{2k}\right)\sum_{k}\sum_{j\ne k}\left(\sum_{i\ne j,k}\sum_{l\ne i,j,k}G_{ij}F_{ik}\check{M}_{ik,-ijl}R_{2l}\right)\left(\sum_{i\ne j,k}\sum_{l\ne i,j,k}G_{ik}F_{ij}\check{M}_{il,-ijk}R_{2l}\right)\\
 & -\max_{i,j,k}\Var\left(V_{3i}v_{1j}v_{2k}\right)\sum_{k}\sum_{j\ne k}\sum_{i\ne j,k}\left(\sum_{l\ne i,j,k}G_{ij}F_{ik}\check{M}_{ik,-ijl}R_{2l}\right)\left(\sum_{l\ne i,j,k}G_{ik}F_{ij}\check{M}_{il,-ijk}R_{2l}\right)\\
 & +\max_{i,j,k}\Var\left(V_{3i}v_{1j}v_{4k}\right)3!\sum_{i\ne j\ne k}^{n}G_{ij}^{2}F_{ik}^{2}\left(\sum_{l\ne i,j,k}\check{M}_{ik,-ijl}R_{2l}\right)^{2}
\end{align*}
with $\left(\sum_{l\ne i,j,k}\check{M}_{ik,-ijl}R_{2l}\right)^{2}\leq C$.
Finally,
\footnotesize
\begin{align*}
 & \Var\left(\sum_{i\ne j\ne k\ne l}^{n}G_{ij}F_{ik}\check{M}_{il,-ijk}V_{3i}v_{1j}R_{2k}R_{4l}\right)
 =\Var\left(\sum_{i\ne j}^{n}G_{ij}V_{3i}v_{1j}\sum_{k\ne i,j}\sum_{l\ne i,j,k}F_{ik}\check{M}_{il,-ijk}R_{2k}R_{4l}\right)\\
 & \leq\max_{i,j}\Var\left(V_{3i}v_{1j}\right)\sum_{i\ne j}^{n}G_{ij}\left(\sum_{k\ne i,j}\sum_{l\ne i,j,k}F_{ik}\check{M}_{il,-ijk}R_{2k}R_{4l}\right)^{2}\\
 & +\max_{i,j}\Var\left(V_{3i}v_{1j}\right)\sum_{i\ne j}^{n}\left(\sum_{k\ne i,j}\sum_{l\ne i,j,k}G_{ij}F_{ik}\check{M}_{il,-ijk}R_{2k}R_{4l}\right)\left(\sum_{k\ne i,j}\sum_{l\ne i,j,k}G_{ji}F_{jk}\check{M}_{jl,-ijk}R_{2k}R_{4l}\right)\\
 & +\max_{i,j}\Var\left(V_{3i}v_{1j}\right)\sum_{j}\left(\left(\sum_{i\ne j}\sum_{k\ne i,j}\sum_{l\ne i,j,k}G_{ij}F_{ik}\check{M}_{il,-ijk}R_{2k}R_{4l}\right)^{2}-\sum_{i\ne j}\left(\sum_{k\ne i,j}\sum_{l\ne i,j,k}G_{ij}F_{ik}\check{M}_{il,-ijk}R_{2k}R_{4l}\right)^{2}\right).
\end{align*}
\small

Now, return to the $R_{1j}$ line: $\sum_{i\ne j\ne k\ne l}^{n}G_{ij}F_{ik}\check{M}_{il,-ijk}V_{3i}R_{1j}\left(R_{2k}R_{4l}+R_{2k}v_{4l}+v_{2k}R_{4l}+v_{2k}v_{4l}\right)$,
so
\footnotesize
\begin{align*}
 & \Var\left(\sum_{i\ne j\ne k\ne l}^{n}G_{ij}F_{ik}\check{M}_{il,-ijk}V_{3i}R_{1j}v_{2k}v_{4l}\right)
 =\Var\left(\sum_{i\ne j\ne k}^{n}G_{il}F_{ik}V_{3i}v_{2k}v_{4j}\sum_{l\ne i,j,k}\check{M}_{ij,-ilk}R_{1l}\right)\\
 & \leq\max_{i,j,k}\Var\left(V_{3i}v_{2k}v_{4j}\right)\left(\sum_{j}\sum_{k\ne j}\left(\sum_{i\ne j,k}\sum_{l\ne i,j,k}G_{il}F_{ik}\check{M}_{ij,-ilk}R_{1l}\right)^{2}-\sum_{j}\sum_{k\ne j}\sum_{i\ne j,k}\left(\sum_{l\ne i,j,k}G_{il}F_{ik}\check{M}_{ij,-ilk}R_{1l}\right)^{2}\right)\\
 & +\max_{i,j,k}\Var\left(V_{3i}v_{2k}v_{4j}\right)\sum_{j}\sum_{k\ne j}\left(\sum_{i\ne j,k}\sum_{l\ne i,j,k}G_{il}F_{ik}\check{M}_{ij,-ilk}R_{1l}\right)\left(\sum_{i\ne j,k}\sum_{l\ne i,j,k}G_{il}F_{ij}\check{M}_{ik,-ilj}R_{1l}\right)\\
 & -\max_{i,j,k}\Var\left(V_{3i}v_{2k}v_{4j}\right)\sum_{j}\sum_{k\ne j}\sum_{i\ne j,k}\left(\sum_{l\ne i,j,k}G_{il}F_{ik}\check{M}_{ij,-ilk}R_{1l}\right)\left(\sum_{l\ne i,j,k}G_{il}F_{ij}\check{M}_{ik,-ilj}R_{1l}\right)\\
 & +\max_{i,j,k}\Var\left(V_{3i}v_{2k}v_{4j}\right)3!\sum_{i\ne j\ne k}^{n}\left(F_{ik}\sum_{l\ne i,j,k}G_{il}\check{M}_{ij,-ilk}R_{1l}\right)^{2},
\end{align*}
and
\begin{align*}
 & \Var\left(\sum_{i\ne j\ne k\ne l}^{n}G_{ij}F_{ik}\check{M}_{il,-ijk}V_{3i}R_{1j}v_{2k}R_{4l}\right)
 =\Var\left(\sum_{i\ne j}^{n}F_{ij}V_{3i}v_{2j}\sum_{k\ne i,j}\sum_{l\ne i,j,k}G_{ik}\check{M}_{il,-ijk}R_{1k}R_{4l}\right)\\
 & \leq\max_{i,j}\Var\left(V_{3i}v_{2j}\right)\sum_{i\ne j}^{n}\left(\sum_{k\ne i,j}\sum_{l\ne i,j,k}F_{ij}G_{ik}\check{M}_{il,-ijk}R_{1k}R_{4l}\right)^{2}\\
 & +\max_{i,j}\Var\left(V_{3i}v_{2j}\right)\sum_{i\ne j}^{n}\left(\sum_{k\ne i,j}\sum_{l\ne i,j,k}F_{ij}G_{ik}\check{M}_{il,-ijk}R_{1k}R_{4l}\right)\left(\sum_{k\ne i,j}\sum_{l\ne i,j,k}F_{ij}G_{jk}\check{M}_{jl,-ijk}R_{1k}R_{4l}\right)\\
 & +\max_{i,j}\Var\left(V_{3i}v_{2j}\right)\sum_{j}\left(\left(\sum_{i\ne j}\sum_{k\ne i,j}\sum_{l\ne i,j,k}F_{ij}G_{ik}\check{M}_{il,-ijk}R_{1k}R_{4l}\right)^{2}-\sum_{i\ne j}\left(\sum_{k\ne i,j}\sum_{l\ne i,j,k}F_{ij}G_{ik}\check{M}_{il,-ijk}R_{1k}R_{4l}\right)^{2}\right).
\end{align*}
\small

The $\sum_{i\ne j\ne k\ne l}^{n}G_{ij}F_{ik}\check{M}_{il,-ijk}V_{3i}R_{1j}R_{2k}v_{4l}$
term is symmetric, because it does not matter which $R_{m}$ we use.
Finally,
\begin{align*}
\Var\left(\sum_{i\ne j\ne k\ne l}^{n}G_{ij}F_{ik}\check{M}_{il,-ijk}V_{3i}R_{1j}R_{2k}R_{4l}\right) & =\sum_{i}\Var\left(V_{3i}\right)\left(\sum_{j\ne i}\sum_{k\ne i,j}\sum_{l\ne i,j,k}G_{ij}F_{ik}\check{M}_{il,-ijk}R_{1j}R_{2k}R_{4l}\right)^{2}.
\end{align*}
\end{proof}

\begin{proof} [Proof of \Cref{lem:a3}]
The proof of \Cref{lem:a3} is entirely analogous to Lemmas \ref{lem:a1} and \ref{lem:a2} just that $G_{ji}$ is used in place of $G_{ij}$.
\end{proof}

\subsection{Proofs for Appendix D} \label{sec:proof_appD}

\subsubsection{Proofs for Propositions in Appendix D} 

\begin{proof}[Proof of \Cref{prop:max_inv_distr}]
Let
\[
\left(\begin{array}{c}
\Pi_{Y}\\
\Pi
\end{array}\right):=\left(\begin{array}{c}
\left(Z^{\prime}Z\right)^{1/2}\pi_{Y}\\
\left(Z^{\prime}Z\right)^{1/2}\pi
\end{array}\right).
\]

With this definition, $\left(\pi_{Y}^{\prime}Z^{\prime}Z\pi_{Y},\pi^{\prime}Z^{\prime}Z\pi_{Y},\pi^{\prime}Z^{\prime}Z\pi\right)=\left(\Pi_{Y}^{\prime}\Pi_{Y},\Pi_{Y}^{\prime}\Pi,\Pi^{\prime}\Pi\right)$, and 
\[
\left(\begin{array}{c}
s_{1}\\
s_{2}
\end{array}\right)\sim N\left(\left(\begin{array}{c}
\Pi_{Y}\\
\Pi
\end{array}\right),\Omega\otimes I_{K}\right).
\]

Split $s_{1}$ and $s_{2}$ into the $\Pi$ component and a random
normal component: $s_{1k}=\Pi_{Yk}+z_{1k}$ and $s_{2k}=\Pi_{k}+z_{2k}$.
Then, for all $k$,
\[
\left(\begin{array}{c}
z_{1k}\\
z_{2k}
\end{array}\right)\sim N\left(\left(\begin{array}{c}
0\\
0
\end{array}\right),\left[\begin{array}{cc}
\omega_{\zeta \zeta} & \omega_{\zeta\eta}\\
\omega_{\zeta\eta} & \omega_{\eta \eta}
\end{array}\right]\right), \text{ and }
\]
\begin{align*}
\left(\begin{array}{c}
s_{1}^{\prime}s_{1}\\
s_{1}^{\prime}s_{2}\\
s_{2}^{\prime}s_{2}
\end{array}\right) & =\left(\begin{array}{c}
\sum_{k}s_{1k}^{2}\\
\sum_{k}s_{1k}s_{2k}\\
\sum_{k}s_{2k}^{2}
\end{array}\right)
=\left(\begin{array}{c}
\sum_{k}\left(\Pi_{Yk}+z_{1k}\right)^{2}\\
\sum_{k}\left(\Pi_{Yk}+z_{1k}\right)\left(\Pi_{k}+z_{2k}\right)\\
\sum_{k}\left(\Pi_{k}+z_{2k}\right)^{2}
\end{array}\right)\\
 & =\left(\begin{array}{c}
\sum_{k}\Pi_{Yk}^{2}+2\sum_{k}\Pi_{Yk}z_{1k}+\sum_{k}z_{1k}^{2}\\
\sum_{k}\Pi_{Yk}\Pi_{k}+\sum_{k}\Pi_{Yk}z_{2k}+\sum_{k}\Pi_{k}z_{1k}+\sum_{k}z_{1k}z_{2k}\\
\sum_{k}\Pi_{k}^{2}+2\sum_{k}\Pi_{k}z_{2k}+\sum_{k}z_{2k}^{2}
\end{array}\right).
\end{align*}

Under the assumption, $\Pi^{\prime}\Pi/\sqrt{K}\rightarrow C_{S}$, so $\frac{1}{\sqrt{K}}\sum_{k}\Pi_{k}^{2}\rightarrow C_{S}$. 
By applying the Lindeberg CLT due to bounded moments, 
\[
\frac{1}{\sqrt{K}}\left(\begin{array}{c}
\sum_{k}\Pi_{k}z_{1k}\\
\sum_{k}\Pi_{Yk}z_{1k}\\
\sum_{k}\Pi_{Yk}z_{2k}\\
\sum_{k}\Pi_{k}z_{2k}\\
\sum_{k}z_{1k}z_{2k}\\
\sum_{k}z_{2k}^{2}\\
\sum_{k}z_{1k}^{2}
\end{array}\right)\stackrel{a}{\sim}N\left(\left(\begin{array}{c}
0\\
0\\
0\\
0\\
\sqrt{K}\omega_{\zeta\eta}\\
\sqrt{K}\omega_{\eta \eta}\\
\sqrt{K}\omega_{\zeta \zeta}
\end{array}\right),V\right),
\]
where $V$ is some variance matrix. 
By assumption, $\frac{1}{\sqrt{K}}\sum_{k}\Pi_{Yk}\Pi_{k}\rightarrow C_{Y}$ and $\frac{1}{\sqrt{K}}\sum_{k}\Pi_{Yk}^{2}\rightarrow C_{YY}$, so
\begin{align*}
\frac{1}{\sqrt{K}}\left(\begin{array}{c}
s_{1}^{\prime}s_{1}\\
s_{1}^{\prime}s_{2}\\
s_{2}^{\prime}s_{2}
\end{array}\right) & =\frac{1}{\sqrt{K}}\left(\begin{array}{c}
\sum_{k}\Pi_{Yk}^{2}+2\sum_{k}\Pi_{Yk}z_{1k}+\sum_{k}z_{1k}^{2}\\
\sum_{k}\Pi_{Yk}\Pi_{k}+\sum_{k}\Pi_{Yk}z_{2k}+\sum_{k}\Pi_{k}z_{1k}+\sum_{k}z_{1k}z_{2k}\\
\sum_{k}\Pi_{k}^{2}+2\sum_{k}\Pi_{k}z_{2k}+\sum_{k}z_{2k}^{2}
\end{array}\right)\\
 & \stackrel{a}{\sim}\left(\begin{array}{c}
C_{YY}\\
C_{Y}\\
C
\end{array}\right)+A\frac{1}{\sqrt{K}}\left(\begin{array}{c}
\sum_{k}\Pi_{k}z_{1k}\\
\sum_{k}\Pi_{Yk}z_{1k}\\
\sum_{k}\Pi_{Yk}z_{2k}\\
\sum_{k}\Pi_{k}z_{2k}\\
\sum_{k}z_{1k}z_{2k}\\
\sum_{k}z_{2k}^{2}\\
\sum_{k}z_{1k}^{2}
\end{array}\right), \text{ where }
\end{align*}

\[
A=\left(\begin{array}{ccccccc}
0 & 2 & 0 & 0 & 0 & 0 & 1\\
1 & 0 & 1 & 0 & 1 & 0 & 0\\
0 & 0 & 0 & 2 & 0 & 1 & 0
\end{array}\right).
\]

This means:
\begin{align*}
\frac{1}{\sqrt{K}}\left(\begin{array}{c}
s_{1}^{\prime}s_{1}\\
s_{1}^{\prime}s_{2}\\
s_{2}^{\prime}s_{2}
\end{array}\right) & \stackrel{a}{\sim}N\left(\left(\begin{array}{c}
C_{YY}+\sqrt{K}\omega_{\zeta \zeta}\\
C_{Y}+\sqrt{K}\omega_{\zeta\eta}\\
C+\sqrt{K}\omega_{\eta \eta}
\end{array}\right),AVA^{\prime}\right).
\end{align*}

Let $\Sigma=AVA^{\prime}$ to obtain the result as stated. 
To derive $\Sigma$ explicitly, I derive $V$ by applying the Isserlis' Theroem. 
As a special case of the Isserlis' Theorem for $X$'s that are multivariate normal and mean zero, 
\[
E\left[X_{1}X_{2}X_{3}X_{4}\right]=E\left[X_{1}X_{2}\right]E\left[X_{3}X_{4}\right]+E\left[X_{1}X_{3}\right]E\left[X_{2}X_{4}\right]+E\left[X_{1}X_{4}\right]E\left[X_{2}X_{3}\right].
\]

Another corrolary is that if $n$ is odd, then there is no such pairing, so the moment is always zero. 
Hence,
\begin{align*}
E\left[z_{1k}^{2}z_{2k}^{2}\right] & =E\left[z_{1k}^{2}\right]E\left[z_{2k}^{2}\right]+2E\left[z_{1k}z_{2k}\right]E\left[z_{1k}z_{2k}\right]
=\omega_{\zeta \zeta}\omega_{\eta \eta}+2\omega_{\zeta\eta}^{2}, \text{ and } \\
\Var\left(z_{1k}z_{2k}\right) & =\omega_{\zeta \zeta}\omega_{\eta \eta}+\omega_{\zeta\eta}^{2}.
\end{align*}

Similarly,
\begin{align*}
\Var\left(z_{2k}^{2}\right) & =E\left[z_{2k}^{4}\right]-\omega_{\eta \eta}^2 =3\omega_{\eta \eta}^2-\omega_{\eta \eta}^2=2\omega_{\eta \eta}^2, \\
Cov\left(z_{1k},z_{1k}z_{2k}\right)&=E\left[z_{1k}^{2}z_{2k}\right]-E\left[z_{1k}\right]E\left[z_{1k}z_{2k}\right]=0, \\
Cov\left(z_{1k}^{2},z_{1k}z_{2k}\right) & =E\left[z_{1k}^{3}z_{2k}\right]-E\left[z_{1k}^{2}\right]E\left[z_{1k}z_{2k}\right]\\
 & =3\omega_{\zeta\eta}\omega_{\zeta \zeta}-\omega_{\zeta \zeta}\omega_{\zeta\eta} =2\omega_{\zeta\eta}\omega_{\zeta \zeta}, \\
Cov\left(z_{1k}^{2},z_{2k}^{2}\right) & =E\left[z_{1k}^{2}z_{2k}^{2}\right]-\omega_{\zeta \zeta}\omega_{\eta \eta} =2\omega_{\zeta\eta}^{2}, \text{ and }
\end{align*}

\[
V=\left[\begin{array}{cc}
V_{11} & 0 \\ 0 & V_{22}
\end{array}\right], \text{ where }
\]

\[
V_{11}=\left[\begin{array}{ccccccc}
\frac{1}{K}\sum_{k}\Pi_{k}^{2}\omega_{\zeta \zeta} & \frac{1}{K}\sum_{k}\Pi_{k}\Pi_{Yk}\omega_{\zeta \zeta} & \frac{1}{K}\sum_{k}\Pi_{k}\Pi_{Yk}\omega_{\zeta\eta} & \frac{1}{K}\sum_{k}\Pi_{k}^{2}\omega_{\zeta\eta} \\
. & \frac{1}{K}\sum_{k}\Pi_{Yk}^{2}\omega_{\zeta \zeta} & \frac{1}{K}\sum_{k}\Pi_{Yk}^{2}\omega_{\zeta\eta} & \frac{1}{K}\sum_{k}\Pi_{k}\Pi_{Yk}\omega_{\zeta\eta} \\
. & . & \frac{1}{K}\sum_{k}\Pi_{Yk}^{2}\omega_{\eta \eta} & \frac{1}{K}\sum_{k}\Pi_{k}\Pi_{Yk}\omega_{\eta \eta} \\
. & . & . & \frac{1}{K}\sum_{k}\Pi_{k}^{2}\omega_{\eta \eta} 
\end{array}\right],
\]

\[
V_{22}=\left[\begin{array}{ccccccc}
\omega_{\zeta \zeta}\omega_{\eta \eta}+\omega_{\zeta\eta}^{2} & 2\omega_{\zeta\eta}\omega_{\eta \eta} & 2\omega_{\zeta\eta}\omega_{\zeta \zeta}\\
 .& 2\omega_{\eta \eta}^2 & 2\omega_{\zeta\eta}^{2}\\
 .& . & 2\omega_{\zeta \zeta}^2
\end{array}\right].
\]

If $\frac{1}{K}\sum_{k}\Pi_{k}^{2}\rightarrow0,\frac{1}{K}\sum_{k}\Pi_{k}\Pi_{Yk}\rightarrow0,\frac{1}{K}\sum_{k}\Pi_{Yk}^{2}\rightarrow0$ under weak identification, then we obtain the $\Sigma$ expression stated in the proposition.
\end{proof}

\begin{proof} [Proof of \Cref{prop:UJIVE_restr}]
Use $n_{q}^{Q}$ and $n_{w}^{W}$ to denote the number of observations in the instrument and covariate groups respectively, so 
\begin{align*}
\mu_{3} & =\sum_{i}\sum_{j\ne i}G_{ij}R_{i}R_{j} =\sum_{q}\frac{n_{q}^{Q}}{n_{q}^{Q}-1}\sum_{i\in\mathcal{N}_{q}^{Q}}\sum_{j\in\mathcal{N}_{q}^{Q},j\ne i}\frac{1}{n_{q}^{Q}}R_{i}R_{j}-\sum_{w}\frac{n_{w}^{W}}{n_{w}^{W}-1}\sum_{i\in\mathcal{N}_{w}^{W}}\sum_{j\in\mathcal{N}_{w}^{W},j\ne i}\frac{1}{n_{w}^{W}}R_{i}R_{j}\\
 & =\sum_{q}\frac{1}{n_{q}^{Q}-1}\sum_{i\in\mathcal{N}_{q}^{Q}}\sum_{j\in\mathcal{N}_{q}^{Q},j\ne i}R_{i}R_{j}-\sum_{w}\frac{1}{n_{w}^{W}-1}\sum_{i\in\mathcal{N}_{w}^{W}}\sum_{j\in\mathcal{N}_{w}^{W},j\ne i}R_{i}R_{j}\\
 & =\sum_{w}\left(\sum_{q\in\mathcal{M}_{w}}\frac{1}{n_{q}^{Q}-1}\sum_{i\in\mathcal{N}_{q}^{Q}}\sum_{j\in\mathcal{N}_{q}^{Q},j\ne i}R_{i}R_{j}-\frac{1}{n_{w}^{W}-1}\sum_{i\in\mathcal{N}_{w}^{W}}\sum_{j\in\mathcal{N}_{w}^{W},j\ne i}R_{i}R_{j}\right)\\
 & =\sum_{w}\left(\sum_{q\in\mathcal{M}_{w}}\frac{n_{q}^{Q}\left(n_{q}^{Q}-1\right)}{n_{q}^{Q}-1}\left(\pi_{q}+\gamma_{w}\right)^{2}-\frac{1}{n_{w}^{W}-1}\sum_{i\in\mathcal{N}_{w}^{W}}\sum_{j\in\mathcal{N}_{w}^{W},j\ne i}R_{i}R_{j}\right).
\end{align*}

Considering the second term,
\begin{align*}
\sum_{i\in\mathcal{N}_{w}^{W}}\sum_{j\in\mathcal{N}_{w}^{W},j\ne i}R_{i}R_{j} & =\sum_{i\in\mathcal{N}_{w}^{W}}\sum_{j\in\mathcal{N}_{w}^{W},j\ne i}\left(\pi_{q(i)}+\gamma_{w}\right)\left(\pi_{q(j)}+\gamma_{w}\right)\\
 & =\sum_{q\in\mathcal{M}_{w}}\sum_{i\in\mathcal{N}_{q}}\sum_{j\in\mathcal{N}_{w}^{W},j\ne i}\left(\pi_{q(i)}+\gamma_{w}\right)\left(\pi_{q(j)}+\gamma_{w}\right)\\
 & =\sum_{q\in\mathcal{M}_{w}}n_{q}^{Q}\left(n_{q}^{Q}-1\right)\left(\pi_{q}+\gamma_{w}\right)^{2}+\sum_{q\in\mathcal{M}_{w}}\sum_{i\in\mathcal{N}_{q}}\sum_{j\in\mathcal{N}_{w}^{W},j\ne i}\left(\pi_{q(i)}+\gamma_{w}\right)\left(\pi_{q(j)}+\gamma_{w}\right)\\
 & =\sum_{q\in\mathcal{M}_{w}}n_{q}^{Q}\left(n_{q}^{Q}-1\right)\left(\pi_{q}+\gamma_{w}\right)^{2}+\sum_{q\in\mathcal{M}_{w}}\sum_{q^{\prime}\in\mathcal{M}_{w},q^{\prime}\ne q}n_{q}^{Q}n_{q^{\prime}}^{Q}\left(\pi_{q}+\gamma_{w}\right)\left(\pi_{q^{\prime}}+\gamma_{w}\right).
\end{align*}

Since
\begin{align*}
 & \sum_{q\in\mathcal{M}_{w}}n_{q}^{Q}\left(\pi_{q}+\gamma_{w}\right)^{2}-\frac{1}{n_{w}^{W}-1}\sum_{q\in\mathcal{M}_{w}}n_{q}^{Q}\left(n_{q}^{Q}-1\right)\left(\pi_{q}+\gamma_{w}\right)^{2} =\sum_{q\in\mathcal{M}_{w}}n_{q}^{Q}\left(\frac{n_{w}^{W}-n_{q}^{Q}}{n_{w}^{W}-1}\right)\left(\pi_{q}+\gamma_{w}\right)^{2},
\end{align*}
and $n_{w}^{W}=\sum_{q\in\mathcal{M}_{w}}n_{q}^{Q}$, we obtain
\begin{align*}
\mu_{3} & =\sum_{w}\left(\sum_{q\in\mathcal{M}_{w}}n_{q}^{Q}\left(\frac{n_{w}^{W}-n_{q}^{Q}}{n_{w}^{W}-1}\right)\left(\pi_{q}+\gamma_{w}\right)^{2}-\frac{1}{n_{w}^{W}-1}\sum_{q\in\mathcal{M}_{w}}\sum_{q^{\prime}\in\mathcal{M}_{w},q^{\prime}\ne q}n_{q}^{Q}n_{q^{\prime}}^{Q}\left(\pi_{q}+\gamma_{w}\right)\left(\pi_{q^{\prime}}+\gamma_{w}\right)\right)\\
 & =\sum_{w}\frac{1}{n_{w}^{W}-1}\left(\sum_{q\in\mathcal{M}_{w}}\sum_{q^{\prime}\in\mathcal{M}_{w},q^{\prime}\ne q}n_{q}^{Q}n_{q^{\prime}}^{Q}\left(\pi_{q}+\gamma_{w}\right)^{2}-\sum_{q\in\mathcal{M}_{w}}\sum_{q^{\prime}\in\mathcal{M}_{w},q^{\prime}\ne q}n_{q}^{Q}n_{q^{\prime}}^{Q}\left(\pi_{q}+\gamma_{w}\right)\left(\pi_{q^{\prime}}+\gamma_{w}\right)\right)\\
 & =\sum_{w}\frac{1}{n_{w}^{W}-1}\left(\sum_{q\in\mathcal{M}_{w}}\sum_{q^{\prime}\in\mathcal{M}_{w},q^{\prime}\ne q}n_{q}^{Q}n_{q^{\prime}}^{Q}\left(\pi_{q}+\gamma_{w}\right)\left(\pi_{q}-\pi_{q^{\prime}}\right)\right).
\end{align*}

Then, observe that
\begin{align*}
 & \sum_{q\in\mathcal{M}_{w}}\sum_{q^{\prime}\in\mathcal{M}_{w},q^{\prime}\ne q}n_{q}^{Q}n_{q^{\prime}}^{Q}\left(\pi_{q}+\gamma\right)\left(\pi_{q}-\pi_{q^{\prime}}\right)\\
 & =\sum_{q\in\mathcal{M}_{w}}\sum_{q^{\prime}\in\mathcal{M}_{w},q^{\prime}<q}n_{q}^{Q}n_{q^{\prime}}^{Q}\left(\pi_{q}+\gamma\right)\left(\pi_{q}-\pi_{q^{\prime}}\right)+\sum_{q\in\mathcal{M}_{w}}\sum_{q^{\prime}\in\mathcal{M}_{w},q^{\prime}>q}n_{q}^{Q}n_{q^{\prime}}^{Q}\left(\pi_{q}+\gamma\right)\left(\pi_{q}-\pi_{q^{\prime}}\right)\\
 & =\sum_{q\in\mathcal{M}_{w}}\sum_{q^{\prime}\in\mathcal{M}_{w},q^{\prime}<q}n_{q}^{Q}n_{q^{\prime}}^{Q}\left(\pi_{q}+\gamma\right)\left(\pi_{q}-\pi_{q^{\prime}}\right)-\sum_{q\in\mathcal{M}_{w}}\sum_{q^{\prime}\in\mathcal{M}_{w},q^{\prime}<q}n_{q}^{Q}n_{q^{\prime}}^{Q}\left(\pi_{q^{\prime}}+\gamma\right)\left(\pi_{q}-\pi_{q^{\prime}}\right)\\
 & =\sum_{q\in\mathcal{M}_{w}}\sum_{q^{\prime}\in\mathcal{M}_{w},q^{\prime}<q}n_{q}^{Q}n_{q^{\prime}}^{Q}\left(\pi_{q}-\pi_{q^{\prime}}\right)\left(\pi_{q}+\gamma-\pi_{q^{\prime}}-\gamma\right) =\sum_{q\in\mathcal{M}_{w}}\sum_{q^{\prime}\in\mathcal{M}_{w},q^{\prime}<q}n_{q}^{Q}n_{q^{\prime}}^{Q}\left(\pi_{q}-\pi_{q^{\prime}}\right)^{2},
\end{align*}
where the second equality switches the indices of $q$ and $q^{\prime}$ in the second element. Hence,
\begin{align*}
\mu_{3} & =\sum_{w}\frac{1}{n_{w}^{W}-1}\left(\sum_{q\in\mathcal{M}_{w}}\sum_{q^{\prime}\in\mathcal{M}_{w},q^{\prime}<q}n_{q}^{Q}n_{q^{\prime}}^{Q}\left(\pi_{q}-\pi_{q^{\prime}}\right)^{2}\right)\geq 0.
\end{align*}
Analogously,
\begin{align*}
\mu_{2} & =\sum_{w}\left(\sum_{q\in\mathcal{M}_{w}}\frac{1}{n_{q}^{Q}-1}\sum_{i\in\mathcal{N}_{q}^{Q}}\sum_{j\in\mathcal{N}_{q}^{Q},j\ne i}R_{i}R_{Yj}-\frac{1}{n_{w}^{W}-1}\sum_{i\in\mathcal{N}_{w}^{W}}\sum_{j\in\mathcal{N}_{w}^{W},j\ne i}R_{i}R_{Yj}\right)\\
 & =\sum_{w}\left(\sum_{q\in\mathcal{M}_{w}}n_{q}^{Q}\left(\pi_{q}+\gamma_{w}\right)\left(\pi_{Yq}+\gamma_{w}\right)-\frac{1}{n_{w}^{W}-1}\sum_{q\in\mathcal{M}_{w}}n_{q}^{Q}\left(n_{q}^{Q}-1\right)\left(\pi_{q}+\gamma_{w}\right)\left(\pi_{Yq}+\gamma_{w}\right)\right)\\
 & \quad-\sum_{w}\frac{1}{n_{w}^{W}-1}\sum_{q\in\mathcal{M}_{w}}\sum_{q^{\prime}\in\mathcal{M}_{w},q^{\prime}\ne q}n_{q}^{Q}n_{q^{\prime}}^{Q}\left(\pi_{q}+\gamma_{w}\right)\left(\pi_{Yq^{\prime}}+\gamma_{w}\right)\\
 & =\sum_{w}\bigg(\sum_{q\in\mathcal{M}_{w}}n_{q}^{Q}\left(\frac{n_{w}^{W}-n_{q}^{Q}}{n_{w}^{W}-1}\right)\left(\pi_{q}+\gamma_{w}\right)\left(\pi_{Yq}+\gamma_{w}\right) \\
&\quad-\frac{1}{n_{w}^{W}-1}\sum_{q\in\mathcal{M}_{w}}\sum_{q^{\prime}\in\mathcal{M}_{w},q^{\prime}\ne q}n_{q}^{Q}n_{q^{\prime}}^{Q}\left(\pi_{q}+\gamma_{w}\right)\left(\pi_{Yq^{\prime}}+\gamma_{w}\right)\bigg)\\
 & =\sum_{w}\frac{1}{n_{w}^{W}-1}\left(\sum_{q\in\mathcal{M}_{w}}\sum_{q^{\prime}\in\mathcal{M}_{w},q^{\prime}\ne q}n_{q}^{Q}n_{q^{\prime}}^{Q}\left(\pi_{q}+\gamma_{w}\right)\left(\pi_{Yq}-\pi_{Yq^{\prime}}\right)\right)\\
 & =\sum_{w}\frac{1}{n_{w}^{W}-1}\left(\sum_{q\in\mathcal{M}_{w}}\sum_{q^{\prime}\in\mathcal{M}_{w},q^{\prime}<q}n_{q}^{Q}n_{q^{\prime}}^{Q}\left(\pi_{q}-\pi_{q^{\prime}}\right)\left(\pi_{Yq}-\pi_{Yq^{\prime}}\right)\right),
\end{align*}
and 
\[
\mu_{1}=\sum_{w}\frac{1}{n_{w}^{W}-1}\left(\sum_{q\in\mathcal{M}_{w}}\sum_{q^{\prime}\in\mathcal{M}_{w},q^{\prime}<q}n_{q}^{Q}n_{q^{\prime}}^{Q}\left(\pi_{Yq}-\pi_{Yq^{\prime}}\right)^{2}\right)\geq0.
\]

\end{proof}

\begin{proof}[Proof of \Cref{prop:lm_opt}]
Fix any alternative $\left(\pi^{A},\pi_{Y}^{A}\right)\in\mathcal{S}$ with a corresponding $\left(\mu_{1}^{A},\mu_{2}^{A},\mu_{3}^{A}\right)$.
Due to the restriction in $\mathcal{S}$, 

\[
\left(\begin{array}{c}
\mu_{1}^{H}\\
\mu_{2}^{H}\\
\mu_{3}^{H}
\end{array}\right)=\left(\begin{array}{c}
\mu_{1}^{A}-\frac{\sigma_{12}}{\sigma_{22}}\mu_{2}^{A}\\
0\\
\mu_{3}^{A}-\frac{\sigma_{23}}{\sigma_{22}}\mu_{2}^{A}
\end{array}\right)
\]
is in the null space. 
The Neyman-Pearson test for $\mu^{H}$ against $\mu^{A}$ rejects large values of:
\[
\log\frac{dN\left(\mu^{A},\Sigma\right)}{dN\left(\mu^{H},\Sigma\right)}=\frac{\mu_{2}^{A}}{\sigma_{22}}X_{2}-\frac{1}{2}\frac{\left(\mu_{2}^{A}\right)^{2}}{\sigma_{22}}.
\]

Hence, the most powerful test rejects large values of $X_{2}$, which
is what LM does.
By \citet{lehmann2005testing} Theorem 3.8.1(i), since LM is valid for any distribution in the null space (by \Cref{thm:normality}) and it is most powerful for some distribution in the null space, LM is most powerful for testing the composite null against the given alternative $\left(\pi^{A},\pi_{Y}^{A}\right)$.
\end{proof}

\begin{proof}[Proof of \Cref{prop:exist_SF}]
The first two are straightforward: $C_{S}=\mu_{3}/\left(c-1\right)$ and $\beta=\mu_{2}/\mu_{3}$ imply $\mu_{3}=\left(c-1\right)C_{S}$ and $\mu_{2}=\left(c-1\right)C_{S}\beta$. 
For $\mu_{1}$, observe that:
\begin{align*}
h & =\sqrt{\frac{1}{\sqrt{K}}\frac{1}{c-1}\left(\mu_{1}-\frac{\mu_{2}^{2}}{\mu_{3}}\right)}
=\sqrt{\frac{1}{\sqrt{K}}\left(\mu_{1}-C_{S}\beta^{2}\right)}, \text{ and }\\
C_{H}=\sqrt{K}h^{2} & =\mu_{1}/\left(c-1\right)-C_{S}\beta^{2}, \text{ so } \\
\left(c-1\right)\left(C_{S}\beta^{2}+C_{H}\right) & =\left(c-1\right)\left(C_{S}\beta^{2}+\mu_{1}/\left(c-1\right)-C_{S}\beta^{2}\right)=\mu_{1}
\end{align*}
as required. 
Next, since $\sigma_{vv}=\sqrt{\frac{\sigma_{33}c}{2\left(c-1\right)}}$, $\sigma_{33}=2\frac{c-1}{c}\sigma_{vv}^{2}$ is immediate. 
Similarly, with $\sigma_{\varepsilon v}=\frac{1}{\sigma_{vv}}\left(\frac{\sigma_{23}c}{2\left(c-1\right)}-\sigma_{vv}^{2}\beta\right)$, $\sigma_{23}=2\frac{c-1}{c}\sigma_{vv}\left(\sigma_{vv}\beta+\sigma_{\varepsilon v}\right)$.
From these two expressions, we can observe that:
\[
\left(\sigma_{vv}\beta+\sigma_{\varepsilon v}\right)^{2}=\frac{c}{2\left(c-1\right)}\frac{\sigma_{23}^{2}}{\sigma_{33}}.
\]

To obtain an expression for $\sigma_{22}$, rearrange $\sigma_{\varepsilon\varepsilon}=\frac{1}{\sigma_{vv}}\frac{c}{c-1}\left(\sigma_{22}-\frac{\sigma_{23}^{2}}{\sigma_{33}}\right)+\frac{\sigma_{\varepsilon v}^{2}}{\sigma_{vv}}\geq0$:
\begin{align*}
\sigma_{22} & =\frac{\sigma_{23}^{2}}{\sigma_{33}}+\frac{c-1}{c}\left(\sigma_{\varepsilon\varepsilon}\sigma_{vv}-\sigma_{\varepsilon v}^{2}\right)\\
 & =\frac{c-1}{c}\left(\sigma_{vv}\left(\sigma_{\varepsilon\varepsilon}+\sigma_{vv}\beta^{2}+\sigma_{vv}\sigma_{\xi\xi}+2\sigma_{\varepsilon v}\beta\right)+\left(\sigma_{vv}\beta+\sigma_{\varepsilon v}\right)^{2}\right)+o(1),
\end{align*}
where the final step uses $\sigma_{\xi\xi}=h/\sigma_{vv}$.
This expression for $\sigma_{22}$ is of the form required in \Cref{lem:RF_in_SF}. Then,
\begin{align*}
det\left(\Sigma_{SF}\right) & =\sigma_{\varepsilon\varepsilon}\sigma_{\xi\xi}\sigma_{vv}-\sigma_{\varepsilon\varepsilon}h^{2}-\sigma_{\varepsilon\xi}^{2}\sigma_{vv}+2\sigma_{\varepsilon\xi}\sigma_{\varepsilon v}h-\sigma_{\xi\xi}\sigma_{\varepsilon v}^{2}\\
 & =\sigma_{\varepsilon\varepsilon}\sigma_{\xi\xi}\sigma_{vv}-\sigma_{\varepsilon\varepsilon}h^{2}-\sigma_{\xi\xi}\sigma_{\varepsilon v}^{2}
 =\sigma_{\varepsilon\varepsilon}h-\sigma_{\varepsilon\varepsilon}h^{2}-h\frac{\sigma_{\varepsilon v}^{2}}{\sigma_{vv}}; \text{ and } \\
 det\left(\Sigma_{SF}\right)/h & =\sigma_{\varepsilon\varepsilon}-\frac{\sigma_{\varepsilon v}^{2}}{\sigma_{vv}}-\sigma_{\varepsilon\varepsilon}h
=\sigma_{\varepsilon\varepsilon}-\frac{\sigma_{\varepsilon v}^{2}}{\sigma_{vv}}+o(1).
\end{align*}

An analogous argument holds for $\sigma_{\xi vk} =-h$. From the $\sigma_{22}$ equation, $\sigma_{\varepsilon\varepsilon}-\frac{\sigma_{\varepsilon v}^{2}}{\sigma_{vv}}=\frac{c}{c-1}\left(\sigma_{22}-\frac{\sigma_{23}^{2}}{\sigma_{33}}\right)\geq0$, which delivers the result that $det\left(\Sigma_{SF}\right)/h \rightarrow C_{D}\geq0$.
\end{proof}

\begin{proof} [Proof of \Cref{lem:CI}]
The $A$ expressions can be written as:
\begin{align*}
A_{1} & =\sum_{i}\sum_{j\ne i}\sum_{k\ne i}\sum_{l\ne k}\check{M}_{il,-ijk}G_{ij}X_{j}G_{ik}X_{k}\left(Y_{i}Y_{l}-X_{i}Y_{l}\beta_{0}-Y_{i}X_{l}\beta_{0}+X_{i}X_{l}\beta_{0}^{2}\right); \\
A_{2} & =\sum_{i}\sum_{j\ne i}\sum_{k\ne i}\sum_{l\ne k}\check{M}_{il,-ijk}G_{ij}X_{j}G_{ki}X_{l}\left(Y_{i}Y_{k}-X_{i}Y_{k}\beta_{0}-Y_{i}X_{k}\beta_{0}+X_{i}X_{k}\beta_{0}^{2}\right); \\
A_{3} & =\sum_{i}\sum_{j\ne i}\sum_{k\ne i}\sum_{l\ne k}\check{M}_{il,-ijk}X_{l}G_{ji}G_{ki}X_{i}\left(Y_{j}Y_{k}-X_{j}Y_{k}\beta_{0}-Y_{j}X_{k}\beta_{0}+X_{j}X_{k}\beta_{0}^{2}\right); \\
A_{4} & =-\sum_{i}\sum_{j\ne i}\sum_{k\ne j}\sum_{l\ne i,k}\check{M}_{jl,-ijk}\check{M}_{ik,-ij}G_{ji}^{2}X_{i}X_{k}\left(Y_{j}Y_{l}-X_{j}Y_{l}\beta_{0}-Y_{j}X_{l}\beta_{0}+X_{j}X_{l}\beta_{0}^{2}\right);  \text{ and }\\
A_{5} & =-\sum_{i}\sum_{j\ne i}\sum_{k\ne j}\sum_{l\ne i,k}\check{M}_{ik,-ij}\check{M}_{jl,-ijk}G_{ij}G_{ji}X_{k}X_{l}\left(Y_{i}Y_{j}-X_{i}Y_{j}\beta_{0}-Y_{i}X_{j}\beta_{0}+X_{i}X_{j}\beta_{0}^{2}\right).
\end{align*}

Since these terms have a quadratic form, the variance estimator is also quadratic in $\beta_{0}$, i.e.,
\[
\hat{V}_{LM}=B_{0}+B_{1}\beta_{0}+B_{2}\beta_{0}^{2},
\]
where the $B$'s can be worked out by collecting the expressions above.
For instance,
\begin{align*}
B_{0} & =\sum_{i}\sum_{j\ne i}\sum_{k\ne i}\sum_{l\ne k}\check{M}_{il,-ijk}G_{ij}X_{j}G_{ik}X_{k}Y_{i}Y_{l}+2\sum_{i}\sum_{j\ne i}\sum_{k\ne i}\sum_{l\ne k}\check{M}_{il,-ijk}G_{ij}X_{j}G_{ki}X_{l}Y_{i}Y_{k}\\
 & \quad+\sum_{i}\sum_{j\ne i}\sum_{k\ne i}\sum_{l\ne k}\check{M}_{il,-ijk}X_{l}G_{ji}G_{ki}X_{i}Y_{j}Y_{k}\\
 & \quad-\sum_{i}\sum_{j\ne i}\sum_{k\ne j}\sum_{l\ne i,k}\check{M}_{jl,-ijk}\check{M}_{ik,-ij}G_{ji}^{2}X_{i}X_{k}Y_{j}Y_{l}-\sum_{i}\sum_{j\ne i}\sum_{k\ne j}\sum_{l\ne i,k}\check{M}_{ik,-ij}\check{M}_{jl,-ijk}G_{ij}G_{ji}X_{k}X_{l}Y_{i}Y_{j}
\end{align*}
$B_{1}$ and $B_{2}$ are analogous by collecting the coefficients on $\beta_{0},\beta_{0}^{2}$ from expressions $A_{1}$ to $A_{5}$.
The test does not reject:
\begin{align*}
\frac{\left(KT_{YX}-KT_{XX}\beta_{0}\right)^{2}}{B_{0}+B_{1}\beta_{0}+B_{2}\beta_{0}^{2}} & \leq q 
\Leftrightarrow \left(KT_{XX}^{2}-qB_{2}\right)\beta_{0}^{2}-\left(2KT_{YX}T_{XX}+qB_{1}\right)\beta_{0} +\left(KT_{YX}^{2}-qB_{0}\right)\leq0. 
\end{align*}

Solutions exist when:
\[
D:=\left(2KT_{YX}T_{XX}+qB_{1}\right)^{2}-4\left(KT_{XX}^{2}-qB_{2}\right)\left(KT_{YX}^{2}-qB_{0}\right)\geq0. 
\]
The rest of the lemma is immediate from properties of solving quadratic inequalities. 
\end{proof}

\subsubsection{Proofs for Lemmas in Appendix D}

\begin{proof}[Proof of \Cref{lem:max_inv}]
The joint distribution of $\left(Y^{\prime},X^{\prime}\right)^{\prime}$ is:
\[
\left[\begin{array}{c}
Y\\
X
\end{array}\right]\sim N\left(\left[\begin{array}{c}
Z\pi_{Y}\\
Z\pi
\end{array}\right],\left[\begin{array}{cc}
I_{n}\omega_{\zeta \zeta} & I_{n}\omega_{\zeta\eta}\\
I_{n}\omega_{\zeta\eta} & I_{n}\omega_{\eta \eta}
\end{array}\right]\right).
\]

Stack them together with their predicted values $PY=Z\left(Z^{\prime}Z\right)^{-1}Z^{\prime}Y$ and $PX=Z\left(Z^{\prime}Z\right)^{-1}Z^{\prime}X$: 
\[
\left[\begin{array}{c}
Y\\
X\\
P^{\prime}Y\\
P^{\prime}X
\end{array}\right]\sim N\left(\left[\begin{array}{c}
Z\pi_{Y}\\
Z\pi\\
Z\pi_{Y}\\
Z\pi
\end{array}\right],\left[\begin{array}{cccc}
I_{n}\omega_{\zeta \zeta} & I_{n}\omega_{\zeta\eta} & \omega_{\zeta \zeta}P^{\prime} & \omega_{\zeta\eta}P^{\prime}\\
I_{n}\omega_{\zeta\eta} & I_{n}\omega_{\eta \eta} & \omega_{\zeta\eta}P^{\prime} & \omega_{\eta \eta}P^{\prime}\\
\omega_{\zeta \zeta}P^{\prime} & \omega_{\zeta\eta}P^{\prime} & \omega_{\zeta \zeta}P^{\prime} & \omega_{\zeta\eta}P^{\prime}\\
\omega_{\zeta\eta}P^{\prime} & \omega_{\eta \eta}P^{\prime} & \omega_{\zeta\eta}P^{\prime} & \omega_{\eta \eta}P^{\prime}
\end{array}\right]\right).
\]

Then, the conditional normal distribution is:
\begin{align*}
\left[\begin{array}{c}
Y\\
X
\end{array}\right]|\left[\begin{array}{c}
Z\left(Z^{\prime}Z\right)^{-1}Z^{\prime}Y\\
Z\left(Z^{\prime}Z\right)^{-1}Z^{\prime}X
\end{array}\right] & \sim N\left(\left[\begin{array}{c}
Z\pi_Y\\
Z\pi
\end{array}\right]+\left[\begin{array}{c}
Z\left(Z^{\prime}Z\right)^{-1}Z^{\prime}Y-Z\pi_Y\\
Z\left(Z^{\prime}Z\right)^{-1}Z^{\prime}X-Z\pi
\end{array}\right],V\right)\\
 & =N\left(\left[\begin{array}{c}
Z\left(Z^{\prime}Z\right)^{-1}Z^{\prime}Y\\
Z\left(Z^{\prime}Z\right)^{-1}Z^{\prime}X
\end{array}\right],V\right)
=N\left(\left[\begin{array}{c}
PY\\
PX
\end{array}\right],V\right)
\end{align*}
Hence, $PX$ and $PY$ (i.e, $Z^{\prime}X$, $Z^{\prime}Y$) are sufficient
statistics for $\pi_{Y},\pi$.

To show that $\left(s_{1}^{\prime}s_{1},s_{1}^{\prime}s_{2},s_{2}^{\prime}s_{2}\right)$ is a maximal invariant, let $F$ be some conformable orthogonal matrix so $F^{\prime}F=I$. 
For invariance, let $s_{1}^{*}=Fs_{1}$. 
Then, $s_{1}^{*\prime}s_{1}^{*}=s_{1}^{\prime}F^{\prime}Fs_{1}=s_{1}^{\prime}s_{1}$.
Invariance of $\left(s_{1}^{\prime}s_{2},s_{2}^{\prime}s_{2}\right)$ is analogous. 
Maximality states that if $s_{1}^{*\prime}s_{1}^{*}=s_{1}^{\prime}s_{1}$, then $s_{1}^{*}=Fs_{1}$ for some $F$. 
Suppose not. 
This means $s_{1}^{*}=Gs_{1}$, and $G$ is not an orthogonal matrix but yet $s_{1}^{*\prime}s_{1}^{*}=s_{1}^{\prime}s_{1}$.
Since $G$ is not an orthogonal matrix, $G^{\prime}G\ne I$. 
Hence, $s_{1}^{*\prime}s_{1}^{*}=s_{1}^{\prime}G^{\prime}Gs_{1}\ne s_{1}^{\prime}s_{1}$,
a contradiction.
To obtain the distribution, 
\begin{align*}
\left[\begin{array}{c}
s_{1}\\
s_{2}
\end{array}\right] & =\left[\begin{array}{c}
\left(Z^{\prime}Z\right)^{-1/2}Z^{\prime}\left(Z\pi_{Y}+\zeta\right)\\
\left(Z^{\prime}Z\right)^{-1/2}Z^{\prime}\left(Z\pi+\eta\right)
\end{array}\right]
=\left[\begin{array}{c}
\left(Z^{\prime}Z\right)^{1/2}\pi_{Y}\\
\left(Z^{\prime}Z\right)^{1/2}\pi
\end{array}\right]+\left[\begin{array}{c}
\left(Z^{\prime}Z\right)^{-1/2}Z^{\prime}\zeta\\
\left(Z^{\prime}Z\right)^{-1/2}Z^{\prime}\eta
\end{array}\right].
\end{align*}

Since $\Var\left(\left(Z^{\prime}Z\right)^{-1/2}Z^{\prime}\eta\right)=\left(Z^{\prime}Z\right)^{-1/2}Z^{\prime}\omega_{\eta \eta}Z\left(Z^{\prime}Z\right)^{-1/2}=I_{K}\omega_{\eta \eta}$,
\[
\left[\begin{array}{c}
s_{1}\\
s_{2}
\end{array}\right]\sim N\left(\left(\begin{array}{c}
\left(Z^{\prime}Z\right)^{1/2}\pi_{Y}\\
\left(Z^{\prime}Z\right)^{1/2}\pi
\end{array}\right),\Omega\otimes I_{K}\right).
\]
\end{proof}

\begin{proof}[Proof of \Cref{lem:RF_in_SF}]
I work out the $\mu$'s first. 
Using the judge structure, $\sum_{i}M_{ii}^{2}=\sum_{k}\frac{\left(c-1\right)^{2}}{c},\sum_{i}\sum_{j\ne i}P_{ij}=\sum_{k}\frac{c-1}{c}$.
We have also chosen $\pi_{k},\sigma_{\xi vk}$ such that $\sum_{k}\pi_{k}=0,\sum_{k}\sigma_{\xi vk}=0,\sum_{k}\pi_{k}\sigma_{\xi vk}=0$.
Then, we get the result for means:
\begin{align*}
\left(\begin{array}{c}
\mu_{1}\\
\mu_{2}\\
\mu_{3}
\end{array}\right) & =\left(\begin{array}{c}
\frac{1}{\sqrt{K}}\sum_{k}\left(c-1\right)\left(\pi_{k}^{2}\beta^{2}+2\pi_{k}\beta\sigma_{\xi vk}+\sigma_{\xi vk}^{2}\right)\\
\frac{1}{\sqrt{K}}\sum_{k}\left(c-1\right)\left(\pi_{k}^{2}\beta+\pi_{k}\sigma_{\xi vk}\right)\\
\frac{1}{\sqrt{K}}\sum_{k}\left(c-1\right)\pi_{k}^{2}
\end{array}\right)
=\left(\begin{array}{c}
\sqrt{K}\left(c-1\right)\left(s^{2}\beta^{2}+h^{2}\right)\\
\sqrt{K}\left(c-1\right)s^{2}\beta\\
\sqrt{K}\left(c-1\right)s^{2}
\end{array}\right).
\end{align*}

Using a derivation similar to that of the lemma for $V_{LM}$ expression,
\begin{align*}
K\sigma_{22} & =\sum_{i}\sum_{j\ne i}\sum_{k\ne i}\left(G_{ji}G_{ki}E\left[\zeta_{i}^{2}\right]R_{j}R_{k}+2G_{ij}G_{ki}E\left[\eta_{i}\zeta_{i}\right]R_{Yj}R_{k}+G_{ij}G_{ik}E\left[\eta_{i}^{2}\right]R_{Yj}R_{Yk}\right)\\
 & \quad+\sum_{i}\sum_{j\ne i}\left(G_{ij}^{2}E\left[\eta_{i}^{2}\right]E\left[\zeta_{j}^{2}\right]+G_{ij}G_{ji}E\left[\eta_{i}\zeta_{i}\right]E\left[\eta_{j}\zeta_{j}\right]\right); \\
K\sigma_{11} & =\sum_{i}\sum_{j\ne i}\sum_{k\ne i}E\left[\zeta_{i}^{2}\right]R_{Yj}R_{Yk}\left(G_{ji}G_{ki}+2G_{ij}G_{ki}+G_{ij}G_{ik}\right)+\sum_{i}\sum_{j\ne i}E\left[\zeta_{i}^{2}\right]E\left[\zeta_{j}^{2}\right]\left(G_{ij}^{2}+G_{ij}G_{ji}\right); \\ 
K\sigma_{33}&=\sum_{i}\sum_{j\ne i}\sum_{k\ne i}E\left[\eta_{i}^{2}\right]R_{j}R_{k}\left(G_{ji}G_{ki}+2G_{ij}G_{ki}+G_{ij}G_{ik}\right)+\sum_{i}\sum_{j\ne i}E\left[\eta_{i}^{2}\right]E\left[\eta_{j}^{2}\right]\left(G_{ij}^{2}+G_{ij}G_{ji}\right); \\
K\sigma_{12} & =\sum_{i}\sum_{j\ne i}\sum_{k\ne i}\left(G_{ji}G_{ki}E\left[\zeta_{i}^{2}\right]R_{j}R_{Yk}+2G_{ij}G_{ki}E\left[\zeta_{i}^{2}\right]R_{Yj}R_{k}+G_{ij}G_{ik}E\left[\eta_{i}\zeta_{i}\right]R_{Yj}R_{Yk}\right)\\
 & \quad+\sum_{i}\sum_{j\ne i}E\left[\eta_{i}\zeta_{i}\right]E\left[\zeta_{j}^{2}\right]\left(G_{ij}^{2}+G_{ij}G_{ji}\right); \\
K\sigma_{23} & =\sum_{i}\sum_{j\ne i}\sum_{k\ne i}\left(G_{ji}G_{ki}E\left[\eta_{i}^{2}\right]R_{Yj}R_{k}+2G_{ij}G_{ki}E\left[\eta_{i}^{2}\right]R_{j}R_{Yk}+G_{ij}G_{ik}E\left[\eta_{i}\zeta_{i}\right]R_{j}R_{k}\right)\\
 & \quad+\sum_{i}\sum_{j\ne i}E\left[\eta_{i}\zeta_{i}\right]E\left[\eta_{j}^{2}\right]\left(G_{ij}^{2}+G_{ij}G_{ji}\right); \text{ and }\\
K\sigma_{13}&=\sum_{i}\sum_{j\ne i}\sum_{k\ne i}E\left[\eta_{i}\zeta_{i}\right]R_{Yj}R_{k}\left(G_{ji}G_{ki}+2G_{ij}G_{ki}+G_{ij}G_{ik}\right)+\sum_{i}\sum_{j\ne i}E\left[\eta_{i}\zeta_{i}\right]E\left[\eta_{j}\zeta_{j}\right]\left(G_{ij}^{2}+G_{ij}G_{ji}\right).
\end{align*}

The equalities hold regardless of whether identification is strong or weak and whether heterogeneity converges or not.
Without covariates, $G=P$ is symmetric and the above expressions simplify. For instance,
\begin{align*}
K\sigma_{22} & =\sum_{k}\frac{\left(c-1\right)^{2}}{c}\left(\omega_{\zeta\zeta k}\pi_{k}^{2}+2\omega_{\zeta\eta k}\pi_{k}\pi_{Yk}+\omega_{\eta\eta k}\pi_{Yk}^{2}\right)+\sum_{k}\frac{c-1}{c}\left(\omega_{\eta\eta k}\omega_{\zeta\zeta k}+\omega_{\zeta\eta k}^{2}\right).
\end{align*}

Evaluate the terms in the expression. 
For higher moments of $\pi_{k}$, $\sum_{k}\pi_{k}^{2}=Ks^{2}$, $\sum_{k}\pi_{k}^{3}=0$, and $\sum_{k}\pi_{k}^{4}=Ks^{4}$.
Similarly, $\sum_{k}\pi_{k}^{3}\sigma_{\xi v}=0$. 
Treating the heterogeneity in the same way, $\sum_{k}\sigma_{\xi v}^{2}=Kh^{2}$. Then,
\begin{align*}
\sum_{k}\omega_{\zeta\zeta k}\pi_{k}^{2} & =\sum_{k}\left(\pi_{k}^{2}\sigma_{\xi\xi}+2\pi_{k}\beta\sigma_{\xi vk}+2\pi_{k}\sigma_{\varepsilon\xi}+\sigma_{\varepsilon\varepsilon}-\sigma_{\xi vk}^{2}+\sigma_{vv}\beta^{2}+\sigma_{vv}\sigma_{\xi\xi}+2\sigma_{\xi vk}^{2}+2\sigma_{\varepsilon v}\beta\right)\pi_{k}^{2}\\
 & =s^{2}K\left(s^{2}\sigma_{\xi\xi}+\sigma_{\varepsilon\varepsilon}+\sigma_{vv}\beta^{2}+\sigma_{vv}\sigma_{\xi\xi}+h^{2}+2\sigma_{\varepsilon v}\beta\right); \text{ and } \\
\sum_{k}\omega_{\zeta\eta k}\pi_{k}\pi_{Yk} & =\sum_{k}\left(\pi_{k}\sigma_{\xi vk}+\sigma_{vv}\beta+\sigma_{\varepsilon v}\right)\pi_{k}\left(\pi_{k}\beta+\sigma_{\xi vk}\right)\\
 & =\sum_{k}\left(\sigma_{vv}\beta^{2}\pi_{k}^{2}+\sigma_{\varepsilon v}\pi_{k}^{2}\beta+\pi_{k}^{2}\sigma_{\xi vk}^{2}\right) =s^{2}K\left(\sigma_{vv}\beta^{2}+\sigma_{\varepsilon v}\beta+h^{2}\right).
\end{align*}

Now, for the $P_{ij}^{2}$ part,
\begin{align*}
\sum_{k}\omega_{\eta\eta k}\omega_{\zeta\zeta k} & =\sum_{k}\sigma_{vv}\left(\pi_{k}^{2}\sigma_{\xi\xi}+2\pi_{k}\beta\sigma_{\xi vk}+2\pi_{k}\sigma_{\varepsilon\xi}+\sigma_{\varepsilon\varepsilon}-\sigma_{\xi vk}^{2}+\sigma_{vv}\beta^{2}+\sigma_{vv}\sigma_{\xi\xi}+2\sigma_{\xi vk}^{2}+2\sigma_{\varepsilon v}\beta\right)\\
 & =\sum_{k}\sigma_{vv}\left(\pi_{k}^{2}\sigma_{\xi\xi}+\sigma_{\varepsilon\varepsilon}-\sigma_{\xi vk}^{2}+\sigma_{vv}\beta^{2}+\sigma_{vv}\sigma_{\xi\xi}+2\sigma_{\xi vk}^{2}+2\sigma_{\varepsilon v}\beta\right)\\
 & =K\sigma_{vv}\left(s^{2}\sigma_{\xi\xi}+\sigma_{\varepsilon\varepsilon}+\sigma_{vv}\beta^{2}+\sigma_{vv}\sigma_{\xi\xi}+h^{2}+2\sigma_{\varepsilon v}\beta\right); \text{ and } \\
\sum_{k}\omega_{\zeta\eta k}^{2} & =\sum_{k}\left(\pi_{k}\sigma_{\xi vk}\pi_{k}\sigma_{\xi vk}+\sigma_{vv}\beta\pi_{k}\sigma_{\xi vk}+\sigma_{\varepsilon v}\pi_{k}\sigma_{\xi vk}+\pi_{k}\sigma_{\xi vk}\sigma_{vv}\beta+\sigma_{vv}\beta\sigma_{vv}\beta+\sigma_{\varepsilon v}\sigma_{vv}\beta\right)\\
 & \quad+\sum_{k}\left(\pi_{k}\sigma_{\xi vk}\sigma_{\varepsilon v}+\sigma_{vv}\beta\sigma_{\varepsilon v}+\sigma_{\varepsilon v}^{2}\right)\\
 & =\sum_{k}\left(\pi_{k}^{2}\sigma_{\xi vk}^{2}+\sigma_{vv}^{2}\beta^{2}+\sigma_{\varepsilon v}\sigma_{vv}\beta+\sigma_{vv}\beta\sigma_{\varepsilon v}+\sigma_{\varepsilon v}^{2}\right) 
 =K\left(s^{2}h^{2}+\left(\sigma_{vv}\beta+\sigma_{\varepsilon v}\right)^{2}\right).
\end{align*}

Combine the expressions for $\sigma_{22}$ and impose asymptotics where $s\rightarrow0$ and $h\rightarrow0$:
\begin{align*}
\sigma_{22} & =\frac{1}{K}\sum_{k}\frac{\left(c-1\right)^{2}}{c}h^{2}+\frac{1}{K}\sum_{k}\frac{c-1}{c}\left(\sigma_{vv}\left(\sigma_{\varepsilon\varepsilon}+\sigma_{vv}\beta^{2}+\sigma_{vv}\sigma_{\xi\xi}+h^{2}+2\sigma_{\varepsilon v}\beta\right)+\left(\sigma_{vv}\beta+\sigma_{\varepsilon v}\right)^{2}\right)+o(1)\\
 & =\frac{c-1}{c}\left(\sigma_{vv}\left(\sigma_{\varepsilon\varepsilon}+\sigma_{vv}\beta^{2}+\sigma_{vv}\sigma_{\xi\xi}+2\sigma_{\varepsilon v}\beta\right)+\left(\sigma_{vv}\beta+\sigma_{\varepsilon v}\right)^{2}\right)+o(1). 
\end{align*}

Next, evaluate a few more sums that feature in the other $\sigma$ expressions: 
\begin{align*}
\sum_{k}\omega_{\zeta\zeta}\pi_{Yk}^{2} & =\sum_{k}\left(\pi_{k}^{2}\sigma_{\xi\xi}+2\pi_{k}\beta\sigma_{\xi vk}+2\pi_{k}\sigma_{\varepsilon\xi}+\sigma_{\varepsilon\varepsilon}+\sigma_{vv}\beta^{2}+\sigma_{vv}\sigma_{\xi\xi}+\sigma_{\xi vk}^{2}+2\sigma_{\varepsilon v}\beta\right) \\
&\quad\left(\pi_{k}^{2}\beta^{2}+2\pi_{k}\sigma_{\xi vk}+\sigma_{\xi v}^{2}\right)\\
\frac{1}{K}\sum_{k}\omega_{\zeta\zeta}\pi_{Yk}^{2} & =\frac{1}{K}\sum_{k}\sigma_{\xi v}^{2}\left(\pi_{k}^{2}\sigma_{\xi\xi}+2\pi_{k}\beta\sigma_{\xi vk}+2\pi_{k}\sigma_{\varepsilon\xi}+\sigma_{\varepsilon\varepsilon}+\sigma_{vv}\beta^{2}+\sigma_{vv}\sigma_{\xi\xi}+\sigma_{\xi vk}^{2}+2\sigma_{\varepsilon v}\beta\right)\\
 & =h^{2}\left(\sigma_{\varepsilon\varepsilon}+\sigma_{vv}\beta^{2}+\sigma_{vv}\sigma_{\xi\xi}+h^{2}+2\sigma_{\varepsilon v}\beta\right)=o(1); \\
\frac{1}{K}\sum_{k}\omega_{\zeta\zeta}^{2} & =\frac{1}{K}\sum_{k}\left(\pi_{k}^{2}\sigma_{\xi\xi}+2\pi_{k}\beta\sigma_{\xi vk}+2\pi_{k}\sigma_{\varepsilon\xi}+\sigma_{\varepsilon\varepsilon}+\sigma_{vv}\beta^{2}+\sigma_{vv}\sigma_{\xi\xi}+\sigma_{\xi vk}^{2}+2\sigma_{\varepsilon v}\beta\right)^{2}\\
 & =\frac{1}{K}\sum_{k}\left(\sigma_{\varepsilon\varepsilon}+\sigma_{vv}\beta^{2}+\sigma_{vv}\sigma_{\xi\xi}+\sigma_{\xi vk}^{2}+2\sigma_{\varepsilon v}\beta\right)^{2}
 =\left(\sigma_{\varepsilon\varepsilon}+\sigma_{vv}\beta^{2}+\sigma_{vv}\sigma_{\xi\xi}+2\sigma_{\varepsilon v}\beta\right)^{2}; \\
\frac{1}{K}\sum_{k}\omega_{\zeta\eta}\pi_{Yk}^{2} & =\frac{1}{K}\sum_{k}\left(\pi_{k}\sigma_{\xi vk}+\sigma_{vv}\beta+\sigma_{\varepsilon v}\right)\left(\pi_{k}^{2}\beta^{2}+2\pi_{k}\sigma_{\xi vk}+\sigma_{\xi v}^{2}\right)\\
 & =h^{2}\left(\sigma_{vv}\beta+\sigma_{\varepsilon v}\right)=o(1); \text{ and } \\
\frac{1}{K}\sum_{k}\omega_{\zeta\eta}\omega_{\zeta\zeta} & =\frac{1}{K}\sum_{k}\left(\pi_{k}\sigma_{\xi vk}+\sigma_{vv}\beta+\sigma_{\varepsilon v}\right) \\
&\quad\left(\pi_{k}^{2}\sigma_{\xi\xi}+2\pi_{k}\beta\sigma_{\xi vk}+2\pi_{k}\sigma_{\varepsilon\xi}+\sigma_{\varepsilon\varepsilon}+\sigma_{vv}\beta^{2}+\sigma_{vv}\sigma_{\xi\xi}+\sigma_{\xi vk}^{2}+2\sigma_{\varepsilon v}\beta\right)\\
 & =\frac{1}{K}\sum_{k}\left(\sigma_{vv}\beta+\sigma_{\varepsilon v}\right)\left(\sigma_{\varepsilon\varepsilon}+\sigma_{vv}\beta^{2}+\sigma_{vv}\sigma_{\xi\xi}+\sigma_{\xi vk}^{2}+2\sigma_{\varepsilon v}\beta\right)\\
 & =\left(\sigma_{vv}\beta+\sigma_{\varepsilon v}\right) \left(\sigma_{\varepsilon\varepsilon}+\sigma_{vv}\beta^{2}+\sigma_{vv}\sigma_{\xi\xi}+2\sigma_{\varepsilon v}\beta\right)+o(1).
\end{align*}

Using these results,
\begin{align*}
\sigma_{22} & =\frac{c-1}{c}\left(\sigma_{vv}\left(\sigma_{\varepsilon\varepsilon}+\sigma_{vv}\beta^{2}+\sigma_{vv}\sigma_{\xi\xi}+2\sigma_{\varepsilon v}\beta\right)+\left(\sigma_{vv}\beta+\sigma_{\varepsilon v}\right)^{2}\right)+o(1); \\
\sigma_{11} & =2\frac{c-1}{c}\left(\sigma_{\varepsilon\varepsilon}+\sigma_{vv}\beta^{2}+\sigma_{vv}\sigma_{\xi\xi}+2\sigma_{\varepsilon v}\beta\right)^{2}+o(1); \\
\sigma_{33} &=2\frac{c-1}{c}\sigma_{vv}^{2}+o(1); \\
\sigma_{12} & =2\frac{c-1}{c}\left(\sigma_{vv}\beta+\sigma_{\varepsilon v}\right)\left(\sigma_{\varepsilon\varepsilon}+\sigma_{vv}\beta^{2}+\sigma_{vv}\sigma_{\xi\xi}+2\sigma_{\varepsilon v}\beta\right)+o(1); \\
\sigma_{23}&=2\frac{c-1}{c}\sigma_{vv}\left(\sigma_{vv}\beta+\sigma_{\varepsilon v}\right)+o(1); \text{ and } \\
\sigma_{13}&=2\frac{c-1}{c}\left(\sigma_{vv}\beta+\sigma_{\varepsilon v}\right)^{2}+o(1).
\end{align*}

Hence, $\sigma_{13}=\sigma_{23}^{2}/\sigma_{33}+o(1)$ is immediate.
Further, for $\sigma_{12}$,
\begin{align*}
2\frac{\sigma_{23}}{\sigma_{33}}&\left(\sigma_{22}-\frac{\sigma_{23}^{2}}{2\sigma_{33}}\right) =2\frac{\sigma_{vv}\beta+\sigma_{\varepsilon v}}{\sigma_{vv}}\left(\sigma_{22}-\frac{\left(2\frac{c-1}{c}\sigma_{vv}\left(\sigma_{vv}\beta+\sigma_{\varepsilon v}\right)\right)^{2}}{2\times2\frac{c-1}{c}\sigma_{vv}^{2}}\right)+o(1)\\
 & =2\frac{c-1}{c}\left(\sigma_{vv}\beta+\sigma_{\varepsilon v}\right)\left(\sigma_{\varepsilon\varepsilon}+\sigma_{vv}\beta^{2}+\sigma_{vv}\sigma_{\xi\xi}+2\sigma_{\varepsilon v}\beta\right)+o(1)
 =\sigma_{12}+o(1).
\end{align*}

Finally, the $\sigma_{11}$ can be obtained:
\begin{align*}
\frac{4}{\sigma_{33}}\left(\sigma_{22}-\frac{\sigma_{23}^{2}}{2\sigma_{33}}\right)^{2} & =\frac{2}{\frac{c-1}{c}\sigma_{vv}^{2}}\left(\frac{c-1}{c}\left(\sigma_{vv}\left(\sigma_{\varepsilon\varepsilon}+\sigma_{vv}\beta^{2}+\sigma_{vv}\sigma_{\xi\xi}+h^{2}+2\sigma_{\varepsilon v}\beta\right)\right)\right)^{2}+o(1)\\
 & =2\frac{c-1}{c}\left(\sigma_{\varepsilon\varepsilon}+\sigma_{vv}\beta^{2}+\sigma_{vv}\sigma_{\xi\xi}+2\sigma_{\varepsilon v}\beta\right)^{2}+o(1)=\sigma_{11}+o(1).
\end{align*}
\end{proof}

\subsubsection{Derivations for Simulations}
\textbf{Derivation for continuous setup without covariates}. 

This subsection derives expressions for objects in the reduced-form model.
Comparing the first-stage equations, $\eta_i = v_i$. 
As a corollary, for all $i$, $E\left[\eta_{i}^{2}\right]=\sigma_{vv}$. 
Then, $\zeta_{i}=Z_{i}^{\prime}\left(\pi\beta_{i}-\pi_{Y}\right)+v_{i}\beta_{i}+\varepsilon_{i}$. 
Define $\pi_{Y}$ using $E\left[\zeta_{i}\right]=0$ and $E\left[v_{i}\beta_{i}\right]=E\left[v_{i}\left(\beta+\xi_{i}\right)\right]=\sigma_{\xi v k(i)}$, 
which implies $\pi_{Yk}=\pi_{k}\beta+\sigma_{\xi vk}$. 
Hence, we can rewrite $\zeta_{i}$ as:
\begin{align*}
\zeta_{i} & =\pi_{k(i)}\xi_{i}-\sigma_{\xi v k(i)}+v_{i}\beta+v_{i}\xi_{i}+\varepsilon_{i}.
\end{align*}

By substituting the expression for $\zeta_i$, the covariance is $E\left[\eta_{i}\zeta_{i} \mid k\right] =\pi_{k}\sigma_{\xi vk} +\sigma_{vv}\beta+E\left[v_{i}^{2}\xi_{i}\right]+\sigma_{\varepsilon v}$.
By Isserlis' theorem, $E\left[v_{i}^{2}\xi_{i}\right]=0$, so $E\left[\eta_{i}\zeta_{i}\mid k\right]=\pi_{k}\sigma_{\xi vk}+\sigma_{vv}\beta+\sigma_{\varepsilon v}$.
The variance of $\zeta_i$ can be derived analogously. 
Since $E\left[v_{i}^{2}\beta_{i}^{2}\right]= \sigma_{vv}\beta^{2}+\sigma_{vv}\sigma_{\xi\xi}+2\sigma_{\xi vk}^{2}$ by applying Isserlis' theorem, with $\omega_{\eta\eta k} := E[\eta_i^2 \mid k(i) =k]$, $\omega_{\zeta\eta k} := E[\zeta_i\eta_i \mid k(i) =k]$, and $\omega_{\zeta\zeta k} := E[\zeta_i^2 \mid k(i) =k]$, we obtain:
\begin{equation} \label{eqn:rf_var_expr}
\begin{split}
\omega_{\eta\eta k} & =\sigma_{vv}^{2}, \\
\omega_{\zeta\eta k} & =\pi_{k}\sigma_{\xi vk}+\sigma_{vv}\beta+\sigma_{\varepsilon v}, \text{ and }\\
\omega_{\zeta\zeta k} & =\pi_{k}^{2}\sigma_{\xi\xi}+2\pi_{k}\beta\sigma_{\xi vk}+2\pi_{k}\sigma_{\varepsilon\xi}+\sigma_{\varepsilon\varepsilon} +\sigma_{\xi vk}^{2}+\sigma_{vv}\beta^{2}+\sigma_{vv}\sigma_{\xi\xi}+2\sigma_{\varepsilon v}\beta.
\end{split}
\end{equation}

In this model, the local average treatment effect (LATE) of judge $k$ relative to the base judge $0$ is:
\begin{equation}
LATE_k = \frac{\pi_{Yk}}{\pi_k} = \beta + \frac{\sigma_{\xi vk}}{\pi_k}.
\end{equation}

\textbf{Derivation for binary setup without covariates}. 

The reduced-form residuals are given by:
\begin{align*}
\eta_{i}\mid v_{i}=\begin{cases}
\begin{array}{c}
1-\pi_{k}\\
-\pi_{k}
\end{array} & \begin{array}{c}
if\\
if
\end{array}\begin{array}{c}
v_{i}\leq\pi_{k}\\
v_{i}>\pi_{k}
\end{array}\end{cases}, \text{ and } \quad 
\zeta_{i} =\pi_{k(i)}\beta_{i}-\pi_{Yk(i)}+\eta_{i}\beta_{i}+\varepsilon_{i}.
\end{align*}

Imposing $E\left[\zeta_{i}\right]=0$, $\pi_{Yk(i)} =\pi_{k(i)}\beta+E\left[\eta_{i}\beta_{i}\right]$, where $E\left[\eta_{i}\beta_{i}\right]=-\left(1-s\right)\left(2p-1\right)\sigma_{\xi vk}$. 
Hence,
\[
\pi_{Yk}=\pi_{k}\beta-\left(1-s\right)\left(2p-1\right)\sigma_{\xi vk}.
\]

Due to the judge setup, the estimand is:
\begin{align*}
\frac{\sum_{k}\pi_{Yk}\pi_{k}}{\sum_{k}\pi_{k}^2} & =\frac{\sum_{k}\left(\pi_{k}\beta-\left(1-s\right)\left(2p-1\right)\sigma_{\xi vk}\right)\pi_{k}}{\sum_{k}\pi_{k}^2}=\beta
\end{align*}
because $\sum_{k}\sigma_{\xi vk}\pi_{k}=0$ by construction.

\textbf{Derivation for binary setup with covariates}. 

Consider the structural model:
\begin{align*}
Y_{i}(x) & =x(\beta + \xi_i)+w^\prime \gamma+ \varepsilon_{i}, \text{ and }\\
X_{i}(z) & =I\left\{ z^{\prime}\pi + w^\prime \gamma  -v_{i}\geq0\right\}. 
\end{align*}
Let $\mathcal{N}_t$ denote the set of observations in state $t$.
Then, using the $G$ that corresponds to UJIVE,
\begin{align*}
\sum_{i\in\mathcal{N}_{t}}&\sum_{j\in\mathcal{N}_{t}\backslash{i}}G_{ij}R_{Yi}R_{j}  =\sum_{i\in\mathcal{N}_{t}}\sum_{j\in\mathcal{N}_{t}\backslash{i}}G_{ij}\left(\pi_{Yk(i)}+\gamma_{t(i)}\right)\left(\pi_{k(j)}+\gamma_{t(j)}\right)\\
 & =\sum_{i\in\mathcal{N}_{t}}\sum_{j\in\mathcal{N}_{t}\backslash{i}}G_{ij}\left(\pi_{Yk(i)}\pi_{k(j)}+\gamma_{t(i)}\pi_{k(j)}+\pi_{Yk(i)}\gamma_{t(j)}+\gamma_{t(i)}\gamma_{t(j)}\right)\\
 & =\frac{1}{1-1/5}\sum_{k\in\{0,t\}}5\times4\times\frac{1}{5}\left(\pi_{Yk}\pi_{k}+\gamma_{t}\pi_{k}+\pi_{Yk}\gamma_{t}+\gamma_{t}^2\right)\\
 & \quad-\frac{1}{1-1/10}\sum_{i\in\mathcal{N}_{t}}\sum_{j\in\mathcal{N}_{t}\backslash{i}}\frac{1}{10}\left(\pi_{Yk(i)}\pi_{k(j)}+\gamma_{t}\pi_{k(j)}+\pi_{Yk(i)}\gamma_{t}+\gamma_{t}^2\right)\\
 & =\sum_{k\in\{0,t\}} 5\left(\pi_{Yk}\pi_{k}+\gamma_{t}\pi_{k}+\pi_{Yk}\gamma_{t}+\gamma_{t}^2\right)
 -\frac{1}{9}\sum_{k\in\{0,t\}} 5\times4\left(\pi_{Yk}\pi_{k}+\gamma_{Yt}\pi_{k}+\pi_{Yk}\gamma_{Xt}+\gamma_{t}^2\right)\\
 & \quad-\frac{1}{9}5\times5\left(\pi_{Yt}\pi_{0}+\gamma_{t}\pi_{0}+\pi_{Yt}\gamma_{t}+\gamma_{t}^2\right)
 -\frac{1}{9}5\times5\left(\pi_{Y0}\pi_{t}+\gamma_{t}\pi_{t}+\pi_{Y0}\gamma_{t}+\gamma_{t}^2\right)\\
 & =5\left(\frac{5}{9}\right)\left(\pi_{Y0}\pi_{0}+\pi_{Yt}\pi_{t}-\pi_{Yt}\pi_{0}-\pi_{Y0}\pi_{t}\right).
\end{align*}

Using the result that $\pi_{Yk}=\pi_{k}\beta-\left(1-s\right)\left(2p-1\right)\sigma_{\xi vk}$,
\begin{align*}
\sum_{i\in\mathcal{N}_{t}}\sum_{j\in\mathcal{N}_{t}\backslash{i}}G_{ij}R_{Yi}R_{j} & =5\left(\frac{5}{9}\right)\left(\pi_{Y0}\pi_{0}+\pi_{Yt}\pi_{t}-\pi_{Yt}\pi_{0}-\pi_{Y0}\pi_{t}\right)
=\frac{25}{9}\pi_{Yt}\pi_{t}.
\end{align*}

Analogously, $\sum_{i\in\mathcal{N}_{t}}\sum_{j\in\mathcal{N}_{t}\backslash{i}}G_{ij}R_{i}R_{j}=\frac{25}{9}\pi_{t}^{2}$.
Hence, as long as $\sum_{t}\sigma_{\xi vt}\pi_{t}=0$, which is the case for the construction in the main text, we still recover $\beta$ as our estimand:
\begin{align*}
\frac{\sum_{i}\sum_{j\ne i}G_{ij}R_{Yi}R_{j}}{\sum_{i}\sum_{j\ne i}G_{ij}R_{i}R_{j}} & =\frac{\sum_{t}\pi_{Yt}\pi_{t}}{\sum_{t}\pi_{t}^{2}}
=\frac{\sum_{t}\left(\pi_{t}\beta-\left(1-s\right)\left(2p-1\right)\sigma_{\xi vt}\right)\pi_{t}}{\sum_{t}\pi_{t}^{2}}\\
 & =\beta-\frac{\sum_{t} \left(1-s\right)\left(2p-1\right)\sigma_{\xi vt} \pi_{t}}{\sum_{t}\pi_{t}^{2}}=\beta,
\end{align*}
regardless of $\gamma_{t}$. 

\subsubsection{Derivations for Variance Estimands}

\begin{proof}[Proof of \Cref{lem:VMO_expr}]
\footnotesize
\begin{align*}
E&\left[\hat{\Psi}_{MO}\right] =E\left[\sum_{i}\left(\sum_{j\ne i}P_{ij}\left(R_{j}+\eta_{j}\right)\right)^{2}\left(R_{\Delta i}+\nu_{i}\right)^{2}+\sum_{i}\sum_{j\ne i}P_{ij}^{2}\left(R_{i}+\eta_{i}\right)\left(R_{\Delta i}+\nu_{i}\right)\left(R_{j}+\eta_{j}\right)\left(R_{\Delta j}+\nu_{j}\right)\right]\\
 & =E\left[\sum_{i}\left(\left(\sum_{j\ne i}P_{ij}R_{j}\right)^{2}+\left(\sum_{j\ne i}P_{ij}\eta_{j}\right)^{2}\right)\left(R_{\Delta i}+\nu_{i}\right)^{2}\right]\\
 & \quad+E\left[\sum_{i}\sum_{j\ne i}P_{ij}^{2}\left(R_{i}R_{\Delta i}+\eta_{i}R_{\Delta i}+R_{i}\nu_{i}+\eta_{i}\nu_{i}\right)\left(R_{j}R_{\Delta j}+\eta_{j}R_{\Delta j}+R_{j}\nu_{j}+\eta_{j}\nu_{j}\right)\right]\\
 & =\sum_{i}M_{ii}^{2}R_{i}^{2}\left(R_{\Delta i}^{2}+E\left[\nu_{i}^{2}\right]\right)+\sum_{i}R_{\Delta i}^{2}E\left[\left(\sum_{j\ne i}P_{ij}\eta_{j}\right)^{2}\right]+\sum_{i}E\left[\nu_{i}^{2}\right]E\left[\left(\sum_{j\ne i}P_{ij}\eta_{j}\right)^{2}\right]\\
 & \quad+\sum_{i}\sum_{j\ne i}P_{ij}^{2}\left(R_{i}R_{\Delta i}+E\left[\eta_{i}\nu_{i}\right]\right)\left(R_{j}R_{\Delta j}+E\left[\eta_{j}\nu_{j}\right]\right)\\
 & =\sum_{i}M_{ii}^{2}R_{i}^{2}\left(R_{\Delta i}^{2}+E\left[\nu_{i}^{2}\right]\right)+\sum_{i}\sum_{j\ne i}P_{ij}^{2}E\left[\eta_{j}^{2}\left(R_{\Delta i}^{2}+\nu_{i}^{2}\right)\right] \\
&\quad+\sum_{i}\sum_{j\ne i}P_{ij}^{2}\left(R_{i}R_{\Delta i}+E\left[\eta_{i}\nu_{i}\right]\right)\left(R_{j}R_{\Delta j}+E\left[\eta_{j}\nu_{j}\right]\right)\\
 & =\sum_{i}M_{ii}^{2}R_{i}^{2}R_{\Delta i}^{2}+\sum_{i}M_{ii}^{2}R_{i}^{2}E\left[\nu_{i}^{2}\right] +\sum_{i}\sum_{j\ne i}P_{ij}^{2}E\left[\nu_{i}^{2}\right]E\left[\eta_{j}^{2}\right]+\sum_{i}\sum_{j\ne i}P_{ij}^{2}R_{\Delta i}^{2}E\left[\eta_{j}^{2}\right]\\
 & \quad+\sum_{i}\sum_{j\ne i}P_{ij}^{2}\left(R_{i}R_{\Delta i}R_{j}R_{\Delta j}+E\left[\eta_{i}\nu_{i}\right]R_{j}R_{\Delta j}+R_{i}R_{\Delta i}E\left[\eta_{j}\nu_{j}\right]+E\left[\eta_{i}\nu_{i}\right]E\left[\eta_{j}\nu_{j}\right]\right)
\end{align*}
\small
\end{proof}

\end{document}

%% file: fig_v8/simf_example_res.tex
\begin{tabular}{ll|rrrrrrrrr}
\toprule
\multicolumn{2}{c|}{Designs} & \multicolumn{9}{c}{Procedures}\\
$E[T_{AR}]$ & $E[T_{FS}]$ & TSLS & EK & MS & MO & $\tilde{X}$-t & $\tilde{X}$-AR& L3O & LMorc & ARorc\\
\midrule
&$2 \sqrt{K}$ & 0.428 & 0.066 & NaN & 0.027 & 0.041 & 0.041 & 0.060 & 0.048 & 1.000\\
$ 2 \sqrt{K}$& 2 & 0.978 & 0.040 & NaN & 0.282 & 0.114 & 0.303 & 0.052 & 0.049 & 1.000\\
&0 & 0.983 & 0.025 & NaN & 0.260 & 0.054 & 0.282 & 0.047 & 0.055 & 1.000\\
\midrule
&$2 \sqrt{K}$ & 0.984 & 0.076 & 1.000 & 0.039 & 0.046 & 0.049 & 0.044 & 0.050 & 1.000\\
$ 2$& 2 & 1.000 & 0.096 & 1.000 & 0.085 & 0.155 & 0.149 & 0.048 & 0.049 & 1.000\\
&0 & 1.000 & 0.128 & 1.000 & 0.103 & 0.225 & 0.177 & 0.060 & 0.051 & 1.000\\
\midrule
&$2 \sqrt{K}$ & 0.994 & 0.097 & 0.064 & 0.064 & 0.071 & 0.067 & 0.059 & 0.055 & 0.057\\
$0$& 2& 1.000 & 0.231 & 0.059 & 0.047 & 0.179 & 0.106 & 0.049 & 0.051 & 0.055\\
&0 & 1.000 & 0.359 & 0.063 & 0.041 & 0.350 & 0.107 & 0.048 & 0.046 & 0.059\\
\bottomrule
\end{tabular}

%% file: fig_v8/simf_example_fixedK_res.tex
\begin{tabular}{ll|rrrrrrrrr}
\toprule
\multicolumn{2}{c|}{Designs} & \multicolumn{9}{c}{Procedures}\\
$E[T_{AR}]$ & $E[T_{FS}]$ & TSLS & EK & MS & MO & $\tilde{X}$-t & $\tilde{X}$-AR& L3O & LMorc & ARorc\\
\midrule
&.5c  & 0.210 & 0.049 & 1.000 & 0.212 & 0.217 & 0.217 & 0.048 & 0.052 & 1.000\\
.5c & 2 & 0.642 & 0.018 & 1.000 & 0.806 & 0.269 & 0.816 & 0.043 & 0.049 & 1.000\\
&0 & 0.498 & 0.005 & 1.000 & 0.881 & 0.330 & 0.893 & 0.062 & 0.051 & 1.000\\
\midrule
&.5c & 0.075 & 0.064 & 1.000 & 0.075 & 0.073 & 0.077 & 0.063 & 0.061 & 0.913\\
2&2 & 0.462 & 0.013 & 0.999 & 0.436 & 0.296 & 0.516 & 0.095 & 0.049 & 0.931\\
&0 & 0.440 & 0.008 & 1.000 & 0.448 & 0.337 & 0.576 & 0.088 & 0.052 & 0.934\\
\midrule
&.5c& 0.052 & 0.048 & 0.061 & 0.045 & 0.050 & 0.046 & 0.046 & 0.048 & 0.069\\
0&2 & 0.376 & 0.088 & 0.075 & 0.044 & 0.238 & 0.123 & 0.101 & 0.045 & 0.071\\
&0 & 0.590 & 0.181 & 0.076 & 0.029 & 0.431 & 0.163 & 0.075 & 0.045 & 0.080\\
\bottomrule
\end{tabular}

%% file: fig_v8/simf_example_power_res.tex
\begin{tabular}{ll|rrrrrrrrr}
\toprule
\multicolumn{2}{c|}{Designs} & \multicolumn{9}{c}{Procedures}\\
$E[T_{AR}]$ & $E[T_{FS}]$ & TSLS & EK & MS & MO & $\tilde{X}$-t & $\tilde{X}$-AR& L3O & LMorc & ARorc\\
\midrule
&$2 \sqrt{K}$ & 1.000 & 1.000 & NaN & 1.000 & 1.000 & 1.000 & 1.000 & 1.000 & 1.000\\
$2 \sqrt{K}$&$2$ & 0.310 & 0.173 & NaN & 0.539 & 0.117 & 0.563 & 0.240 & 0.225 & 1.000\\
&$0$ & 0.722 & 0.029 & NaN & 0.291 & 0.055 & 0.309 & 0.048 & 0.055 & 1.000\\
\midrule
&$2 \sqrt{K}$ & 1.000 & 1.000 & 1.000 & 1.000 & 1.000 & 1.000 & 1.000 & 1.000 & 1.000\\
2&$2 $ & 0.244 & 0.543 & 1.000 & 0.886 & 0.163 & 0.907 & 0.789 & 0.823 & 1.000\\
&$ 0$ & 0.998 & 0.090 & 1.000 & 0.157 & 0.169 & 0.221 & 0.073 & 0.054 & 1.000\\
\midrule
&$2 \sqrt{K}$ & 1.000 & 1.000 & 1.000 & 1.000 & 1.000 & 1.000 & 1.000 & 1.000 & 1.000\\
0&$2$ & 0.259 & 0.662 & 1.000 & 0.967 & 0.230 & 0.978 & 0.936 & 0.961 & 1.000\\
&$0$ & 1.000 & 0.373 & 0.059 & 0.048 & 0.333 & 0.117 & 0.067 & 0.055 & 0.054\\
\bottomrule
\end{tabular}

%% file: fig/simf_contX_res.tex
\begin{tabular}{lrrrrrrrr}
\toprule
  & TSLS & EK & ARorc & MO  & $\tilde{X}$-t & $\tilde{X}$-AR & L3O & LMorc\\
\midrule
$C_H=C_S=3\sqrt{K},\sigma_{\varepsilon v}=0$ & 0.061 & 0.017 & 1.000 & 0.061 & 0.079 & 0.078 & 0.042 & 0.044\\
$C_H=2\sqrt{K},C_S=2\sqrt{K}$ & 0.952 & 0.022 & 1.000 & 0.073 & 0.087 & 0.084 & 0.058 & 0.055\\
$C_H=2\sqrt{K},C_S=2$& 1.000 & 0.009 & 1.000 & 0.096 & 0.076 & 0.127 & 0.053 & 0.050\\
$C_H=2\sqrt{K},C_S=0$ & 1.000 & 0.006 & 1.000 & 0.103 & 0.061 & 0.127 & 0.059 & 0.052\\
\addlinespace
$C_H=3,C_S=3\sqrt{K}$ & 0.986 & 0.033 & 0.109 & 0.057 & 0.062 & 0.064 & 0.056 & 0.047\\
$C_H=3,C_S=3$ & 1.000 & 0.036 & 0.168 & 0.055 & 0.078 & 0.087 & 0.055 & 0.047\\
$C_H=3,C_S=0$& 1.000 & 0.048 & 0.184 & 0.058 & 0.106 & 0.088 & 0.053 & 0.057\\
\addlinespace
$C_H=0,C_S=2\sqrt{K}$ & 1.000 & 0.089 & 0.049 & 0.063 & 0.083 & 0.080 & 0.061 & 0.058\\
$C_H=0,C_S=2$ & 1.000 & 0.207 & 0.045 & 0.054 & 0.243 & 0.135 & 0.057 & 0.045\\
$C_H=0,C_S=0$ & 1.000 & 0.337 & 0.051 & 0.042 & 0.413 & 0.127 & 0.045 & 0.048\\
$C_H=C_S=0,\sigma_{\varepsilon v}=1$ & 1.000 & 1.000 & 0.044 & 0.042 & 1.000 & 0.157 & 0.052 & 0.044\\
\bottomrule
\end{tabular}

%% file: fig/simf_contX_res_K40.tex
\begin{tabular}{lrrrrrrrr}
\toprule
  & TSLS & EK & ARorc & MO  & $\tilde{X}$-t & $\tilde{X}$-AR & L3O & LMorc\\
\midrule
$C_H=C_S=3\sqrt{K},\sigma_{\varepsilon v}=0$ & 0.072 & 0.022 & 0.525 & 0.051 & 0.074 & 0.068 & 0.039 & 0.055\\
$C_H=2\sqrt{K},C_S=2\sqrt{K}$ & 0.238 & 0.034 & 0.388 & 0.051 & 0.074 & 0.077 & 0.055 & 0.062\\
$C_H=2\sqrt{K},C_S=2$ & 0.547 & 0.033 & 0.475 & 0.083 & 0.096 & 0.133 & 0.077 & 0.053\\
$C_H=2\sqrt{K},C_S=0$ & 0.651 & 0.013 & 0.511 & 0.072 & 0.088 & 0.102 & 0.068 & 0.054\\
\addlinespace
$C_H=3,C_S=3\sqrt{K}$ & 0.213 & 0.025 & 0.109 & 0.048 & 0.057 & 0.063 & 0.055 & 0.046\\
$C_H=3,C_S=3$ & 0.658 & 0.032 & 0.129 & 0.045 & 0.074 & 0.063 & 0.064 & 0.055\\
$C_H=3,C_S=0$ & 0.849 & 0.049 & 0.127 & 0.063 & 0.109 & 0.103 & 0.087 & 0.057\\
\addlinespace
$C_H=0,C_S=2\sqrt{K}$ & 0.853 & 0.105 & 0.049 & 0.064 & 0.068 & 0.098 & 0.085 & 0.056\\
$C_H=0,C_S=2$ & 0.999 & 0.152 & 0.048 & 0.045 & 0.201 & 0.132 & 0.098 & 0.037\\
$C_H=0,C_S=0$ & 1.000 & 0.342 & 0.052 & 0.051 & 0.439 & 0.143 & 0.080 & 0.049\\
$C_H=C_S=0,\sigma_{\varepsilon v}=1$& 1.000 & 1.000 & 0.045 & 0.040 & 1.000 & 0.179 & 0.082 & 0.045\\
\bottomrule
\end{tabular}

%% file: fig/simf_binX_res.tex
\begin{tabular}{lrrrrrrrr}
\toprule
  & TSLS & EK & ARorc & MO  & $\tilde{X}$-t & $\tilde{X}$-AR & L3O & LMorc\\
\midrule
$C_H=C_S=3\sqrt{K},\sigma_{\varepsilon v}=0$ & 0.046 & 0.049 & 0.059 & 0.045 & 0.045 & 0.045 & 0.049 & 0.054\\
$C_H=2\sqrt{K},C_S=2\sqrt{K}$ & 0.097 & 0.047 & 0.177 & 0.037 & 0.038 & 0.041 & 0.051 & 0.052\\
$C_H=2\sqrt{K},C_S=2$ & 0.727 & 0.059 & 1.000 & 0.127 & 0.051 & 0.143 & 0.058 & 0.051\\
$C_H=2\sqrt{K},C_S=0$ & 0.891 & 0.037 & 1.000 & 0.204 & 0.067 & 0.247 & 0.059 & 0.045\\
\addlinespace
$C_H=3,C_S=3\sqrt{K}$& 0.092 & 0.060 & 0.051 & 0.055 & 0.057 & 0.056 & 0.055 & 0.047\\
$C_H=3,C_S=3$& 0.996 & 0.089 & 0.888 & 0.059 & 0.086 & 0.096 & 0.055 & 0.048\\
$C_H=3,C_S=0$ & 1.000 & 0.124 & 0.999 & 0.101 & 0.289 & 0.181 & 0.068 & 0.052\\
\addlinespace
$C_H=0,C_S=2\sqrt{K}$ & 0.408 & 0.058 & 0.055 & 0.043 & 0.046 & 0.046 & 0.045 & 0.041\\
$C_H=0,C_S=2$ & 1.000 & 0.212 & 0.052 & 0.061 & 0.188 & 0.108 & 0.078 & 0.057\\
$C_H=0,C_S=0$ & 1.000 & 0.654 & 0.046 & 0.034 & 0.750 & 0.149 & 0.069 & 0.039\\
$C_H=C_S=0,\sigma_{\varepsilon \varepsilon}=0$ & 1.000 & 1.000 & 0.053 & 0.057 & 1.000 & 0.173 & 0.076 & 0.053\\
\bottomrule
\end{tabular}

%% file: fig/simf_binXcov_res.tex
\begin{tabular}{lrrrrrrrr}
\toprule
  & TSLS & EK & ARorc & MO  & $\tilde{X}$-t & $\tilde{X}$-AR & L3O & LMorc\\
\midrule
$C_H=C_S=3\sqrt{K},\sigma_{\varepsilon v}=0$& 0.048 & 0.123 & 0.049 & 0.052 & 0.047 & 0.055 & 0.054 & 0.060\\
$C_H=2\sqrt{K},C_S=2\sqrt{K}$ & 0.072 & 0.111 & 0.052 & 0.044 & 0.041 & 0.046 & 0.050 & 0.053\\
$C_H=2\sqrt{K},C_S=2$ & 0.171 & 0.016 & 0.471 & 0.083 & 0.012 & 0.092 & 0.060 & 0.050\\
$C_H=2\sqrt{K},C_S=0$& 0.259 & 0.002 & 0.960 & 0.126 & 0.008 & 0.135 & 0.047 & 0.058\\
\addlinespace
$C_H=3,C_S=3\sqrt{K}$ & 0.065 & 0.132 & 0.048 & 0.053 & 0.056 & 0.054 & 0.060 & 0.049\\
$C_H=3,C_S=3$& 0.131 & 0.015 & 0.108 & 0.040 & 0.003 & 0.042 & 0.044 & 0.050\\
$C_H=3,C_S=0$ & 0.247 & 0.003 & 0.300 & 0.087 & 0.004 & 0.091 & 0.062 & 0.053\\
\addlinespace
$C_H=0,C_S=2\sqrt{K}$ & 0.084 & 0.099 & 0.054 & 0.041 & 0.036 & 0.043 & 0.048 & 0.050\\
$C_H=0,C_S=2$& 0.178 & 0.006 & 0.058 & 0.043 & 0.002 & 0.044 & 0.052 & 0.051\\
$C_H=0,C_S=0$ & 0.246 & 0.006 & 0.048 & 0.063 & 0.005 & 0.069 & 0.081 & 0.050\\
$C_H=C_S=0,\sigma_{\varepsilon \varepsilon}=0$ & 1.000 & 0.497 & 0.042 & 0.013 & 0.147 & 0.049 & 0.092 & 0.035\\
\bottomrule
\end{tabular}